%% file: main.tex
\documentclass[runningheads]{llncs}

\usepackage{xspace}
\usepackage{flushend}

\usepackage{enumitem}
\usepackage{xcolor}
\usepackage{graphicx,epstopdf}
\usepackage{hyperref}

\usepackage{tabularx,booktabs}
\usepackage{multirow}

\usepackage{subcaption}

\usepackage{algorithm}
\usepackage[noend]{algpseudocode}
\usepackage{amsmath}
\usepackage{amssymb}

\usepackage{thmtools}
\usepackage{thm-restate}

\makeatletter
\newcommand*{\rom}[1]{\expandafter\@slowromancap\romannumeral #1@}
\makeatother

\usepackage{setspace}
\usepackage{enumitem}

\usepackage{tikz}
\usetikzlibrary{arrows,automata,positioning,shapes}
\usetikzlibrary{decorations.pathreplacing,calc}

\input{./notation-defs.tex}

\usepackage{xspace}
\usepackage{xargs}
\usepackage{xifthen}
\usepackage{framed}

\newenvironment{tightcenter}
 {\parskip=0pt\par\nopagebreak\centering}
 {\par\noindent\ignorespacesafterend}

\usepackage{ctable}
\newlength{\RoundedBoxWidth}
\newsavebox{\GrayRoundedBox}
\newenvironment{GrayBox}[1]%
{\setlength{\RoundedBoxWidth}{\linewidth-4.5ex}
\def\boxheading{#1}
\begin{lrbox}{\GrayRoundedBox}
\begin{minipage}{\RoundedBoxWidth}%
}{%
\end{minipage}
\end{lrbox}%
\begin{tightcenter}%
\begin{tikzpicture}%
\node(Text)[draw=black!20,fill=white,rounded corners,%
inner sep=2ex,text width=\RoundedBoxWidth]%
{\usebox{\GrayRoundedBox}};
\coordinate(x) at (current bounding box.north west);
\node [draw=white,rectangle,inner sep=3pt,anchor=north west,fill=white]
at ($(x)+(10.5pt,.75em)$) {\boxheading};
\end{tikzpicture}
\end{tightcenter}\vspace{0pt}%
\ignorespacesafterend
}

\algnewcommand\And{\textbf{and}}

\usepackage{graphicx}
\begin{document}
\title{Faster Biclique Mining in Near-Bipartite Graphs\thanks{This work was supported by the Gordon \& Betty Moore Foundation's Data-Driven Discovery Initiative under Grant GBMF4560 to Blair D. Sullivan.}}
\author{Blair D. Sullivan \and
Andrew van der Poel \and
Trey Woodlief}
\institute{North Carolina State University, Raleigh NC, 27607, USA\\
\email{\{blair\_sullivan, ajvande4, adwoodli\}@ncsu.edu}}
\maketitle     
\begin{abstract}

Identifying dense bipartite subgraphs is a common graph data mining task. Many applications focus on the enumeration of all maximal bicliques (MBs), though sometimes the stricter variant of maximal induced bicliques (MIBs) is of interest.
Recent work of Kloster et al. introduced a MIB-enumeration approach designed for ``near-bipartite" graphs, where the runtime is parameterized by the size $k$ of an odd cycle transversal (OCT), a vertex set whose deletion results in a bipartite graph.
Their algorithm was shown to outperform the previously best known algorithm even when $k$ was logarithmic in $|V|$. In this paper, we  introduce two new algorithms optimized for near-bipartite graphs - one which enumerates MIBs in time $O(M_I |V| |E| k)$, and another based on the approach of Alexe et al. which enumerates MBs in time $O(M_B |V| |E| k)$, where $M_I$ and $M_B$ denote the number of MIBs and MBs in the graph, respectively.
 We implement all of our algorithms in open-source C++ code and experimentally verify that the OCT-based approaches are faster in practice than the previously existing algorithms on graphs with a wide variety of sizes, densities, and OCT decompositions.\looseness-1

\end{abstract}

\section{Introduction}

\input{./introduction.tex}

\section{Preliminaries}\label{sec:prelims}
\input{./prelims.tex}

\section{Algorithms}\label{sec:algs}
\input{./algs.tex}

\subsection{MIB Algorithm Framework}\label{sec:mibfk}
\input{./MIB_frwk.tex}

\subsection{\nonlex}\label{sec:nonlex}
\input{./enummibs.tex}

\subsection{\iOCTMIB}\label{sec:ioctmib}
\input{./ioctmib.tex}

\subsection{\OCTMICA}\label{sec:octmica}
\input{./octmica.tex}

\section{Implementation} \label{sec:implementation}
\input{./implementation.tex}

\section{Experiments} \label{sec:experiments}
\input{./experiments.tex}

\section{Conclusion}
\input{./conclusion.tex}

%
% ---- Bibliography ----

\input{./bibliography.tex}
\appendix
\part*{Appendices}

\section{MIB-Enumeration Framework Subroutines}\label{app:mib}
\input{./app-MIB.tex}

\section{MB-Enumeration Framework Subroutines}\label{app:mb}
\input{./app-MB.tex}

\section{Additional Enumeration Experiments}\label{app:plots}
\input{./app-exp.tex}

\end{document}

%% file: notation-defs.tex
\newcommand{\biportion}{\ensuremath{G[L\cup R]}\xspace}

\newcommand{\MBs}{\mathcal{M_B}}
\newcommand{\numMBs}{M_B}

\newcommand{\OCTMIB}{\texttt{OCT-MIB}\xspace}
\newcommand{\iOCTMIB}{\texttt{OCT-MIB-\rom{2}}\xspace}
\newcommand{\imbea}{\texttt{iMBEA}\xspace}
\newcommand{\nonlex}{\texttt{Enum-MIB}\xspace}

\newcommand{\MICA}{\texttt{MICA}\xspace}
\newcommand{\OCTMICA}{\texttt{OCT-MICA}\xspace}

\newcommand{\moddias}{\texttt{LexMIB}\xspace}

\newcommand{\MM}{\texttt{MakeInd\-Maximal}\xspace}

\newcommand{\MMNI}{\texttt{MakeMaximal}\xspace}
\newcommand{\consensus}{\texttt{Consensus}\xspace}

\newcommand{\AT}{\texttt{AddTo}\xspace}

\newcommand{\AMBpfull}{\textsc{Maximal Induced Biclique Enumeration}\xspace}

\newcommand{\MBpfull}{\textsc{Maximal Biclique Enumeration}\xspace}

\newcommand{\algoref}[1]{\hyperref[#1]{Algorithm~\ref*{#1}}}

%% file: introduction.tex
Bicliques (complete bipartite graphs) naturally arise in many data mining applications, including detecting cyber communities~\cite{KUMAR}, data compression~\cite{AGARWAL}, epidemiology~\cite{MUSHLIN}, artificial intelligence~\cite{WILLE}, and gene co-expression analysis~\cite{KAYTOUE,KAYTOUETWO}. In many
settings, the bicliques of interest are \emph{maximal} (not contained in any larger biclique) and/or
\emph{induced} (each side of the bipartition is independent in the host graph), and there
is a large body of literature giving algorithms for enumerating all such subgraphs~\cite{ALEXE,DIAS,EPPSTEIN,LI,MAKINO,MUSHLIN,SANDERSON,ZHANG}. Many of these approaches make strong structural assumptions on the host graph; the case when the host graph is bipartite has been particularly
well-studied, and the \imbea algorithm of Zhang et al. has been empirically established to be state-of-the-art~\cite{ZHANG}. Until recently, the only known non-trivial algorithm for enumerating maximal induced bicliques (MIBs) in general graphs was that of Dias et al. which did so in lexicographic order~\cite{DIAS}. In~\cite{KLOSTER}, Kloster et al. presented a new algorithm for enumerating MIBs in general graphs, \OCTMIB, which extended ideas from \imbea to work on non-bipartite graphs by using an \emph{odd cycle transversal} (OCT set): a set of nodes $O$ such that $G[V\setminus O]$ is bipartite.
This yielded an algorithm with runtime $O(M_I n m n_O^2 3^{n_O/3})$ where $n_O = |O|$, $M_I$ is the number of MIBs in $G=(V,E)$, and $n$ and $m$ denote $|V|$ and $|E|$, respectively. The $3^{n_O/3}$ term arises from \OCTMIB's dependence on the number of maximal independent sets (MISs) in $O$.
In this paper, we give new algorithms for enumerating both MIBs and maximal, not necessarily induced bicliques (MBs) in general graphs.
We first present \iOCTMIB which again leverages odd cycle transversals to enumerate MIBs in time $O(M_I n m n_O)$. In contrast to \OCTMIB, the worst-case runtime of \iOCTMIB is not dependent on the number of MISs in $O$, making it better than \OCTMIB when $n_O \in \omega(1)$. We also give a second algorithm for MIB-enumeration, \nonlex, which has runtime $O(M_I n m)$. \nonlex is essentially a modified version of the algorithm of Dias et al.~\cite{DIAS}, which achieves a faster runtime by dropping the lexicographic output requirement.

In the setting considering non-induced bicliques, the state-of-the-art approach is \MICA of Alexe et al.~\cite{ALEXE}. \MICA employs a consensus mechanism to iteratively find maximal bicliques by combining them together, resulting in an $O(M_B n^3)$ algorithm, where $M_B$ is the number of MBs. We introduce a new algorithm \OCTMICA which leverages odd cycle transversals and runs in $O(M_B(n^2 n_O + m n))$ time.

Since all graphs have OCT sets (although they can be size $O(n)$, as in cliques),
\OCTMIB, \iOCTMIB, and \OCTMICA can all be run in the general case; their correctness does not require minimality or
optimality of the OCT set.
Further, we implement \iOCTMIB, \nonlex and \OCTMICA in open source C++ code, and evaluate their performance on a suite of synthetic graphs with known OCT decompositions. Our experiments show that \OCTMICA and \iOCTMIB are the dominant algorithms for their respective problems in many settings. Their efficiencies allow us to run on larger graphs than in~\cite{KLOSTER}.

We begin with preliminaries and a brief discussion of related work in Section~\ref{sec:prelims}, then describe each of our three new algorithms and provide proofs of their correctness and runtimes in Section~\ref{sec:algs}. We highlight several implementation details in Section~\ref{sec:implementation}, before presenting our experimental evaluation in Section~\ref{sec:experiments}.

%% file: prelims.tex
\subsection{Related work} \label{sec:related-work}

The complexity of finding bicliques is well-studied, beginning with the results of Garey and Johnson~\cite{GAREY} which establish that in bipartite graphs, finding the largest balanced biclique is NP-hard but the largest biclique can be found in polynomial time.
Particularly relevant to the mining setting, Kuznetsov showed that enumerating MBs in a bipartite graph is \#P-complete~\cite{KUZNETSOV}. Finding the biclique with the largest number of edges was shown to be NP-complete in general graphs~\cite{YANNAKAKIS}, but the case of bipartite graphs remained open for many years. Several variants (including the weighted version) were proven NP-complete in~\cite{DAWANDE} and in 2000, Peeters finally resolved the problem, proving the edge maximization variant is NP-complete in bipartite graphs~\cite{PEETERS}.

For the problem of enumerating MIBs, the best known algorithm in general graphs is due to Dias et al.~\cite{DIAS}; in the non-induced setting, approaches include a consensus algorithm \MICA~\cite{ALEXE}, an efficient algorithm for small arboricity~\cite{EPPSTEIN}, and a general framework for enumerating maximal cliques and bicliques~\cite{GELY}, with \MICA the most efficient among them, running in $O(M_B n^3)$.
We note that, as described, the method in~\cite{DIAS} may fail to enumerate all MIBs; a modified, correct version was given in~\cite{KLOSTER}.

There has also been significant work on enumerating MIBs in bipartite graphs.
We note that since all bicliques in a bipartite graph are necessarily induced, non-induced solvers for general graphs (such as \MICA) can be applied, and have been quite competitive.
The best known algorithm however, is due to Zhang et al.~\cite{ZHANG} and directly exploits the bipartite structure.
Other approaches in bipartite graphs include frequent closed itemset mining~\cite{LI} and transformations to the maximal clique problem~\cite{MAKINO}; faster algorithms are known when a lower bound on the size of bicliques to be enumerated is assumed~\cite{MUSHLIN,SANDERSON}.

Kloster et al.~\cite{KLOSTER} extended techniques for bipartite graphs to the general setting using
odd cycle transversals, a form of ``near-bipartiteness'' which arises naturally in many applications~\cite{GULPINAR,PANCONESI,SCHROOK}. This work resulted in \OCTMIB, an algorithm for enumerating MIBs in a general graph, parameterized by the size of a given OCT set.
Although finding a minimum size
OCT set is NP-hard, the problem of deciding if an OCT set with size $k$ exists is fixed parameter tractable (FPT) with algorithms in~\cite{LOKSHTANOV} and~\cite{IWATA} running in times $O(3^k k m n)$ and $O(4^k n)$, respectively. We note non-optimal OCT sets only affect the runtime (not correctness) of our algorithms,
allowing us to use heuristic solutions. Recent implementations~\cite{GOODRICH} of a heuristic ensemble alongside algorithms from ~\cite{AKIBA,HUFFNER} alleviate concerns about finding an OCT decomposition creating a barrier to usability.
\subsection{Notation and terminology}

Let $G = (V,E)$ be a graph; we set $n = |V|$ and $m = |E|$. We define $N(v)$ to be the neighborhood of $v \in V$ and write $\overline{N}(v)$ for $v$'s non-neighbors.
An independent set $T \subseteq V(G)$ is a \emph{maximal independent set} (MIS) if $T$ is not contained in any other independent set of $G$.  Unless otherwise noted, we assume without loss of generality that $G$ is connected.

A biclique $A \times B$ in a graph $G = (V,E)$ consists of non-empty disjoint sets $A, B \subset V$ such that every vertex of $A$ is neighbors with every vertex of $B$.
We say a biclique $A \times B$ is \emph{induced} if both $A$ and $B$ are independent sets in $G$.
A \emph{maximal biclique} (MB) in $G$ is a biclique not properly contained in any other; a \emph{maximal induced biclique} (MIB) is analogous among induced bicliques. We use $M_B$ and $M_I$ to denote the number of MBs and MIBs in $G$, respectively.
If $O$ is an OCT set in $G$, we denote the corresponding OCT decomposition of $G$ by $G[L,R,O]$, where the induced subgraph $G[L \cup R]$ is bipartite. We write
$n_L, n_R,$ and $n_O$ for $|L|, |R|,$ and $|O|$, respectively.

%% file: algs.tex
In this section we provide three novel algorithms, two of which of solve \AMBpfull (\nonlex and \iOCTMIB) and the other of which solves \MBpfull (\OCTMICA). Both \nonlex and \iOCTMIB follow the same general framework, which we now describe.

%% file: MIB_frwk.tex
The MIB-enumeration algorithms both use two subroutines, \MM and \AT.
\MM takes in $(C, S)$, where $C$ is an induced biclique and $S \subseteq V$, and either returns a MIB $C^+$ where $C \subseteq C^+$, $C^+ \subseteq C \cup S$, $C \neq \emptyset$, or returns $\emptyset$. If it returns $\emptyset$ and $C \neq \emptyset$ then there is another MIB $D$ which contains $C$ and $v \in (V \setminus S) \setminus C$.
\AT takes in $(C, v)$ where $C=C_1 \times C_2$ is an induced biclique and $v \in V \setminus (C_1 \cup C_2)$, and returns the induced biclique where $v$ is added to $C_1$, $N(v)$ is removed from $C_1$, and $\overline{N}(v)$ is removed from $C_2$ if $C_2 \cap {N}(v) \neq \emptyset$; otherwise,  $\emptyset$ is returned. Both \MM and \AT operate in $O(m)$ time.
We defer algorithmic details and proofs of the complexity and correctness for these routines to the Appendix.

The MIB-enumeration framework (shown in Algorithm~\ref{pc:MIB-frwk}) begins by finding a seed set of MIBs $\mathcal{C}_S$. At a high level, it operates by attempting to add vertices from the designated set $I_S$ to previously found MIBs to make them maximal. We utilize a dictionary $\mathcal{D}$ to track which MIBs have already been found and a queue $\mathcal{Q}$ to store bicliques which have not yet been explored.  We now prove two technical lemmas used to show the correctness of this framework.

\begin{algorithm}[t]
\begin{algorithmic}[1]
	\State Input: {$G=(V,E)$, $I_S$}

	\State $\mathcal{C}_S =$ \texttt{FindSeedSet}$(G)$	\Comment{set of initial MIBs}
	\State Add each $C \in \mathcal{C}_S$ to $\mathcal{D}$ and $\mathcal{Q}$
	\While{$\mathcal{Q}$ is not empty}
		\State $X \times Y \leftarrow ~$\texttt{pop}$(Q)$\label{alg1:xy}
		\For{$j \in I_S \setminus (X \cup Y)$}\label{alg1:forj}
			\State $C_1 = \AT(X \times Y, j)$
			\State $C_1' = \MM(C_1, I_S)$
			\If {$C_1'$ is not in $\mathcal{D}$}
				\State Add $C_1'$ to $\mathcal{D}$ and $\mathcal{Q}$
			\EndIf
			\State $C_2 = \AT(Y \times X, j)$
			\State $C_2' = \MM(C_2, I_S)$
			\If {$C_2'$ is not in $\mathcal{D}$}
				\State Add $C_2'$ to $\mathcal{D}$ and $\mathcal{Q}$
			\EndIf
		\EndFor
	\EndWhile
	\State\Return{$\mathcal{D}$}
\end{algorithmic}
\caption{MIB-enumeration algorithm framework} \label{pc:MIB-frwk}
\end{algorithm}

\begin{lemma}\label{lem1}
Let $X \times Y$ be a MIB in graph $G$ which contains a non-empty subset of $R \times S$, another MIB in $G$. Running {\AT} with parameters $X \times Y$ and $v \in R \setminus (X \cup Y)$ returns a biclique which contains $R \cap X$, $S \cap Y$, and $v$ if $Y \cap N(v) \neq \emptyset$.
\end{lemma}

\begin{proof}
By construction, $v$ must be independent from $R$ and completely connected to $S$. Thus, none of $R \cap X$ will be removed from $X$ and all of $S \cap Y$ will remain in $Y$, as required. Therefore, as long as $Y \cap {N}(v) \neq \emptyset$, the desired biclique is returned.
\end{proof}

\begin{lemma}\label{lem2}
In Algorithm~\ref{pc:MIB-frwk}, if there exists a MIB $A' \times B'$ in $\mathcal{D}$ such that $A \setminus I_S \subseteq A'$, $B \setminus I_S \subseteq B'$ and $(A\cup B) \cap (A' \cup B') \neq \emptyset$, for each MIB $A \times B$ in $G$, then all MIBs in $G$ are included in $\mathcal{D}$.
\end{lemma}

\begin{proof}
Assume not.  Let $A \times B$ be a MIB in $G$ which is not in $\mathcal{D}$ with $|(A \cup B) \setminus I_S|$ maximum. Let $A' \times B'$ be the MIB in $\mathcal{D}$ such that $A \setminus I_S \subseteq A'$, $B \setminus I_S \subseteq B'$ and $(A\cup B) \cap (A' \cup B') \neq \emptyset$ and let $v \in ((A \cup B) \setminus (A' \cup B')) \subseteq I_S$. Without loss of generality assume $B \cap B' \neq \emptyset$ and $v \in A$.

Consider the iteration of Algorithm~\ref{pc:MIB-frwk} when $X \times Y = A' \times B'$ and $j = v$ (lines~\ref{alg1:xy}-\ref{alg1:forj}).  By Lemma~\ref{lem1}, one of the calls to \AT returns an induced biclique $C$ which contains $A \setminus I_S$, $B \setminus I_S$, and $v$. Both sides of $C$ are non-empty (since $B \cap B' \neq \emptyset$ and $v \in A$).
If $C = A \times B$ we obtain a contradiction, as \MM$(C, I_S)$ would return $C$, resulting in its addition to $\mathcal{D}$. Otherwise, either \MM returns $\emptyset$ or a biclique $C' = A' \times B'$ which is added to $\mathcal{D}$. Since both sides of $C$ are nonempty, if \MM returns $\emptyset$, there exists a MIB in $G$ containing $C$ and $x \in (V \setminus I_S) \setminus C$ . Let $A' \times B'$ be such a MIB; since it has more vertices in $V \setminus I_S$ than $C$, it must be in $\mathcal{D}$, and we set $C' = A' \times B'$. In either case, $C \subseteq (A' \cup B')$, $|(A \cup B) \setminus (A' \cup B')| < |(A \cup B) \setminus (X \cup Y)|$. We can repeat this argument for the new $A' \times B'$, noting that $(A \cup B) \cap (A' \cup B')$ will include vertices on both sides. Thus, the argument still holds without any assumption on the non-empty side of the intersection and $|(A \cup B) \setminus (A' \cup B')|$ will strictly decrease; when it reaches 0, $A' \times B'  = A \times B$, a contradiction.
\end{proof}

Note that as \MM only returns MIBs, this framework will only include MIBs in $\mathcal{D}$. Together with Lemma~\ref{lem2}, this yields the following corollary.\looseness-1

\begin{corollary}\label{cor1}
If for every MIB $A \times B \in G$ there is a MIB $A' \times B' \in \mathcal{C}_S$ such that $A \setminus I_S \subseteq A'$, $B \setminus I_S \subseteq B'$ and $(A\cup B) \cap (A' \cup B') \neq \emptyset$, then upon completion of Algorithm~\ref{pc:MIB-frwk}, $\mathcal{D}$ will contain exactly the MIBs in $G$.
\end{corollary}

Recall that \AT and \MM each run in $O(m)$ time. Combining this with the fact that each MIB in $G$ is popped at most once from $\mathcal{Q}$ we have:

\begin{corollary}\label{cor2}
The time complexity of this framework is $O(M_I m n + INIT)$, where $INIT$ is the time needed by \texttt{FindSeedSet} to compute $\mathcal{C}_S$.
\end{corollary}

%% file: enummibs.tex
We now present \nonlex, which follows the MIB-enumeration framework. To form $\mathcal{C}_S$, for each vertex $v \in V$ we run \MM$(\{v\} \times \{x\}, V)$ where $x \in N(v)$ and add it to $\mathcal{C}_S$. We also let $I_S = V$. To show the correctness of this approach, we note that $V \setminus V = \emptyset$ and any MIB contains the empty set. 
Thus all that remains to show is that for each MIB there is a MIB in $\mathcal{C}_S$ with which it has a non-empty intersection. As every $v\in V$ is in some MIB in $\mathcal{C}_S$, this condition is met. Thus, via Corollary~\ref{cor1}, \nonlex will find all MIBs.
There may be $O(n)$ duplicates in $\mathcal{C}_S$ which can be removed in $O(n)$ time per duplicate.
 As \MM runs in $O(m)$ time, by Corollary~\ref{cor2}, the time complexity of \nonlex is $O(M_I m n )$.
We note that \nonlex is essentially a simplified version of the \moddias algorithm from~\cite{KLOSTER} which does not guarantee lexicographic order on output.

%% file: ioctmib.tex
Next we describe \iOCTMIB, an algorithm for enumerating all MIBs in a graph with a given OCT decomposition $G[L,R,O]$. \iOCTMIB also makes use of the MIB-enumeration framework described in Section~\ref{sec:mibfk}. In the calls to
\MM we let $I_S = O$. To form $\mathcal{C}_S$, we begin by running \imbea~\cite{ZHANG} to find the set $\mathcal{C}_B$ of MIBs in $G[L \cup R]$. For each $C_B \in \mathcal{C}_B$ we run \MM on $(C_B, O)$. This creates a set $X_B$ of MIBs in $G$. 

Then for each node $o \in O$, we find the set of MISs in $N(o)$. This can be done in $O(mn)$ time per MIS using the algorithm of Tsukiyama et al.~\cite{TSUKIYAMA}. For each MIS $I_{o}$ found, run \MM on
the induced biclique $\{o\} \times I_{o}$. Let the multiset of all MIBs produced by this process be denoted $X_Q$. Note that a MIB may be in $X_Q$ up to $O(n_O)$ times (once per $o \in O$, stemming from an MIS in $N(o)$), but we can remove duplicates from $X_Q$ in $O(n)$ per MIB, forming $X_M'$. We then let $\mathcal{C}_S = X_B \cup X_M'$.
Thus, \texttt{FindSeedSet} runs in $O(mnn_O)$ per unique MIB found, and by Corollary~\ref{cor2}, the total time complexity of \iOCTMIB is $O(M_Imnn_O)$.

To show the correctness of \iOCTMIB, we must show that for every MIB in $G$, we include a MIB in $\mathcal{C}_S$ which includes all of its non-OCT nodes and a node in the MIB if the MIB is completely contained in $O$.
If an entire MIB $C$ is contained in $O$, then any MIB containing $\{o\} \times I_o$ for $o \in C$ suffices. If a MIB has non-OCT nodes on both sides, then there must be a MIB in $X_B$ which contains these non-OCT nodes because there is a MIB in $G[L \cup R]$ containing them. If a MIB has all of its non-OCT nodes on one side, then there is an OCT node $o$ which is neighbors with all of the non-OCT nodes, which thus must be contained in an MIS in $N(o)$. Thus, by Corollary~\ref{cor1}, we find all of the MIBs in $G$.\looseness-1

%% file: octmica.tex
\OCTMICA is an algorithm for enumerating the maximal bicliques (MBs) in a general graph with a given OCT decomposition $G[L,R,O]$. We adapt the approach of \MICA~\cite{ALEXE}, which relies on a
seed set of bicliques which ``cover'' the graph. Specifically, we restrict \MICA's coverage requirement for the
seed set to only the OCT set and leverage \imbea~\cite{ZHANG} to enumerate the MBs entirely within  \biportion.  This reduces the runtime from $O(n^3\numMBs)$ to $O(n^2n_O\numMBs)$.

\OCTMICA begins by running \imbea (line \ref{octmica:imbea} in Algorithm \ref{algOCTMICA}) to get $\MBs'$, the MBs in \biportion,  in time $O(nm'\numMBs')$, where $m'$ is the number of edges in \biportion and $\numMBs' = |\MBs'|$.
Using \MMNI, we convert elements of $\MBs'$ to be maximal with respect to $G$ (lines \ref{octmica:mmniloopimbea}-\ref{octmica:mmniimbea}). \MMNI runs in $O(m)$ time and its algorithmic details are deferred to the Appendix.
\OCTMICA then initializes its seed set of size $O(n_O)$ consisting of bicliques from the stars of the OCT set (lines \ref{octmica:octstarloop}-\ref{octmica:addstar}), and adds these to the working set $C$ of all identified MBs (line \ref{octmica:workingsetinit}). Similar to \MICA, the remainder of the algorithm builds new bicliques by combining (via \consensus, see Appendix) pairs of elements from the seed set $C_O$ and previously identified MBs $C$ (lines \ref{octmica:initfound}-\ref{octmica:insert}), until no new bicliques are generated.
This runs in time $O(n^2n_O\numMBs)$.

\begin{algorithm}[t!]
\caption{\OCTMICA}
\label{algOCTMICA}
\begin{algorithmic}[1]
\Procedure{Enumerate}{$G = (L,R,O)$} \label{octmica:fncall}
	\State$\MBs'=\Call{BipartiteSolve}{L,R}$\Comment{Implementation of \imbea, $O(m'n\numMBs')$} \label{octmica:imbea}
	\For{$B\in\MBs'$}\Comment{$O(\numMBs')$} \label{octmica:mmniloopimbea}
		\State $B=\Call{\MMNI}{B}$\Comment{Extend in place, $O(m)$} \label{octmica:mmniimbea}
	\EndFor
	\State$C_0=\{\}$ \label{octmica:initC_0}
	\For{$v$ in $O$} \Comment{Initialize Bicliques from stars, $O(n_O)$} \label{octmica:octstarloop}
		\State$B=\Call{\MMNI}{v\times N(v)}$\Comment{$O(m)$} \label{octmica:mmnistar}
		\State$C_0$\Call{.add}{$B$} \label{octmica:addstar}
	\EndFor
	\State$C=\MBs'\cup C_0$. \label{octmica:workingsetinit}%
	\State\Call{sort}{$C$}\Comment{$O(\numMBs'\log(\numMBs'))$} \label{octmica:initsort}%
	\State$found=true$ \label{octmica:initfound}
	\While{$found$} \label{octmica:foundloop}
		\State$found=false$ \label{octmica:foundfalse}
		\For{$B_1$ in $C_0$} \Comment{$O(n_O)$} \label{octmica:loopC_0}
			\For{$B_2$ in $C$} \Comment{$O(\numMBs)$} \label{octmica:loopC}
				\For{$B_3$ in \Call{\consensus}{$B_1, B_2$}} \label{octmica:loopconsensus}
					\State$B_4=\Call{\MMNI}{B_3}$\Comment{$O(m)$} \label{octmica:loopmmni}
					\If{$B_4$ not in $C$}\Comment{$O(n\log(\numMBs))$} \label{octmica:absorbcheck}
						\State$found=true$ \label{octmica:foundtrue}
						\State$C$\Call{.InsertInSortedOrder}{$B_4$} \label{octmica:insert}
					\EndIf
				\EndFor
			\EndFor
		\EndFor
	\EndWhile
	\State\Return{$C$} \label{octmica:return}
\EndProcedure
\end{algorithmic}
\end{algorithm}

\begin{lemma}
\OCTMICA returns exactly $\MBs$, the set of maximal bicliques in $G$.
\end{lemma}
\begin{proof}
Running \imbea and \MMNI ensures all maximal bicliques from \biportion were found and added to $C$.
Thus, we restrict our attention to maximal bicliques with at least one node from $O$, and proceed
similarly to the proof of Theorem 3 in \cite{ALEXE}. We say that a biclique $B_1 = X_1 \times Y_1$ \emph{absorbs} a biclique $B_2 = X_2 \times Y_2$ if $X_2\subseteq X_1$ and $Y_2\subseteq Y_1$ or $Y_2\subseteq X_1$ and $X_2\subseteq Y_1$.

We show that every biclique $B^* = X^* \times Y^*$ in $G$ is absorbed by some biclique in $C$ by induction on $k$, the number of OCT vertices in $B^*$. In the base case ($k=0$), $B^*$ is contained in \biportion and is absorbed by a biclique in $\MBs'\subseteq C$. We now consider $k \geq 1$; without loss of generality, assume $X^*$ contains some OCT vertex $v$. Then $B' = \{v\}\times Y^*$ is
absorbed by some biclique $B_1= X_{1} \times Y_{1}, v \in X_{1}, Y^* \subseteq Y_{1}$, where $B_1\in C_0$ is formed from the star centered on $v$. Further, $B''=(X^* \setminus \{v\}) \times Y^*$ has fewer vertices from OCT than $B^*$, so by induction it is absorbed by some biclique $B_2=X_{2} \times Y_{2}, (X \setminus \{v\})\subseteq X_{2}, Y^*\subseteq Y_{2}$, where $B_2\in C$. Now $B^*$ is a consensus of $B'$ and $B''$, and will be absorbed by the corresponding consensus of $B_1$ and $B_2$, guaranteeing absorption by a biclique in $C$.\looseness-1
\end{proof}

\begin{lemma}
The runtime of \OCTMICA after \imbea is $O(n^2n_O\numMBs)$.
\end{lemma}
\begin{proof}
	We begin by noting that $\numMBs \leq 2^n$, so $\log(\numMBs)$ is $O(n)$.

	Finding the bicliques in $\MBs'$ requires time $O(m'n\numMBs')$ for \imbea (line \ref{octmica:imbea}); making them maximal (lines \ref{octmica:mmniloopimbea}-\ref{octmica:mmniimbea}) is $O(m\numMBs')$. The bicliques generated by the OCT stars (lines \ref{octmica:octstarloop}-\ref{octmica:addstar}) can be found in $O(mn_O)$. Sorting the initial set $C$ (line \ref{octmica:initsort}) incurs an additional $O(\numMBs'\log(\numMBs'))$. Since $\log(\numMBs)$ is $O(n)$, the total runtime for our initialization (lines \ref{octmica:imbea}-\ref{octmica:initsort}) is $O(mn\numMBs'+mn_O)$.

 	The consensus-building stage of \OCTMICA contains nested loops over $C_0$ (line \ref{octmica:loopC_0}) and $C$ (line \ref{octmica:loopC}), which execute at most $O(n_O)$ and  $O(\numMBs)$ times, respectively. The \consensus operation (line \ref{octmica:loopconsensus}) executes in $O(n)$, and produces a constant number of candidate bicliques to check. Each execution of the inner loop incurs a cost of
	$O(m)$ for \MMNI (line \ref{octmica:loopmmni}) and $O(n\log(\numMBs))$ to insert the new MB in sorted order (lines \ref{octmica:absorbcheck}-\ref{octmica:insert}). We note that the runtime of \consensus is dominated by the cost of the loop.
 	Thus, the total runtime of consensus-building is {$O(n_O \numMBs n \log(\numMBs))$}, or $O(n^2n_O\numMBs)$.
\end{proof}

This analysis leads to an overall runtime of $O(m'n\numMBs'+n^2n_O\numMBs)$, as desired. We note that for $n_O\in\Theta(n)$, \OCTMICA's runtime degenerates to the $O(n^3\numMBs)$ of \MICA. Additionally, the stronger results for incremental polynomial time described for \MICA in \cite{ALEXE}  still apply; the proofs are similar and are omitted for space. For bipartite graphs ($n_O = 0$), \OCTMICA is effectively \imbea, which was empirically shown to be more efficient than \MICA on bipartite graphs~\cite{ZHANG}.

%% file: implementation.tex
In this section we describe several relevant implementation details and design decisions.

\subsection{Algorithm Framework}

We always (re-label and) store vertices as $\{0, 1, \ldots n\}$ and maintain internal dictionaries as needed to recover original labels -- e.g. when taking subgraphs. This allows us to leverage native data types and structures; vertices are stored as \texttt{size\_t}.

For efficiency in subroutines, we utilize two representations of $G$. One representation is as adjacency lists, stored as sorted vectors (to improve union and intersection relative to dictionaries or unsorted vectors). This representation is essential in the performance of \consensus in \MICA / \OCTMICA and \MM and \AT in \OCTMIB or \iOCTMIB. We also store the graph as a dictionary of dictionaries which is more amenable to taking subgraphs (as when finding MISs in \OCTMIB, \iOCTMIB). Deleting a node requires time $O(N(v))$ as compared to $O(N(v)\Delta(G))$, where $\Delta(G)$ is the maximum degree, in the adjacency list representation.

\subsection{\MICA}
The public implementation of \MICA used in~\cite{ZHANG} is available at \cite{IASTATE}. However, this implementation is only suitable for bipartite graphs as it makes certain efficiency improvements in storage, etc. which assume bipartite input. As such, we implemented \MICA from scratch in the same framework as \OCTMIB and \OCTMICA, etc., using the data structures discussed above. This is incompatible with the technique
described in~\cite{ALEXE} for storing only one side of each biclique (since in the non-induced case, maximality completely determines the other side). We note this could improve efficiency of both \MICA and \OCTMICA in a future version of our software, and should not significantly affect their relative performance as analyzed in this work.

%% file: experiments.tex
\subsection{Data and experimental setup}

\begin{figure*}[t!]
    \includegraphics [width=0.5\textwidth,trim=0cm .25cm .25cm .25cm,clip]{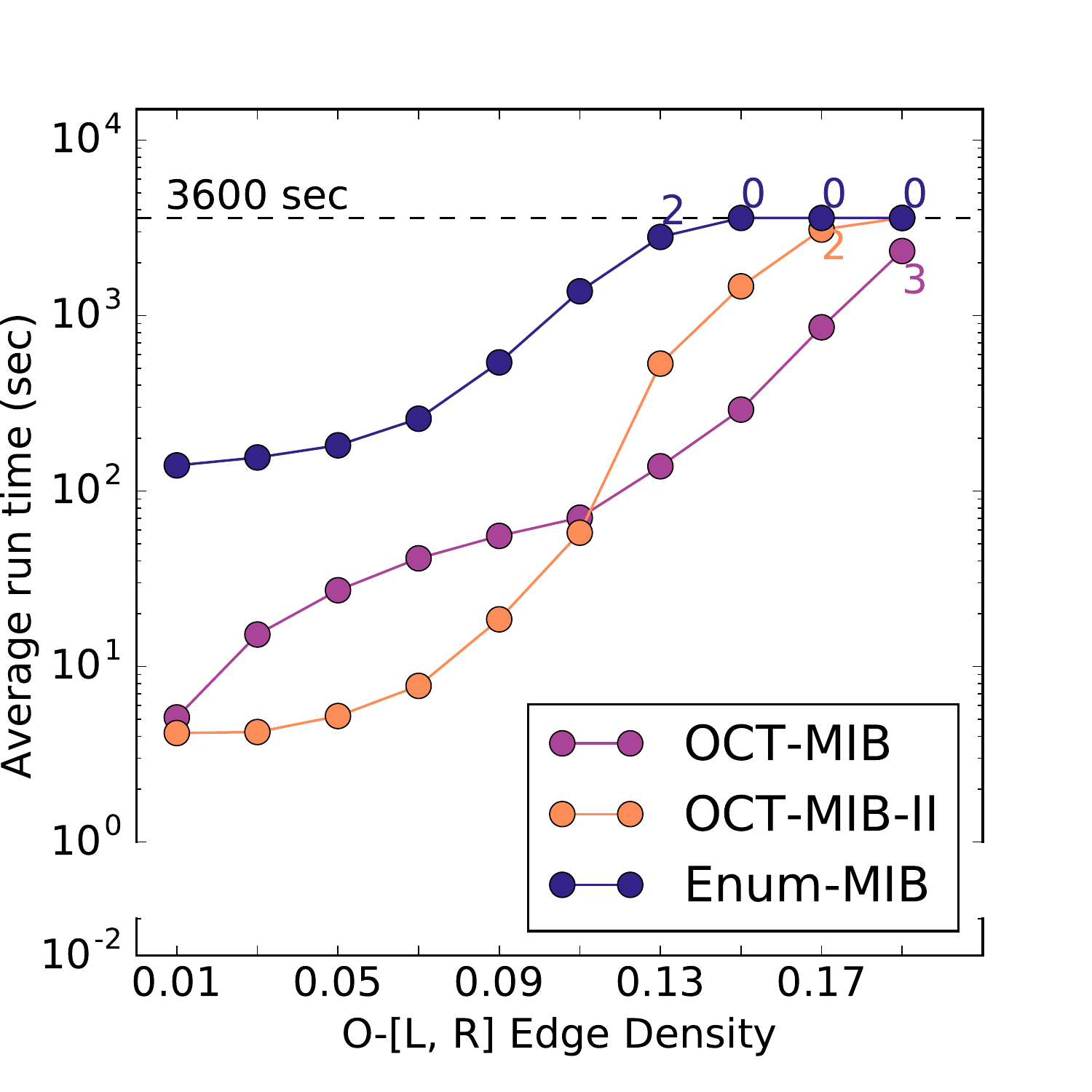}%
\hspace{0.01cm}
    \includegraphics [width=0.5\textwidth,trim=0cm .25cm .25cm .25cm,clip]{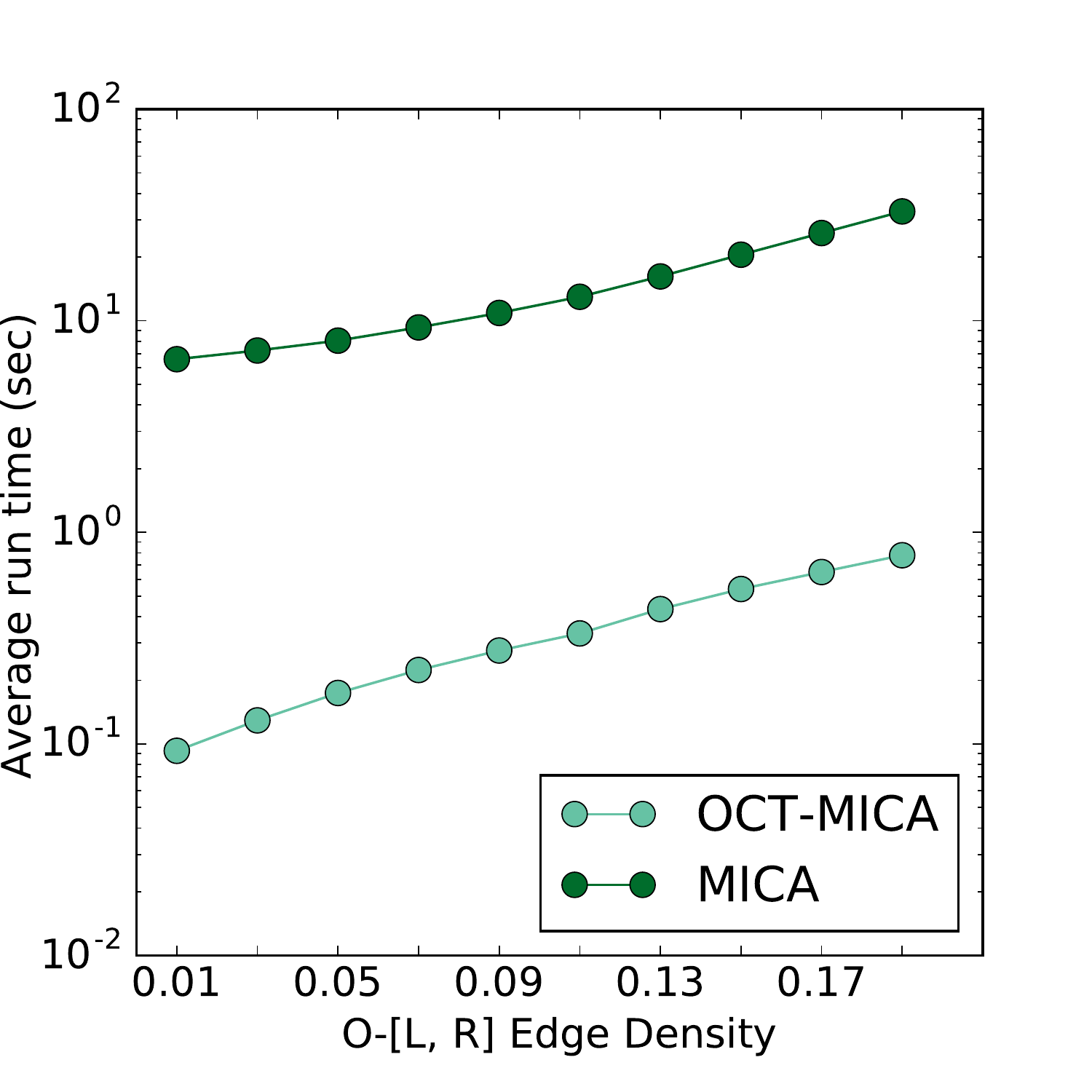}%
    \caption{\label{fig:orig}
	Runtimes of the MIB-enumerating (left) and MB-enumerating (right) algorithms on graphs where $n_B = 1000$, $n_L/n_R = 10$, and $n_O = 10$. The expected edge density between $O$ and $L\cup R$ was varied; all other densities were $0.05$.
    }
\end{figure*}

We implemented \iOCTMIB, \nonlex, \MICA, and \OCTMICA in C++, and used the implementation of \OCTMIB from~\cite{KLOSTER}. All code is open source under a BSD 3-clause license and publicly available as part of \textsf{MI-Bicliques} at~\cite{HORTON}.

\paragraph{Data} For convenience, throughout this section, we assume $n_L \geq n_R$ and let $n_B = n_L + n_R$. Our synthetic data was generated using a modified version of the random graph generator of Zhang et al.~\cite{ZHANG} that augments random bipartite graphs to have OCT sets of known size. The generator allows a user to specify the sizes of $L$, $R$, and $O$ ($n_L$, $n_R$, and $n_O$), the expected edge densities between $L$ and $R$, $O$ and $L \cup R$, and within $O$, and the coefficient of variation ($cv$; the standard deviation divided by the mean) of the expected number of neighbors in $L$ over $R$ and in $L \cup R$ over $O$. The generator is seeded for replicability. We use the na\"ive OCT decomposition $[L,R,O]$ returned by the generator for our algorithm evaluation, but the techniques mentioned in Section~\ref{sec:prelims} could also be used to find alternative OCT sets. Unless otherwise specified, the following default parameters are used: expected edge density $\bar{d} = 5\%$, $cv=0.5$, $n_B = 1000$ and $n_L/n_R = 1/10$; additionally, the edge density between $O$ and $L \cup R$ is the same as that between $L$ and $R$.

To add the edges between $L$ and $R$, the edge density and $cv$ values are used to assign vertex degrees to $R$, and then neighbors are selected from $L$ uniformly at random; this was implemented in the generator of~\cite{ZHANG}. Edges are added between $O$ and $L \cup R$ via the same process, only with the corresponding edge density and $cv$ values.  Finally, we add edges within $O$ with an Erd\H{o}s-R\'enyi process based on expected density (no $cv$ value is used here).

In most experiments we limit $n_O$ to be $O(3\log_3n_B)$, and use a timeout of one hour (3600s). Unless otherwise noted we run each parameter setting with five seeds and plot the average over these instances, using the time-out value as the runtime for instances that don't finish. If not all instances used for a plot point finished, we annotate it with the number of instances that did not time out.\looseness-1

We began by running our algorithms on the same corpus of graphs as in~\cite{KLOSTER} (see ~\ref{sec:original}). As the new algorithms finished considerably faster than those in~\cite{KLOSTER}, we were able to scale up both $n_B$ and $n_O$ to create new sets of experiments, discussed in~\ref{sec:larger}.
We also ran our algorithms on computational biology graphs from~\cite{WERNICKE}, which have been shown to be near-bipartite; these results are in~\ref{sec:huffner}.\looseness-1

\paragraph{Hardware}
All experiments were run on identical hardware; each server had four Intel Xeon E5-2623 v3 CPUs (3.00GHz)
and 64GB DDR4 memory. The servers ran Fedora 27 with Linux kernel 4.16.7-200.fc27.x86\_64.
The C/C++ codes were compiled using gcc/g++ 7.3.1 with optimization flag -O3.

\subsection{Initial Benchmarking}\label{sec:original}

We begin by evaluating our algorithms on the corpus of graphs used in~\cite{KLOSTER}. This dataset was designed to independently test the effect of each parameter (the expected densities in various regions of the graph, the $cv$ values, $n_O$, $n_B$, and $n_L/n_R$) on the algorithms' runtime.
We observe that \iOCTMIB and \OCTMICA are generally the best algorithms for their respective problems, and include comprehensive plots of all experiments in the Appendix.

For \AMBpfull, we observe that in general, \iOCTMIB outperforms \OCTMIB and \nonlex. This is the case when the varying parameter is the density within $O$, the $cv$ between $L$ and $R$, the size of the OCT set $n_O$, and the ratio between $L$ and $R$, amongst other settings.
In these ``near-bipartite" synthetic graphs, \nonlex unsurprisingly is slowest on most instances.
When $n_B= 1000$ and $n_O = 3\log_3(n_B)$, \nonlex outperforms \OCTMIB when the density within $O$ increases above 0.05. This is likely due to the adverse effect of the number of MISs in the OCT set on \OCTMIB. The most interesting observation occurs when varying the edge density between $O$ and $L \cup R$ (left panel of Figure~\ref{fig:orig}). In the $n_O = 10$ case, \iOCTMIB is the fastest algorithm until the density exceeds 0.11, when \OCTMIB becomes faster. We believe this is likely due to \OCTMIB efficiently pruning away attempted expansions which are guaranteed to fail, while the number of MISs in $O$ does not increase. This behavior is also seen in the case where $n_O = 3\log_3n_B$, though the magnitude of the difference is not as extreme.

In the non-induced setting of \MBpfull, \OCTMICA consistently outperforms \MICA on this corpus, typically by at least an order of magnitude. The more interesting takeaway is that both MB-enumerating algorithms run considerably faster than their MIB-enumerating counterparts (e.g. right panel of Figure~\ref{fig:orig}), mostly because the number of MIBs is often one to two orders of magnitude larger than the number of MBs in these instances.

\begin{figure*}[t!]
    \includegraphics [width=0.5\textwidth]{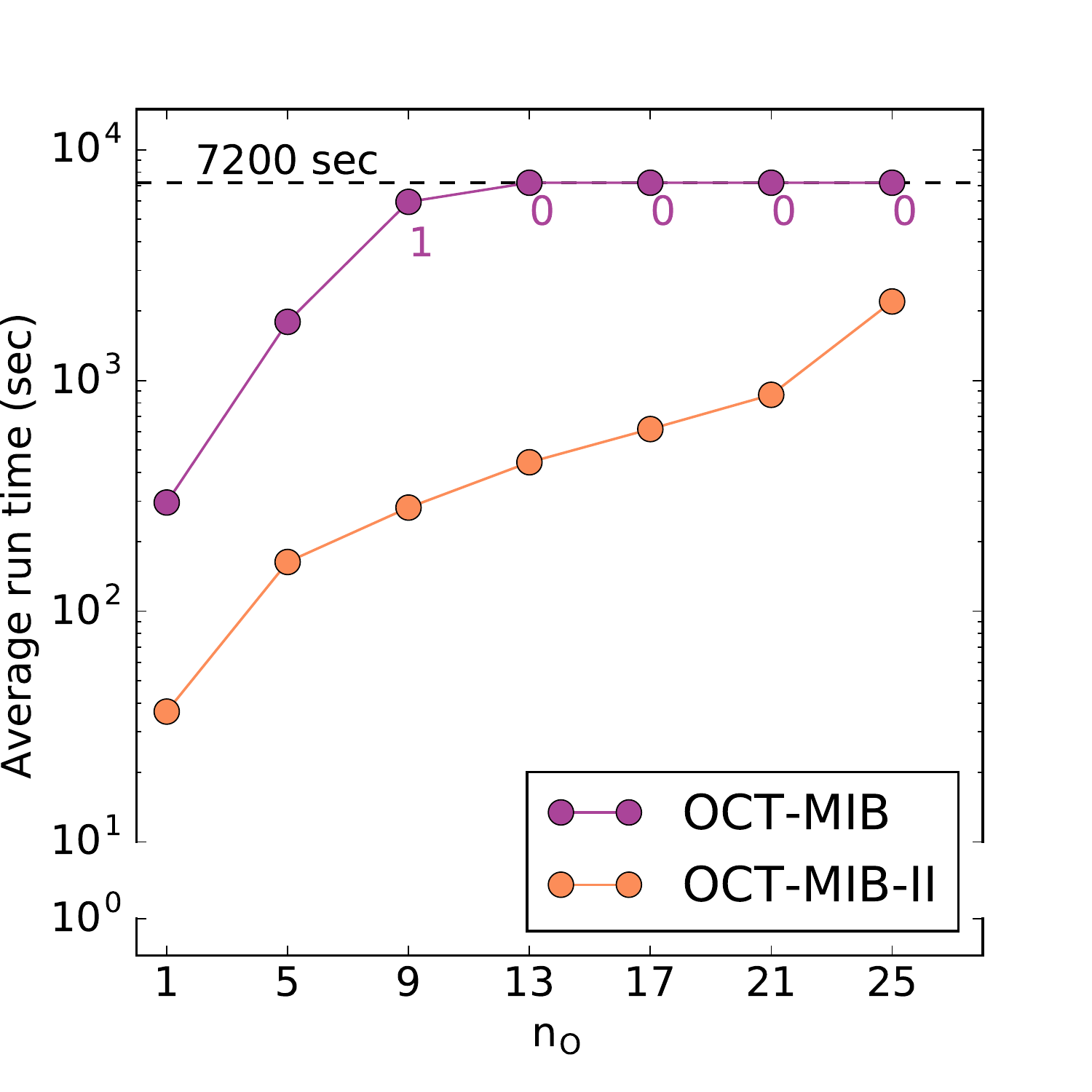}%
    \includegraphics [width=0.5\textwidth]{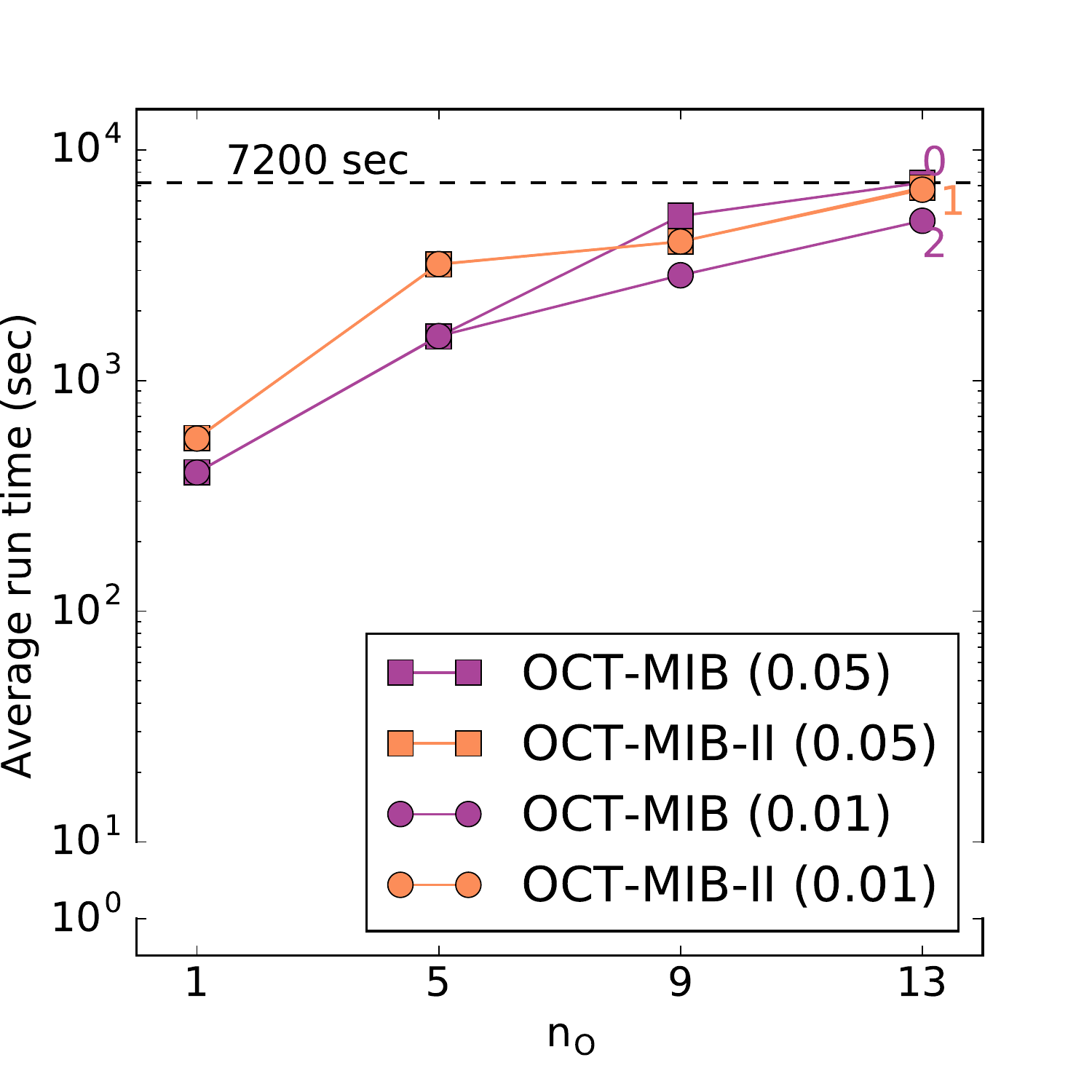}%
    \caption{\label{fig:MIB-10000}
Runtimes of the OCT-based MIB-enumeration algorithms on graphs where $n_B = 10000$ and $n_O$ varies.
In the left panel, $n_L = 9901, n_R = 99$ ($n_L/n_R \approx 100$) and the expected edge density is $0.03$.
In the right panel, $n_L = 9091, n_R = 909$ ($n_L/n_R \approx 10$) and the expected edge density (excluding within $O$) is $0.01$; the marker-type denotes the expected edge density within $O$ (see legend).
For these larger instances we used 3 seeds and a 7200s timeout.
    }
\end{figure*}

\subsection{Larger Graphs}\label{sec:larger}

Given the much faster runtimes achieved in Section~\ref{sec:original} we created a new corpus of larger synthetic graphs. For \AMBpfull, we scaled up $n_B$ to 10,000 and varied $n_O$ in two settings, increasing the timeout to 7200 seconds.
When the expected density was 0.03 and $n_L/n_R = 100$, \iOCTMIB outperformed \OCTMIB for all values of $n_O$ by at least an order of magnitude and finished on all instances, whereas \OCTMIB timed out on all instances with $n_O \geq 13$ (left panel of Figure~\ref{fig:MIB-10000}). However, when the expected density was 0.01 and $n_L/n_R = 9$, \OCTMIB was faster (right panel of Figure~\ref{fig:MIB-10000}). We speculated that this was due to the sparsity of $O$, allowing for a speed-up due to the efficient pruning of \OCTMIB similar to what was seen in Section~\ref{sec:original}. To test this theory, we increased expected edge density within $O$ to 0.05 while leaving the other parameters the same (right panel of Figure~\ref{fig:MIB-10000}), and observed that once $n_O \geq 9$, \iOCTMIB outperforms \OCTMIB, confirming our hypothesis.\looseness-1

For \MBpfull, we also designed a new experiment where $n_B = 10000$ and $n_O$ was scaled up to 1000 (left panel of Figure~\ref{fig:MB-large}). \OCTMICA finished on all instances, whereas \MICA finished on none when $n_O$ was 1000. We also tested how large we could scale the expected density between $L$ and $R$ (right panel of Figure~\ref{fig:MB-large}). When $n_B = 100$, \OCTMICA finished on all instances with density at most 0.4, while \MICA finished on two of five when density is 0.4. Neither algorithm finished in less than the timeout of an hour when the density was 0.5 or greater, exhausting the hardware's memory in many cases. Thus \OCTMICA is able to scale to graphs with considerably larger OCT sets and higher density than both \MICA and the MIB-enumerating algorithms.

We additionally created graphs with $n_O > 3\log_3n_B$, which was not done in~\cite{KLOSTER}, and ran the algorithms for both MIBs and MBs (Figure~\ref{fig:varyingoct}). These graphs had $n_B$ values up to 4000 and for each value of $n_B$, we used three values of $n_O$; $10, 3\log_3n_B,$ and $\sqrt{n_B}$.
The results were most interesting for the MIB-enumerating algorithms (Figure~\ref{fig:varyingoct} top). \OCTMIB performed the worst of the three algorithms when $n_O = \sqrt{n_B}$, but outperformed \nonlex in the other settings. This verifies the analysis from~\cite{KLOSTER} on the range in which \OCTMIB is most effective. In general, \iOCTMIB once again was the fastest algorithm and did best when $n_O$ was smaller. The impact of $n_O$ on \iOCTMIB and \nonlex appeared comparable.
In the MB-enumeration case, \OCTMICA consistently outperforms \MICA, and there is a distinguishable difference in the runtime based on the value of $n_O$ (Figure~\ref{fig:varyingoct} bottom). The value of $n_O$ has far less effect on \MICA, which does not finish on any graphs with $n_B = 4000$.

\begin{figure*}[t!]
    \includegraphics [width=0.5\textwidth]{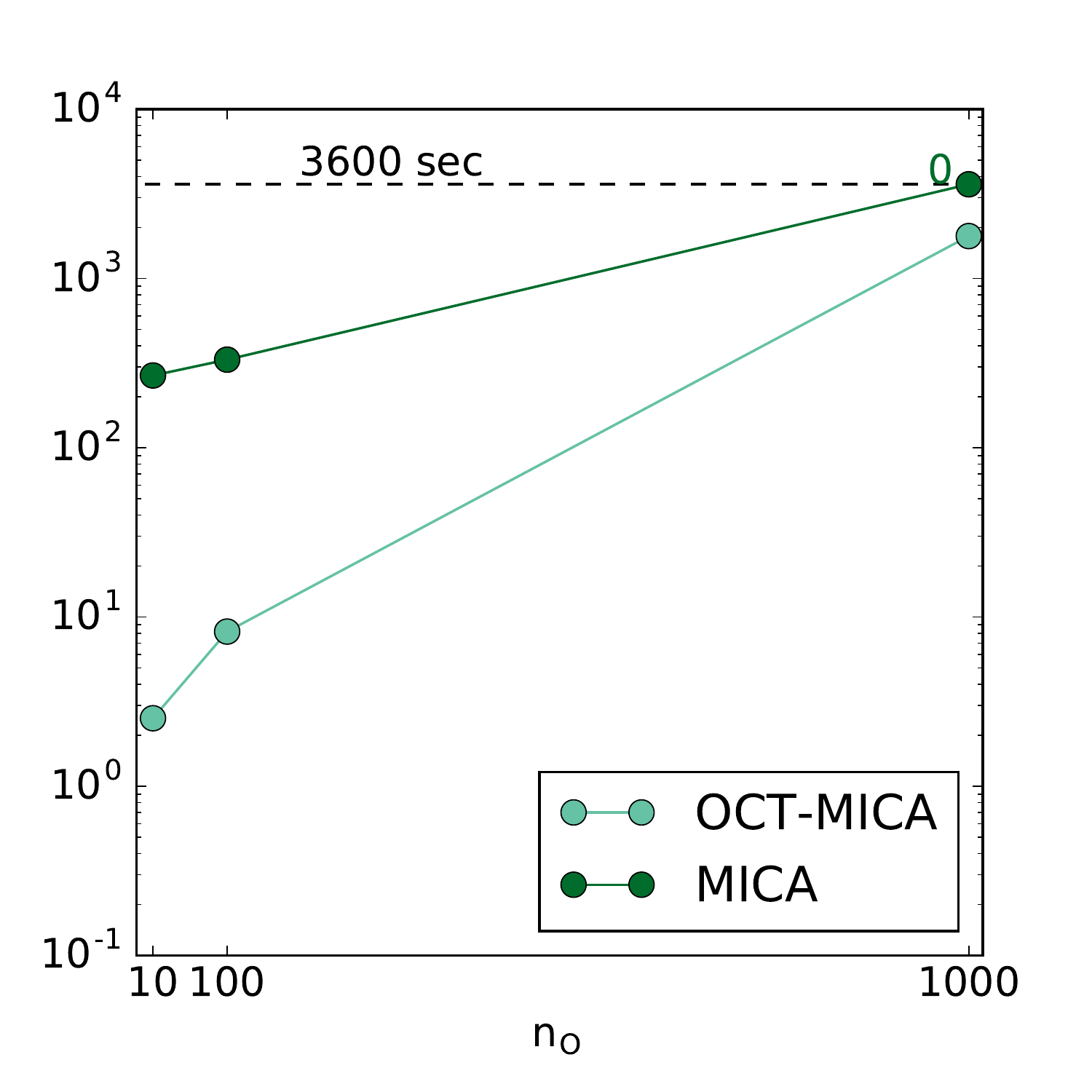}%
    \includegraphics [width=0.5\textwidth]{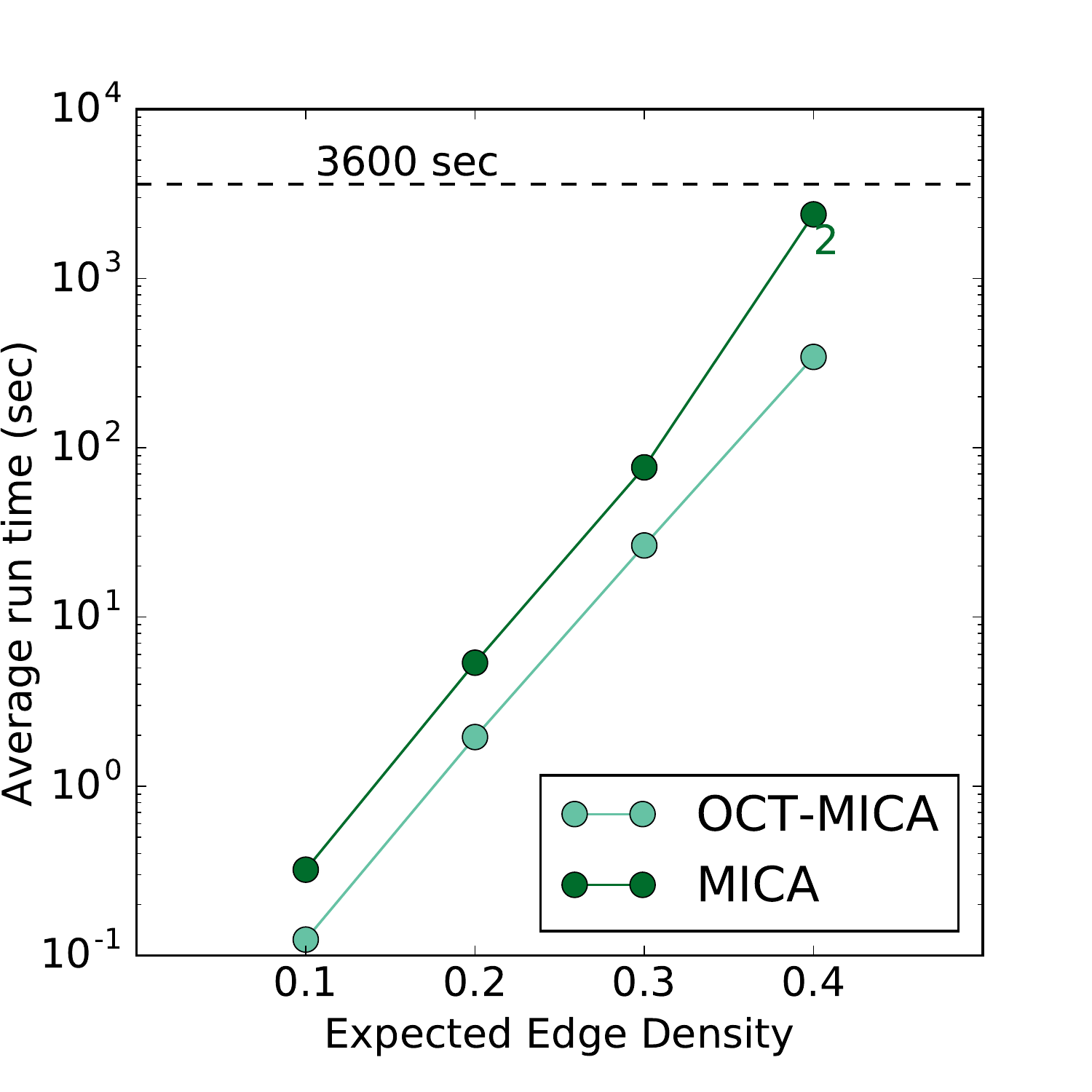}%
    \caption{\label{fig:MB-large}
   	Runtimes of the MB-enumerating algorithms on graphs with larger $n_B$ and expected edge density. In the left panel, $n_L = 9091, n_R = 909$ ($n_L/n_R \approx 10$), the expected edge density is $0.05$, and $n_O$ varied. In the right panel, $n_L = 91, n_R = 9$ ($n_L/n_R \approx 10$),  $n_O = 50$, and the expected edge density varied.\looseness-1
}
\end{figure*}

\begin{figure*}[t!]
\captionsetup[subfigure]{labelformat=empty}

\centering

\begin{subfigure}[b]{1\textwidth}
   \includegraphics[width=1\textwidth, trim=0cm .25cm .25cm .75cm,clip]{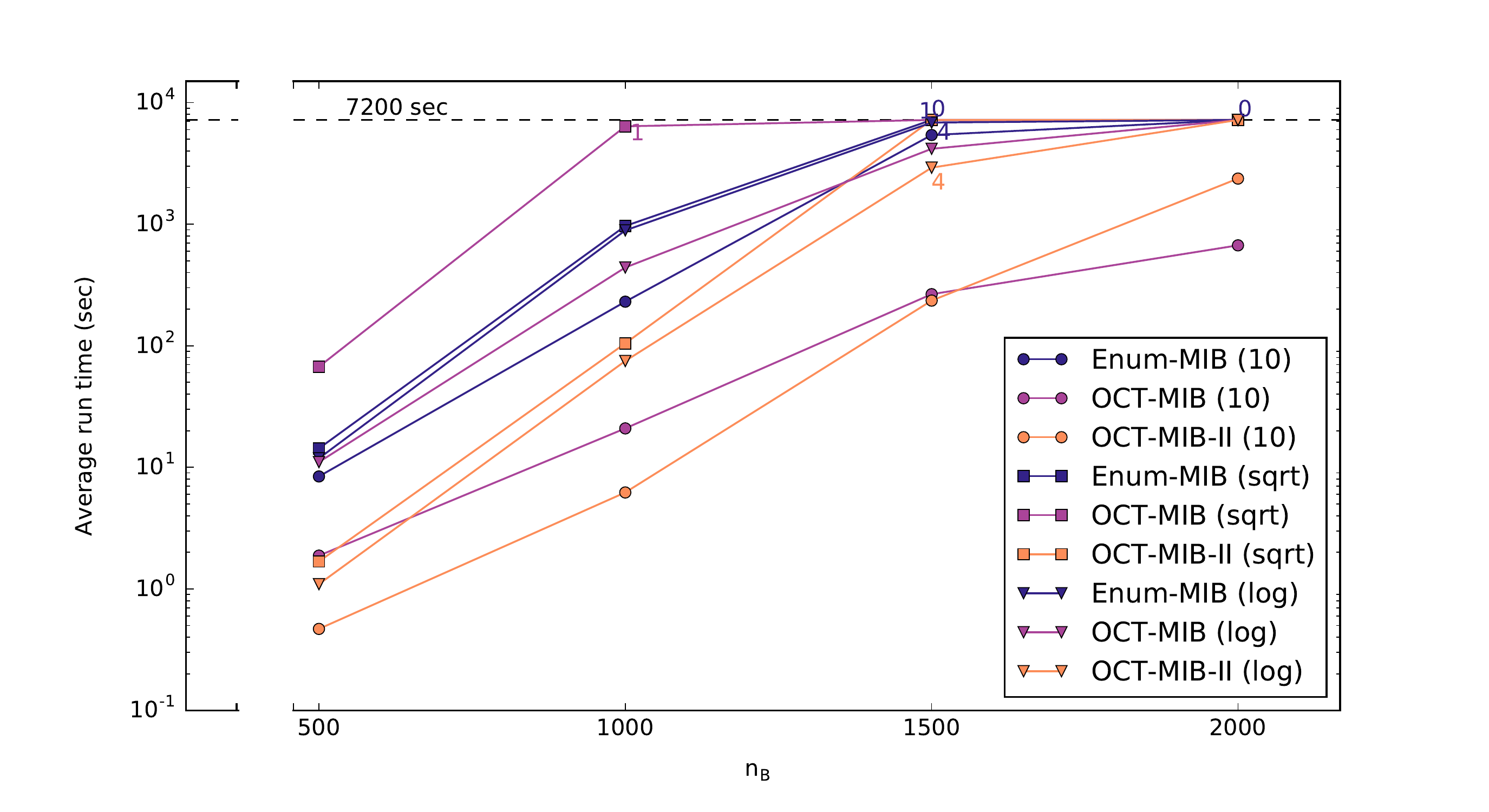}
   \caption{}

\end{subfigure}\\[-5ex]

\begin{subfigure}[b]{1\textwidth}
   \includegraphics[width=1\textwidth, trim=0cm .25cm .25cm .25cm,clip]{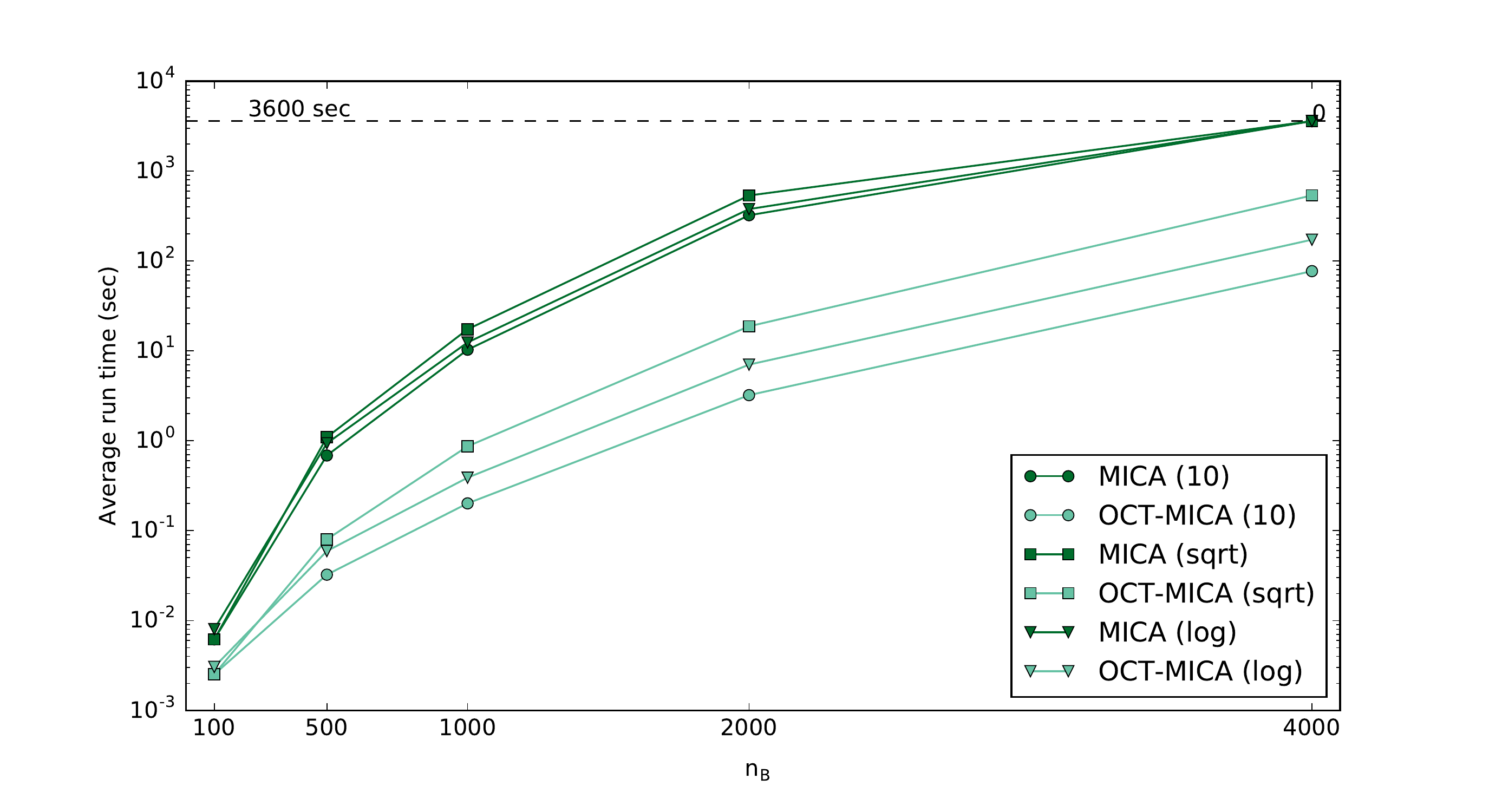}
   \caption{}

\end{subfigure}\\[-4ex]

    \caption{\label{fig:varyingoct}
   Runtimes of the MIB-enumerating (top) and MB-enumerating (bottom) algorithms on graphs where $n_L/n_R = 9$ and all expected edge densities are $0.05$. $n_B$ is varied (x-axis) and  the marker-type denotes the value of $n_O \in \{10, \sqrt{n_B}, 3\log_3(n_B)\}$ (see legend). The time-out value is set to 7200s for the MIB-enumerating algorithms and 3600s for the MB-enumerating algorithms.
    }
\end{figure*}

\subsection{Computational Biology Data}\label{sec:huffner}

Finally, we tested performance on real-world data using the graphs from~\cite{WERNICKE}, which come from computational biology. These graphs have previously been exhibited to have small OCT sets~\cite{HUFFNER}, and we used the implementation from~\cite{GOODRICH} of H{\"u}ffner's iterative compression algorithm~\cite{HUFFNER} to find the OCT decompositions. Computing the OCT decomposition for each graph ran in less than ten seconds, and often in less than one second. As can be seen in Table~\ref{tab:WH-sub}, \iOCTMIB performs the best of the MIB-enumerating algorithms and \OCTMICA is faster than \MICA. Full results are in the Appendix.

\begin{table}

{\small
    \begin{tabularx}{\textwidth}{@{\extracolsep{\fill}} cccc|cccc|ccc}
        \toprule
        $G$ & $n_B$ & $m$ & $n_O$ & $|M_I|$ & \iOCTMIB & \OCTMIB & \nonlex & $|M_B|$ & \OCTMICA & \MICA \\
        \midrule
\textsf{aa-24} & 258 & 1108 & 21 & 3890 & 2.108 & 9.140 & 14.167 & 1334 & 0.237 & 2.477 \\
\textsf{aa-30} & 39 & 71 & 4 & 56 & 0.002 & 0.007 & 0.006 & 36 & 0.002 & 0.007 \\
\textsf{aa-41} & 296 & 1620 & 40 & 11705 & 16.519 & 82.439 & 50.205 & 20375 & 9.059 & 47.789 \\
\textsf{aa-50} & 113 & 468 & 18 & 1272 & 0.322 & 0.778 & 1.098 & 1074 & 0.132 & 0.612 \\
\textsf{j-20} & 241 & 640 & 1 & 274 & 0.013 & 0.065 & 0.484 & 228 & 0.009 & 0.188 \\
\textsf{j-24} & 142 & 387 & 4 & 150 & 0.013 & 0.027 & 0.089 & 104 & 0.007 & 0.025 \\
        \bottomrule
    \end{tabularx}
}
    \caption{\label{tab:WH-sub}
        A sampling of the runtimes of the biclique-enumeration algorithms on the Wernicke-H{\"u}ffner computational biology data~\cite{WERNICKE}.
    }
\end{table}

%% file: conclusion.tex
We present a suite of new algorithms for enumerating maximal (induced) bicliques in general graphs, two of which are parameterized by the size of an odd cycle transversal. It is particularly noteworthy that the parameterized algorithms empirically outperform the general approaches even when their asymptotic worst-case complexities are worse. This highlights a weakness of standard complexity analysis, as many aspects of an algorithm get ``swept under the rug". 

It is also interesting that even though \AMBpfull and \MBpfull are closely related problems, the MB-enumerating algorithms are often an order of magnitude faster than their MIB-enumerating counterparts. The reason for this can likely be attributed to two causes: the number of MBs is significantly less than the number of MIBs in sparse graphs, and that the stricter structure of MIBs requires more work to ensure. For $S \subseteq V$, there is exactly one MB of the form $S \times T \subseteq V$ in $G$, but there can be many MIBs with this structure.

We implement and benchmark all of the algorithms on a corpus of synthetic and real-world computational biology graphs, and establish that parameterized approaches are often at least an order of magnitude faster than the general approaches. This remains true even when $n_O \in O(\sqrt{n})$. It would be interesting to experimentally evaluate as $n_O$ increases, at what point the standard methods outperform those optimized for near-bipartite graphs. Finally, we note as in ~\cite{KLOSTER}, the current implementations of the algorithms could be improved by replacing the MIS-enumeration algorithm with that of~\cite{TSUKIYAMA}, and the M(I)B-enumeration on bipartite graphs with the implementation used in~\cite{ZHANG}.

%% file: app-MIB.tex
We now provide algorithmic details and proofs of the complexity and correctness of \MM and \AT.

\subsection{\MM}

Recall that \MM takes in $(C, S)$, where $C$ is an induced biclique and $S \subseteq V$, and either returns a MIB $C^+$ where $C \subseteq C^+$, $C^+ \subseteq C \cup S$, $C \neq \emptyset$, or returns $\emptyset$. If it returns $\emptyset$ and $C \neq \emptyset$ then there is another MIB $D$ which contains $C$ and $v \in (V \setminus S) \setminus C$. We give pseudo-code of \MM in Algorithm~\ref{pc:MM}.

\begin{algorithm}[h!]
\begin{algorithmic}[1]
	\State Input: {$G=(V,E)$, $C=C_1\times C_2$, $S$}

	\State Let $C_S = S \setminus (C_1 \cup C_2)$

	\If {$C == \emptyset$}
		\State\Return{$\emptyset$}
	\EndIf

	\For {$v \in C_S$}
		\If {$|N(v) \cap C_1| == |C_1| \And |N(v) \cap C_2| == 0$}
			\State $C_2 = C_2 \cup \{v\}$
			\State $C_S \setminus \{v\}$
		\EndIf
	\EndFor

	\For {$v \in C_S$}
		\If {$|N(v) \cap C_2| == |C_2| \And |N(v) \cap C_1| == 0$}
			\State $C_1 = C_1 \cup \{v\}$
		\EndIf
	\EndFor

	\State $V_S = V \setminus (S \cup C_1 \cup C_2)$

	\For {$v \in V_S$}
		\If {$|N(v) \cap C_1| == |C_1| \And |N(v) \cap C_2| == 0$}
			\State\Return{$\emptyset$}
		\EndIf
	\EndFor
	\For {$v \in V_S$}
		\If {$|N(v) \cap C_2| == |C_2| \And |N(v) \cap C_1| == 0$}
			\State\Return{$\emptyset$}
		\EndIf
	\EndFor

	\State\Return{$C^+ = C_1\times C_2$}

\end{algorithmic}
\caption{\MM} \label{pc:MM}
\end{algorithm}

\begin{lemma}
\MM returns a MIB $C^+$ where $C \subseteq C^+$, $C^+ \subseteq C \cup S$, $C \neq \emptyset$, or returns $\emptyset$. 
\end{lemma}

\begin{proof}
Referring to the pseudo-code in Algorithm~\ref{pc:MM},
it is clear that $C \subseteq C^+$, as no vertices are ever removed from the input biclique $C$. 
Furthermore, the only vertices added to $C^+$ are from $S$, so $C^+ \subseteq C \cup S$ and $C^+$ is the only biclique returned by \MM. 
Note that neither side of $C$ is empty and the only vertices added are independent from the side of the biclique which they are added to, so if we do not return $\emptyset$ the object returned is an induced biclique.
If no node from outside of $S$ can be added to $C^+$, then we will not return $\emptyset$ and thus $C^+$ is maximal.

\end{proof}

\begin{lemma}
If \MM returns $\emptyset$ and $C \neq \emptyset$ then there is another MIB $D$ in $G$ which contains $C$ and $v \in (V \setminus S) \setminus C$.
\end{lemma}

\begin{proof}
Note that $C \subseteq C^* = C_1 \times C_2$ at line 12. As \MM returns $\emptyset$ there must be a vertex $v \in V_S = V \setminus (S \cup C^*)$ which can be added to $C^*$.
Let $D$ be a MIB containing $C^*$ and $v$, thus $D$ suffices to prove the lemma.

\end{proof}

\begin{lemma}

\MM runs in $O(m)$ time.

\end{lemma}

\begin{proof}
Note that because $G$ is connected, $n \in O(m)$.
Setting $C_S$ and $V_S$ can be done in $O(n)$ time.
In each for loop, we can scan all of the edges incident to each $v$ in the iterated-over set and keep count of how many nodes from $C_i$ have been seen (checking for inclusion can be done in $O(1)$ time with an $O(n)$ initialization step).
Thus, each edge is scanned at most once per for loop.

\end{proof}

\subsection{\AT}
Recall that \AT takes in $(C, v)$ where $C=C_1 \times C_2$ is an induced biclique and $v \in V \setminus (C_1 \cup C_2)$, and returns the induced biclique where $v$ is added to $C_1$, $N(v)$ is removed from $C_1$, and $\overline{N}(v)$ is removed from $C_2$ if $C_2 \setminus \overline{N}(v) \neq \emptyset$ and $\emptyset$ otherwise.
We give pseudo-code of \AT in Algorithm~\ref{pc:AT}.

\begin{algorithm}[h!]
\begin{algorithmic}[1]
	\State Input: {$G=(V,E)$, $C=C_1\times C_2$, $v \in V \setminus (C_1 \cup C_2)$}

	\State $C_1' = (C_1 \cup \{v\}) \setminus N(v)$
	\State $C_2' = C_2 \cap {N(v)}$

	\If{$C_2' == \emptyset$}
		\State\Return{$\emptyset$}

	\EndIf

	\State\Return{$C_1'\times C_2'$}

\end{algorithmic}
\caption{\AT} \label{pc:AT}
\end{algorithm}

\begin{lemma}
\AT returns the induced biclique where $v$ is added to $C_1$, $N(v)$ is removed from $C_1$, and $\overline{N}(v)$ is removed from $C_2$ if $C_2 \setminus \overline{N}(v) \neq \emptyset$, and $\emptyset$ otherwise.
\end{lemma}

\begin{proof}
Referring to the pseudo-code in Algorithm~\ref{pc:AT}, it is clear that $v$ is added to $C_1$ and $N(v)$ is removed from $C_1$. Additionally $v$'s non-neighbors are effectively removed from $C_2$ by intersecting it with $N(v)$. 
If $C_2' = \emptyset$ then  $C_2 \setminus \overline{N}(v) = \emptyset$ and $\emptyset$ is returned. Otherwise $C_1' \neq \emptyset$ since it includes $v$ and thus $C_1' \times C_2'$ is a biclique. $C_1' \times C_2'$ must be an induced biclique as $C_2' \subseteq C_2$, $C_1' \setminus \{v\} \subseteq C_1$, and $C_1 \times C_2$ is an induced biclique and $(N(v) \cap C_1') = \emptyset$ by definition.

\end{proof}

\begin{lemma}

\AT runs in $O(m)$ time.

\end{lemma}

\begin{proof}
Note that because $G$ is connected, $n \in O(m)$.
\AT can be completed by scanning all of $v$'s $O(m)$ incident edges in tandem with an $O(n)$ preprocessing step to allow for constant-time look-ups when checking for inclusion in a set.

\end{proof}

%% file: app-MB.tex
We give a detailed description of the \MMNI and \consensus subroutines used in \OCTMICA, along with arguments of their correctness and complexity.
\subsection{\MMNI}
Extending a biclique to be maximal is different in the non-induced case from the induced case, since MBs are completely characterized by one side of the biclique. 
\begin{algorithm}[h!]
\caption{\MMNI} \label{pc:MMNI}
	\begin{algorithmic}[1]
		\State Input: {$G=(V,E)$, $B=X\times Y$}
		
		\State$X^*=\cap_{i\in Y}N(i)$
		\State$Y^*=\cap_{i\in X^*}N(i)$
		\State\Return$B^*=X^* \times Y^*$
	\end{algorithmic}
\end{algorithm}
\begin{lemma}
\MMNI runs in $O(m)$ time.
\end{lemma}
\begin{proof}
In order to form $X^*$, we can scan the edges incident to each $v \in Y$ and keep count of how many nodes from $X^*$ have been seen (checking for inclusion can be done in $O(1)$ time with an $O(n)$ initialization step). The same can be done for $Y^*$, where instead we scan the edges incident to each $v \in X^*$. Thus, each edge is scanned at most twice in \MMNI.

\end{proof}
\subsection{\consensus}
The \MICA section of \OCTMICA relies heavily on the \consensus operation introduced in \cite{ALEXE} for finding new candidate bicliques. For each pair of bicliques, there are four candidate bicliques which form the {\em consensus} of the pair. Note that any of the four candidates may be empty and if so discarded. \consensus runs in $O(n)$ time using standard techniques for set union and intersection.

\begin{algorithm}[h!]
	\caption{\consensus} \label{pc:consensus}
	\begin{algorithmic}[1]
		\State Input: {$G=(V,E)$, $B_\alpha=X_\alpha \times Y_\alpha$, $B_\beta=X_\beta \times Y_\beta$}
		
		\State$B_1=(X_\alpha \cup X_\beta) \times (Y_\alpha \cap Y_\beta)$
		\State$B_2=(X_\alpha \cap X_\beta) \times (Y_\alpha \cup Y_\beta)$
		\State$B_3=(Y_\alpha \cup X_\beta) \times (X_\alpha \cap Y_\beta)$
		\State$B_4=(X_\alpha \cap Y_\beta) \times (Y_\alpha \cup X_\beta)$
		\State $S=\{\}$
		\For{$B_i=X_i\times Y_i\in\{B_1,B_2,B_3,B_4\}$}
			\If{$|X_i|>0 \And |Y_i|>0$}
				\State $S$\Call{.add}{$B_i$}
			\EndIf
		\EndFor
		\State\Return$S$
	\end{algorithmic}
\end{algorithm}

%% file: app-exp.tex
Here we include figures corresponding to additional experimental results of our initial benchmarking and on the computation biology data from~\cite{WERNICKE} described in sections \ref{sec:original} and \ref{sec:huffner} respectively.

\begin{figure*}[t!]
    \includegraphics [width=0.5\textwidth]{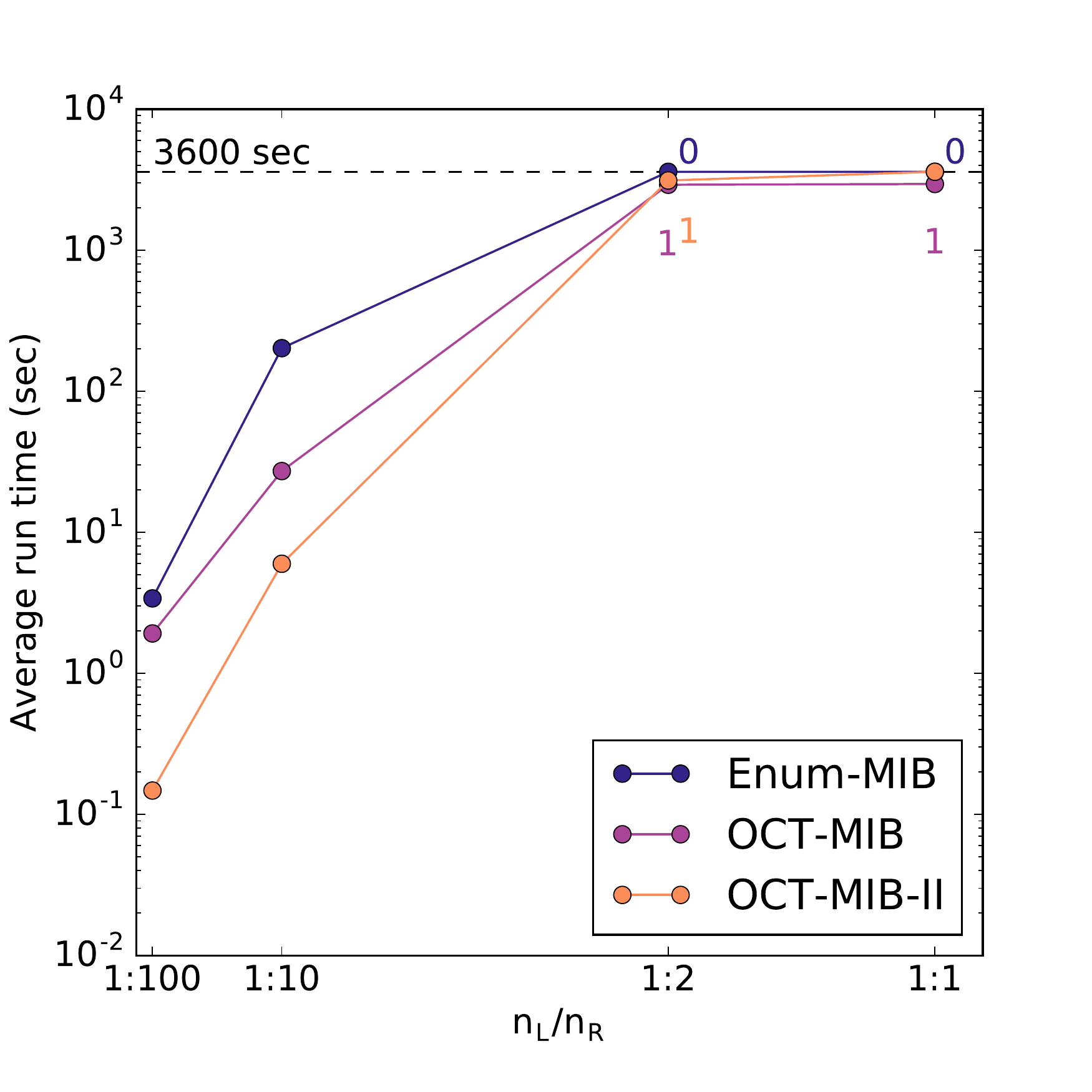}%
\hspace{0.01cm}
    \includegraphics [width=0.5\textwidth]{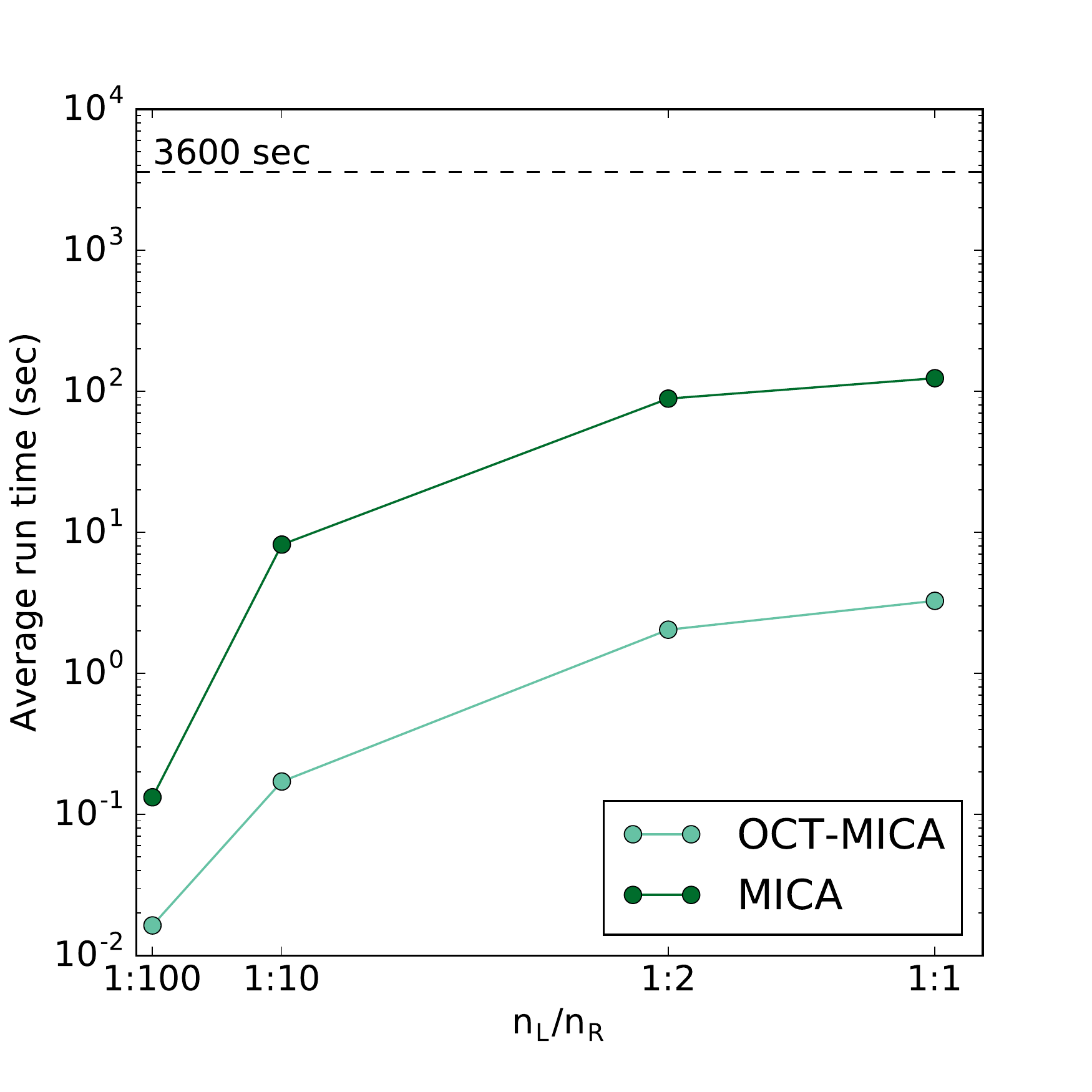}%
    \caption{\label{fig:app1}
	Runtimes of the MIB-enumerating (left) and MB-enumerating (right) algorithms on graphs where $n_B = 1000$ and $n_O = 10$. The ratio $n_L/n_R$ was varied.
    }
\end{figure*}

\begin{figure*}[t!]
    \includegraphics [width=0.5\textwidth]{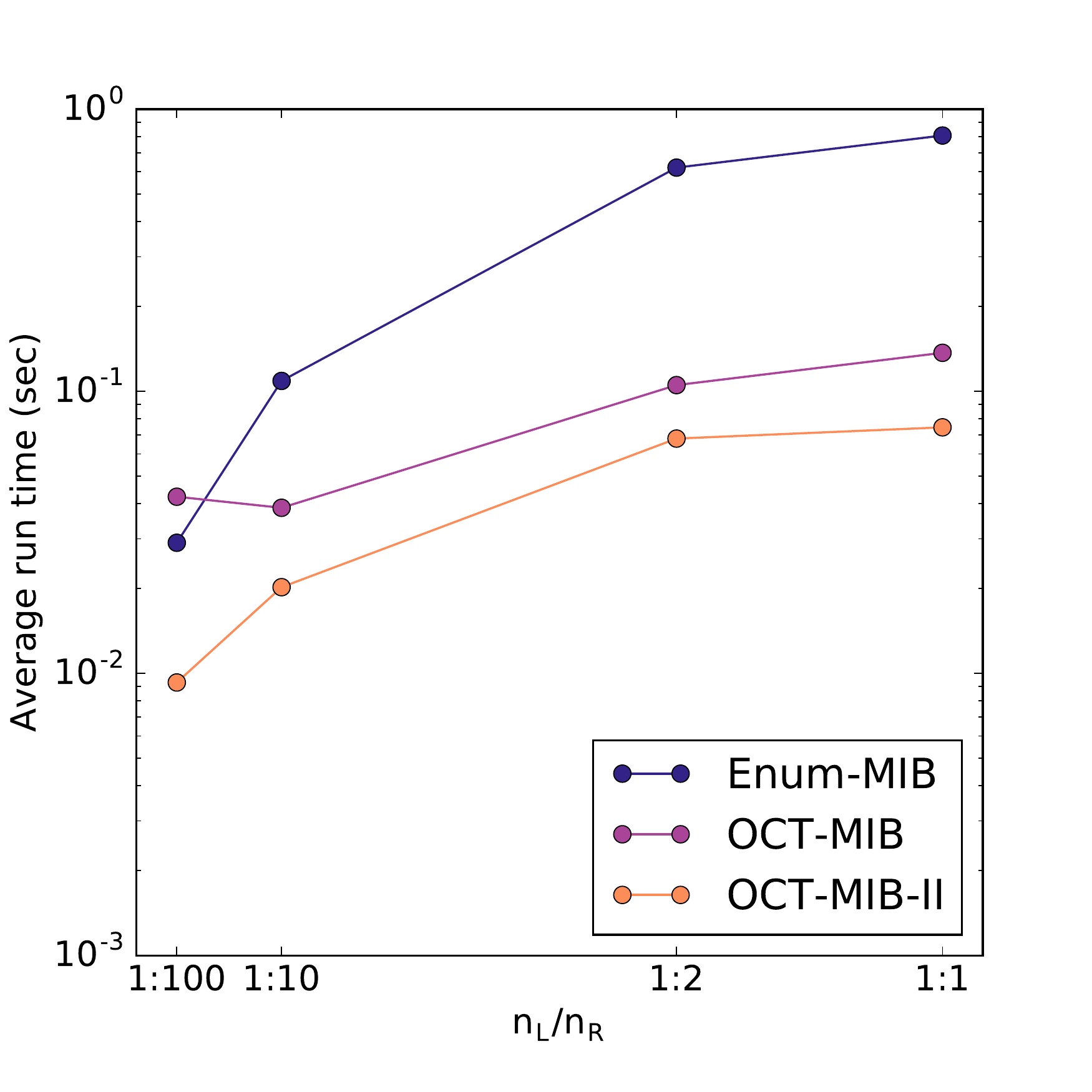}%
\hspace{0.01cm}
    \includegraphics [width=0.5\textwidth]{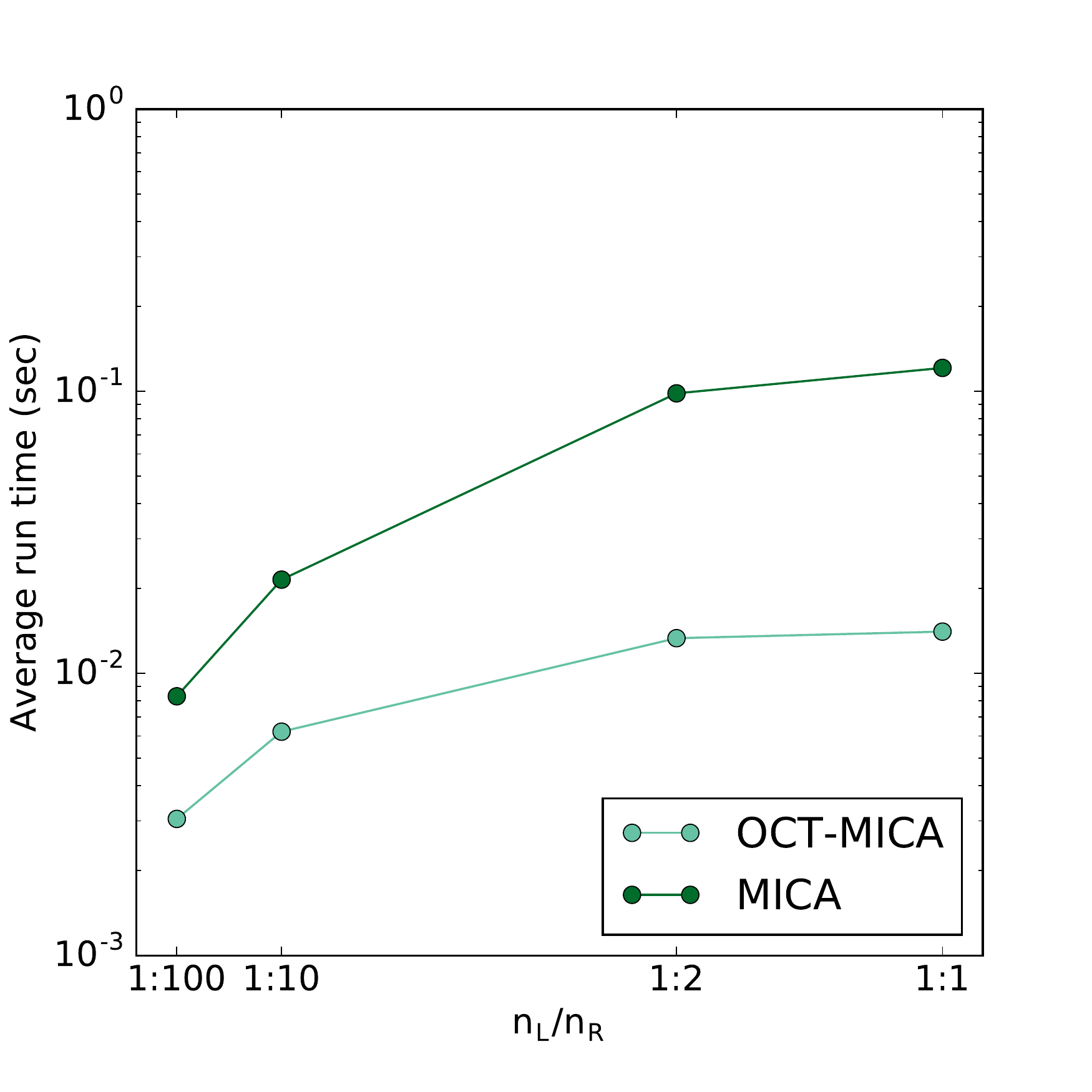}%
    \caption{\label{fig:app2}
	Runtimes of the MIB-enumerating (left) and MB-enumerating (right) algorithms on graphs where $n_B = 200$ and $n_O = 10$. The ratio $n_L/n_R$ was varied.
    }
\end{figure*}

\begin{figure*}[t!]
    \includegraphics [width=0.5\textwidth]{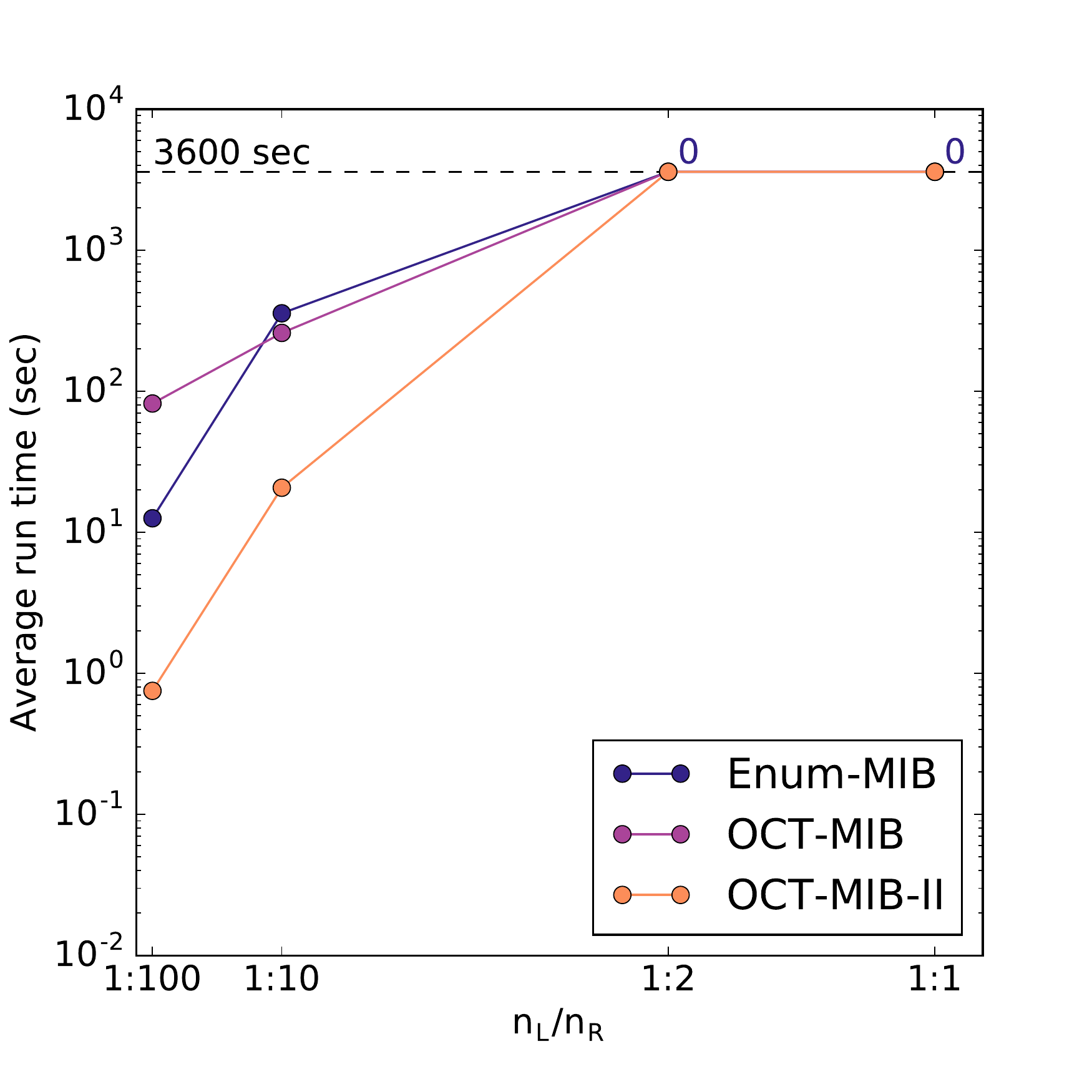}%
\hspace{0.01cm}
    \includegraphics [width=0.5\textwidth]{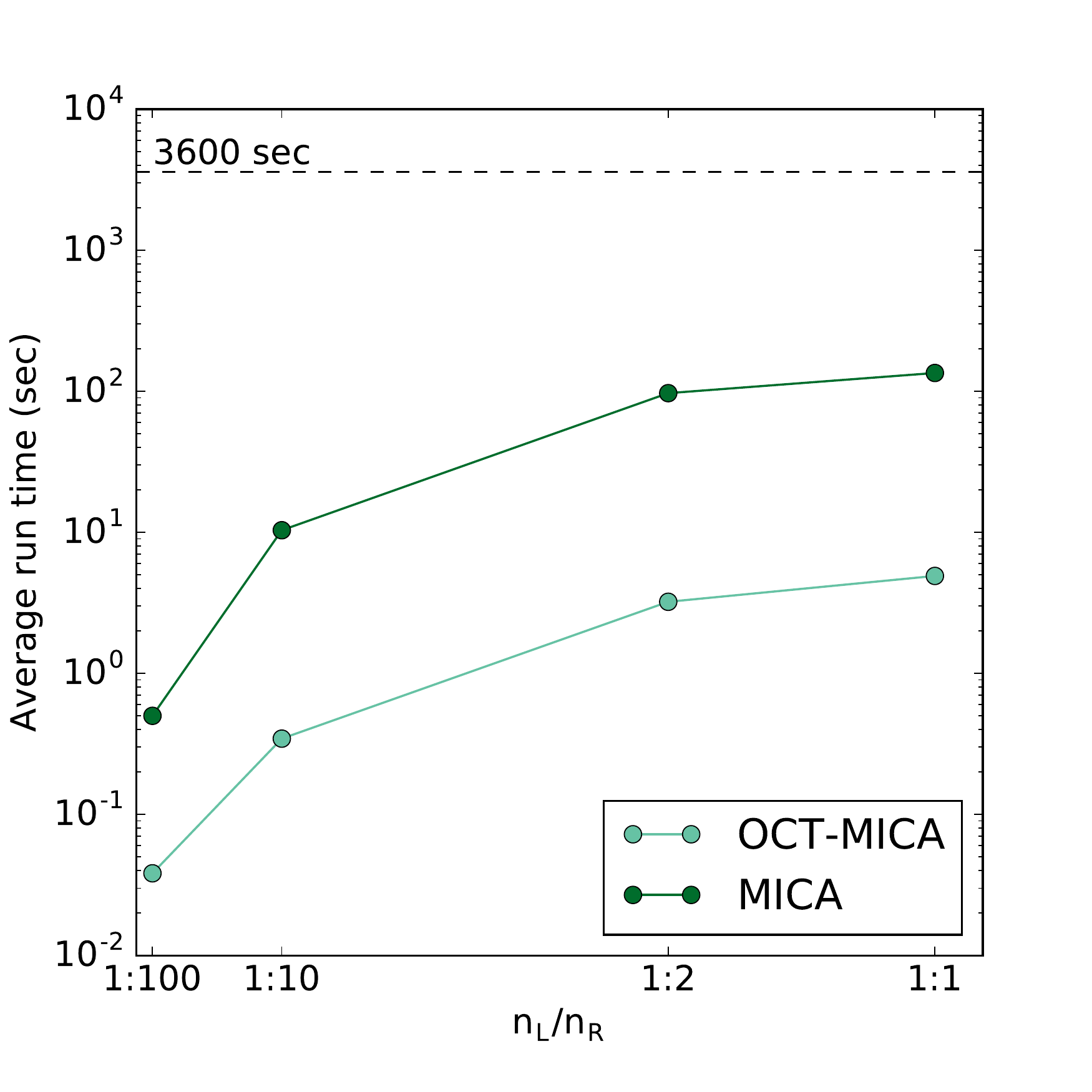}%
    \caption{\label{fig:app3}
	Runtimes of the MIB-enumerating (left) and MB-enumerating (right) algorithms on graphs where $n_B = 1000$ and $n_O = 19 \approx 3\log_3(n_B)$. The ratio $n_L/n_R$ was varied.
    }
\end{figure*}

\begin{figure*}[t!]
    \includegraphics [width=0.5\textwidth]{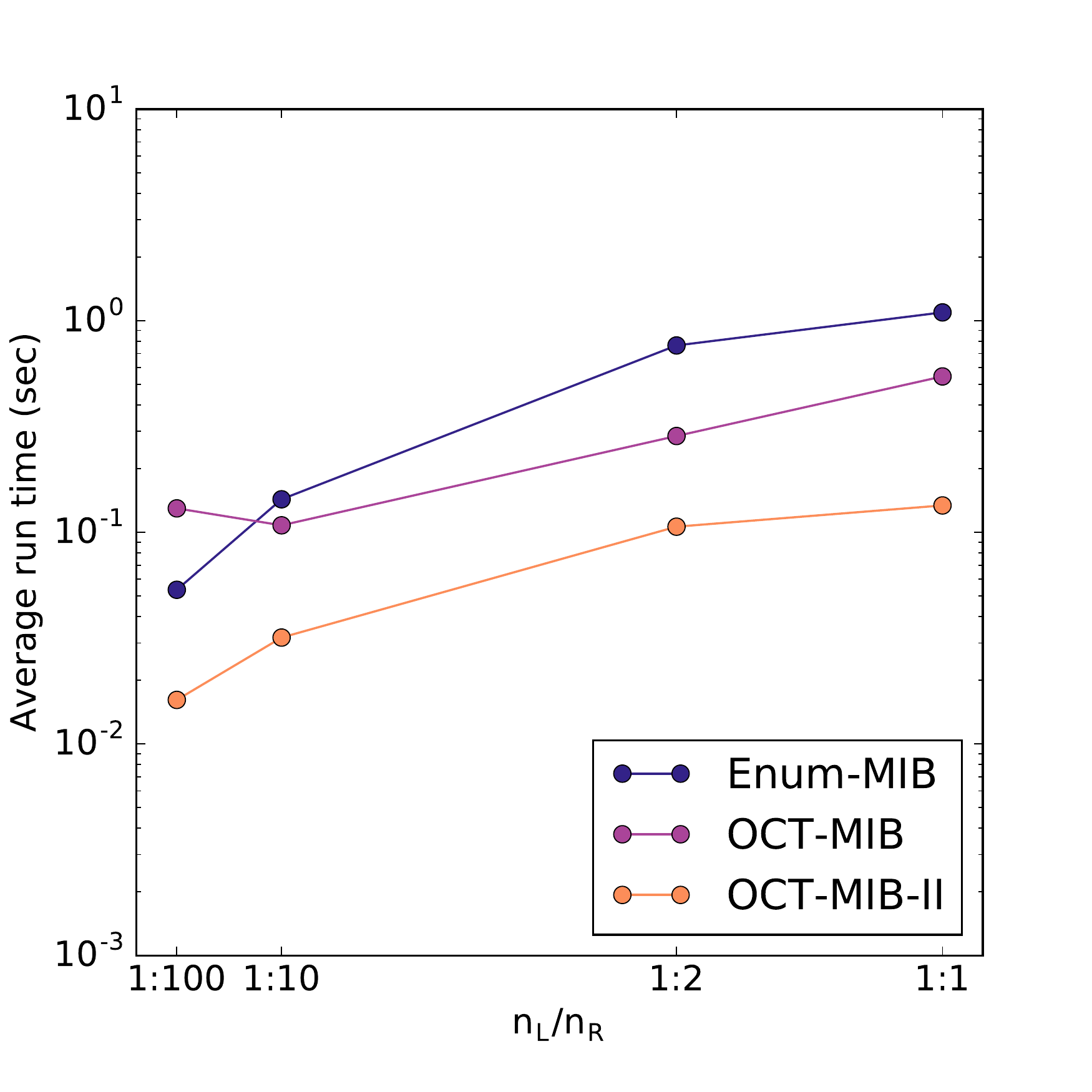}%
\hspace{0.01cm}
    \includegraphics [width=0.5\textwidth]{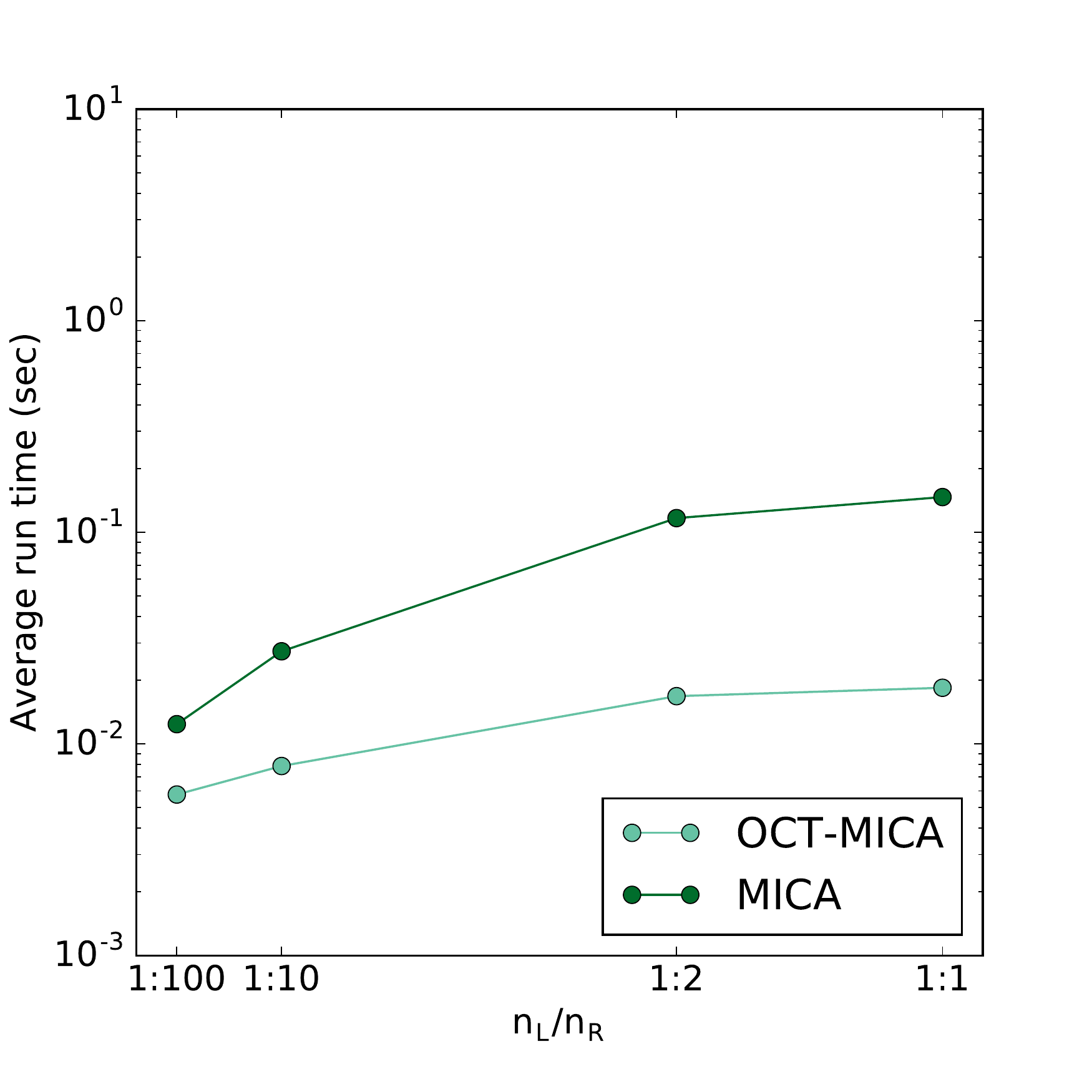}%
    \caption{\label{fig:app4}
	Runtimes of the MIB-enumerating (left) and MB-enumerating (right) algorithms on graphs where $n_B = 200$ and $n_O = 14 \approx 3\log_3(n_B)$. The ratio $n_L/n_R$ was varied.
    }
\end{figure*}

\begin{figure*}[t!]
    \includegraphics [width=0.5\textwidth]{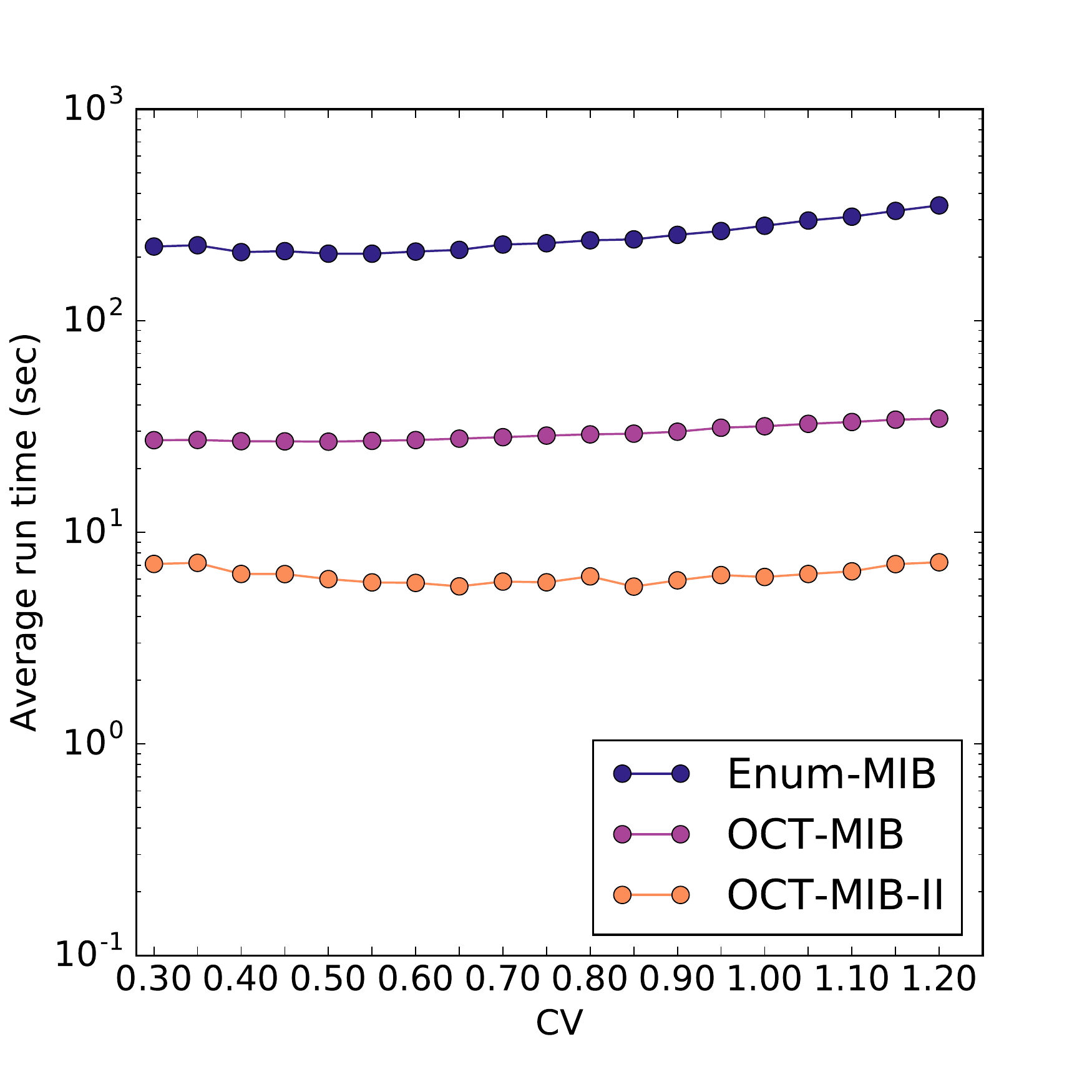}%
\hspace{0.01cm}
    \includegraphics [width=0.5\textwidth]{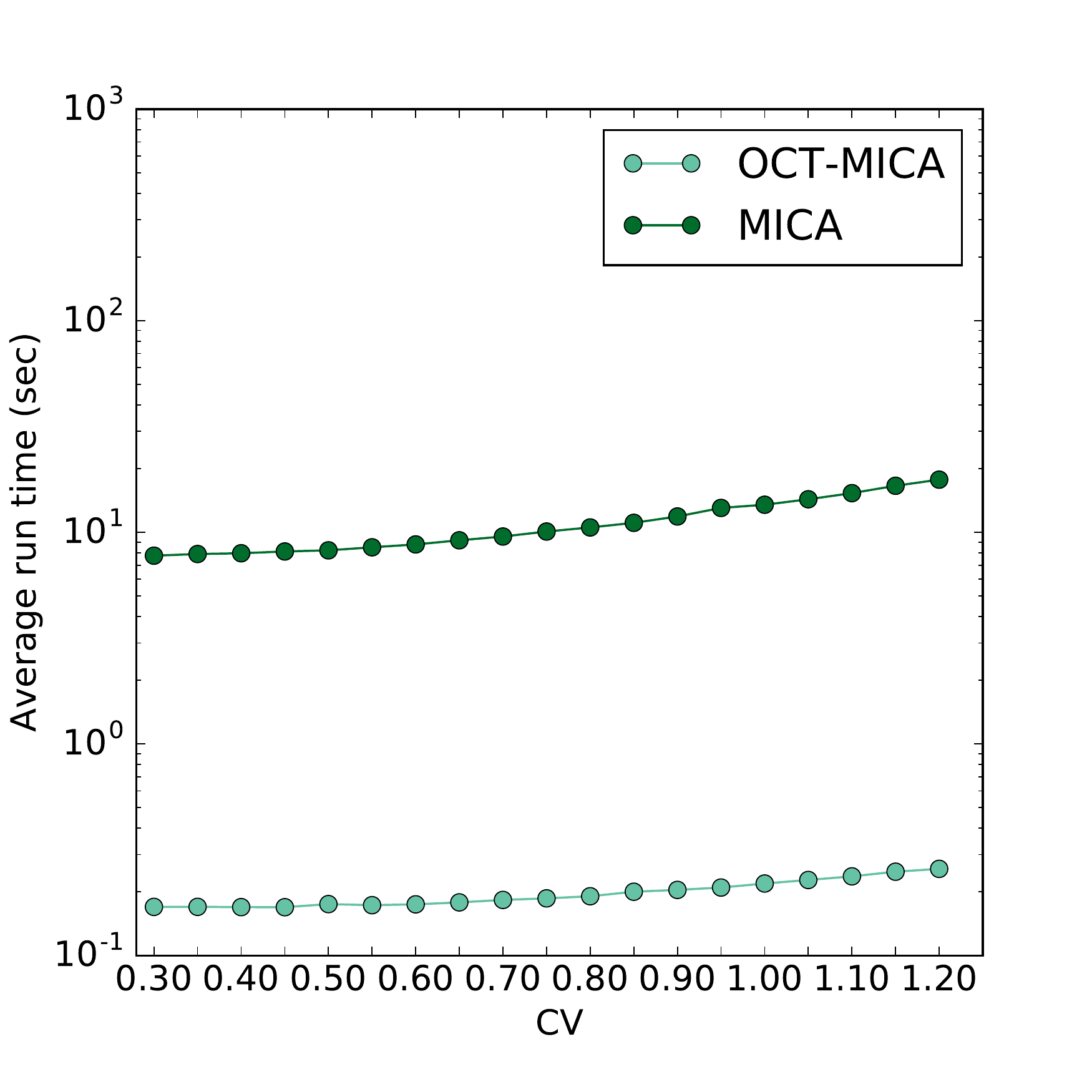}%
    \caption{\label{fig:app5}
	Runtimes of the MIB-enumerating (left) and MB-enumerating (right) algorithms on graphs where $n_B = 1000$ and $n_O = 10$. The coefficient of variation between $L$ and $R$ was varied.
    }
\end{figure*}

\begin{figure*}[t!]
    \includegraphics [width=0.5\textwidth]{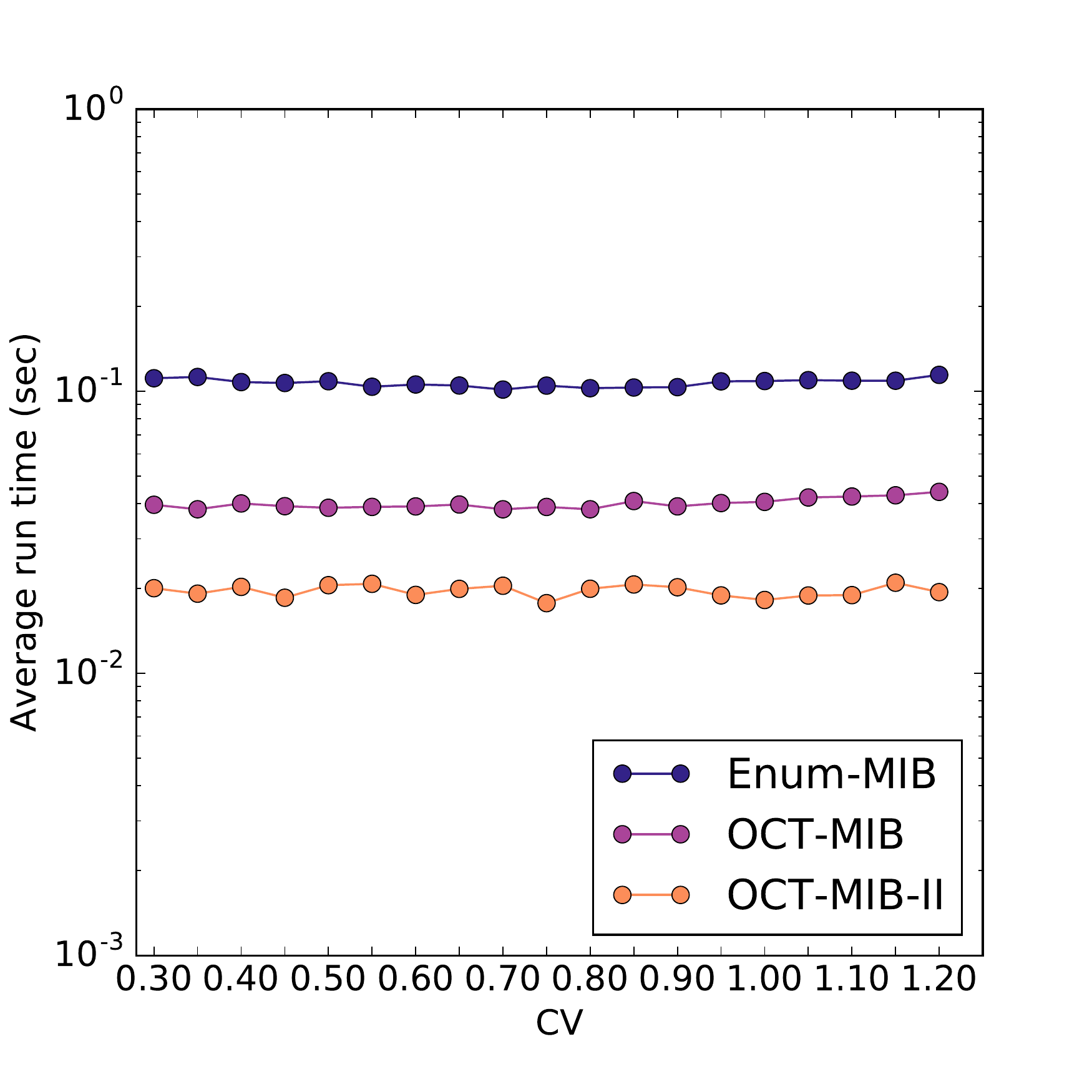}%
\hspace{0.01cm}
    \includegraphics [width=0.5\textwidth]{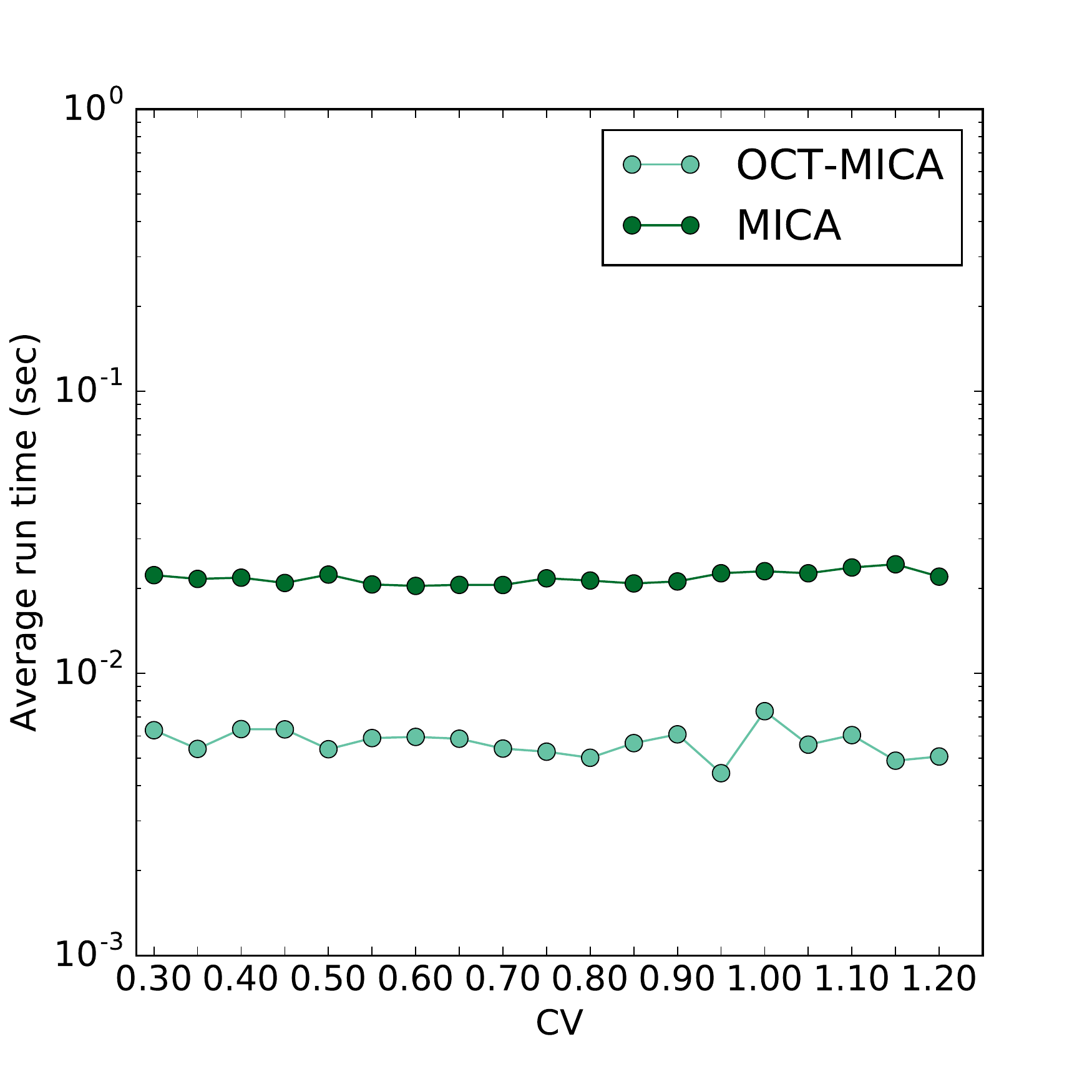}%
    \caption{\label{fig:app6}
	Runtimes of the MIB-enumerating (left) and MB-enumerating (right) algorithms on graphs where $n_B = 200$ and $n_O = 10$. The coefficient of variation between $L$ and $R$ was varied.
    }
\end{figure*}

\begin{figure*}[t!]
    \includegraphics [width=0.5\textwidth]{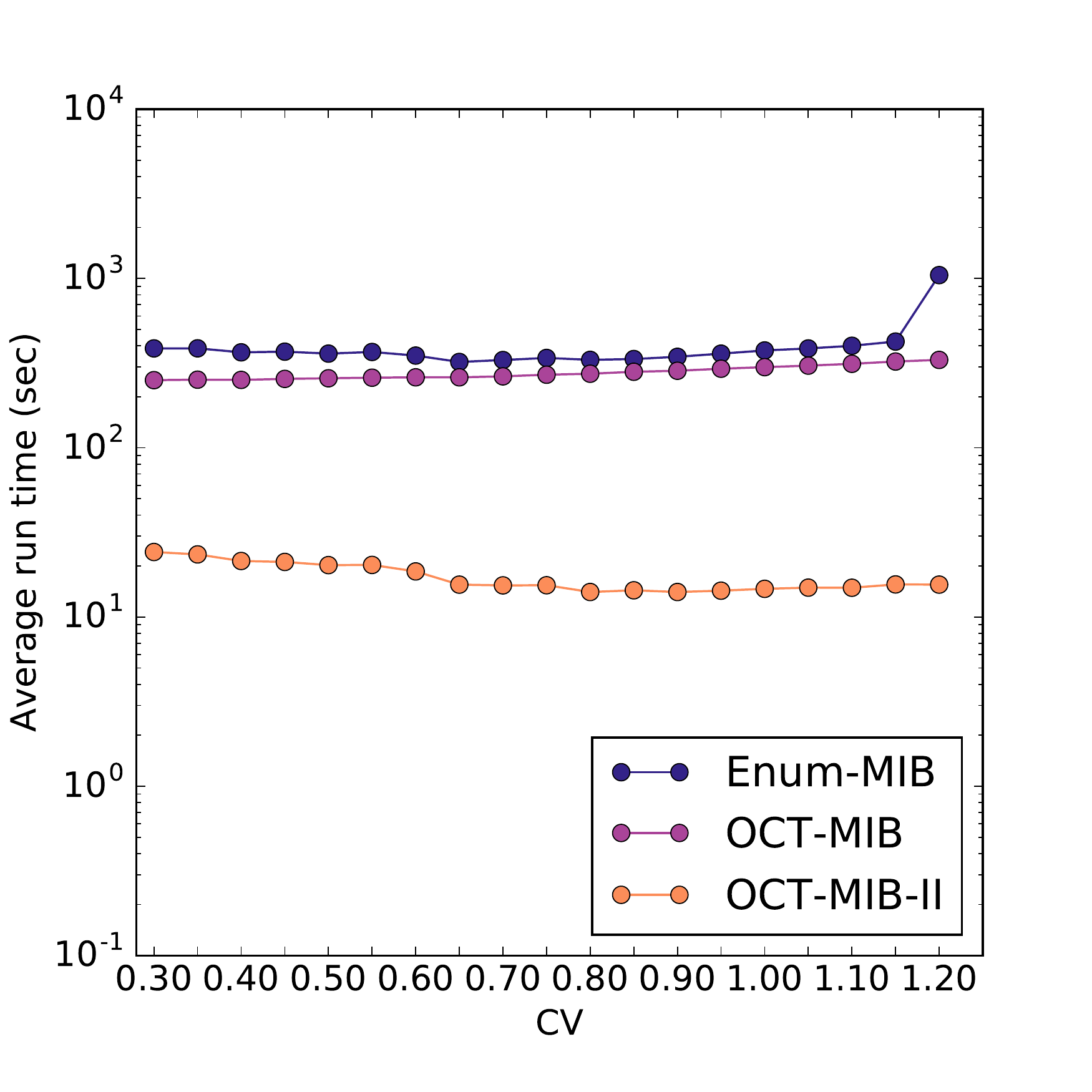}%
\hspace{0.01cm}
    \includegraphics [width=0.5\textwidth]{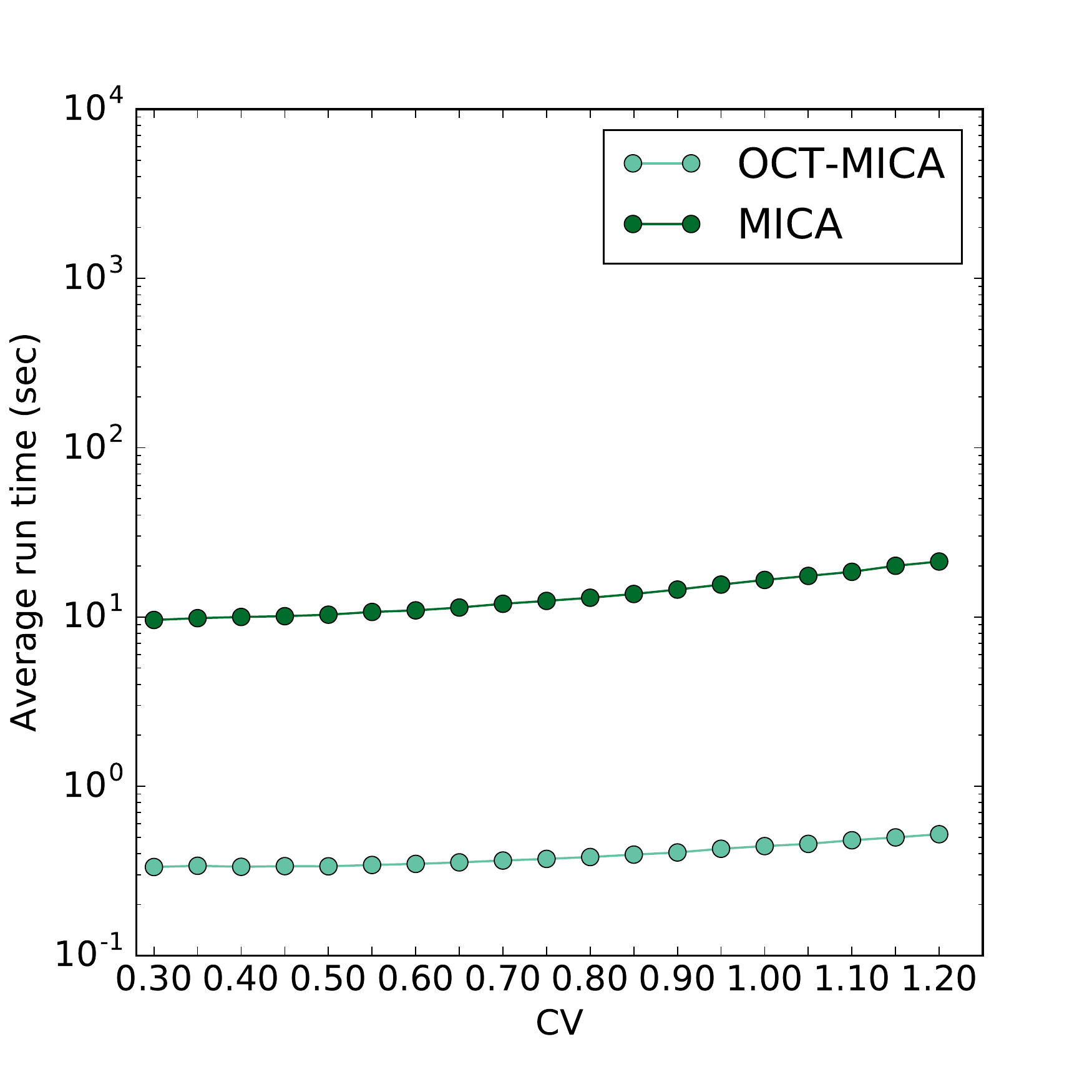}%
    \caption{\label{fig:app7}
	Runtimes of the MIB-enumerating (left) and MB-enumerating (right) algorithms on graphs where $n_B = 1000$ and $n_O = 19 \approx 3\log_3(n_B)$. The coefficient of variation between $L$ and $R$ was varied.
    }
\end{figure*}

\begin{figure*}[t!]
    \includegraphics [width=0.5\textwidth]{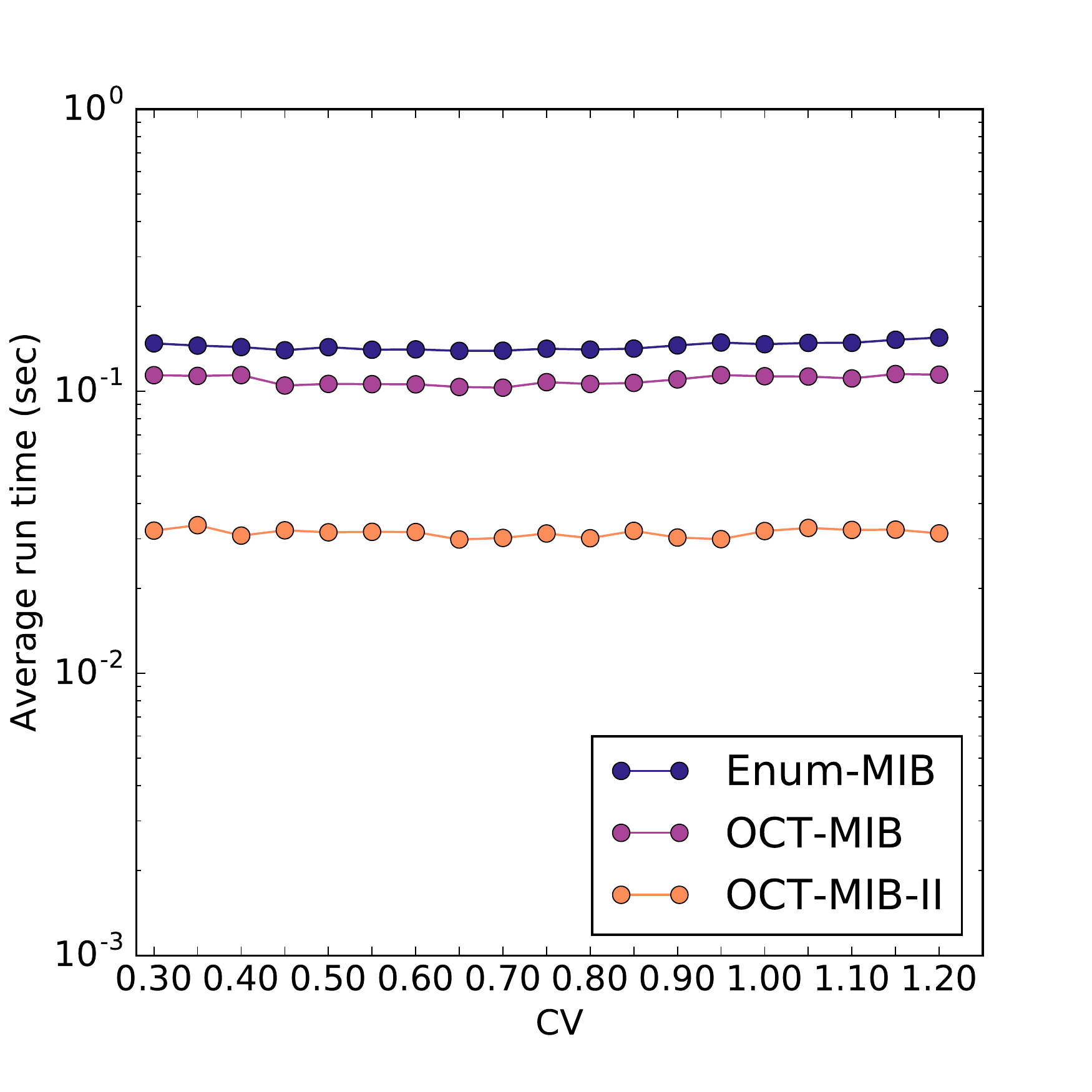}%
\hspace{0.01cm}
    \includegraphics [width=0.5\textwidth]{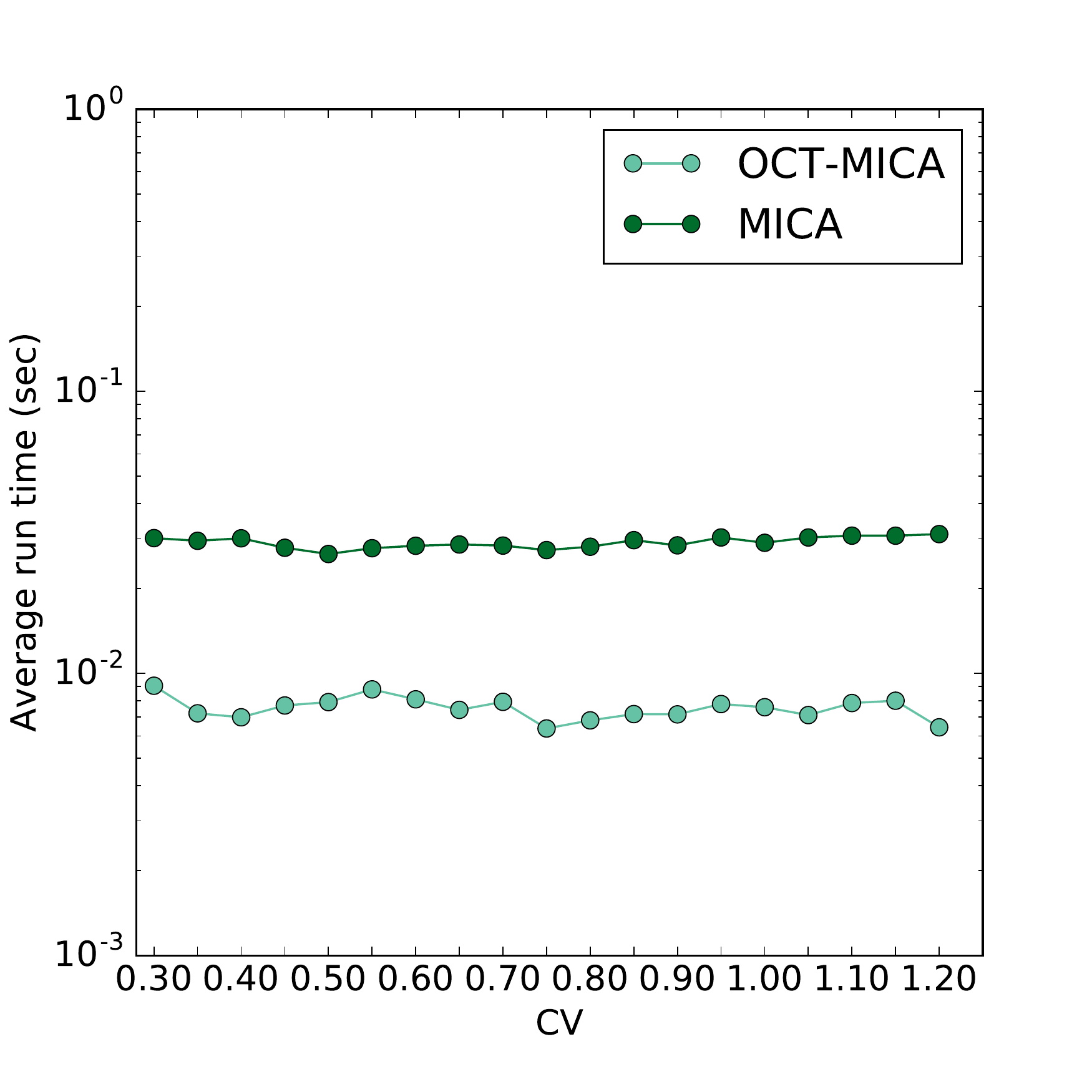}%
    \caption{\label{fig:app8}
	Runtimes of the MIB-enumerating (left) and MB-enumerating (right) algorithms on graphs where $n_B = 200$ and $n_O = 14 \approx 3\log_3(n_B)$. The coefficient of variation between $L$ and $R$ was varied.
    }
\end{figure*}

\begin{figure*}[t!]
    \includegraphics [width=0.5\textwidth]{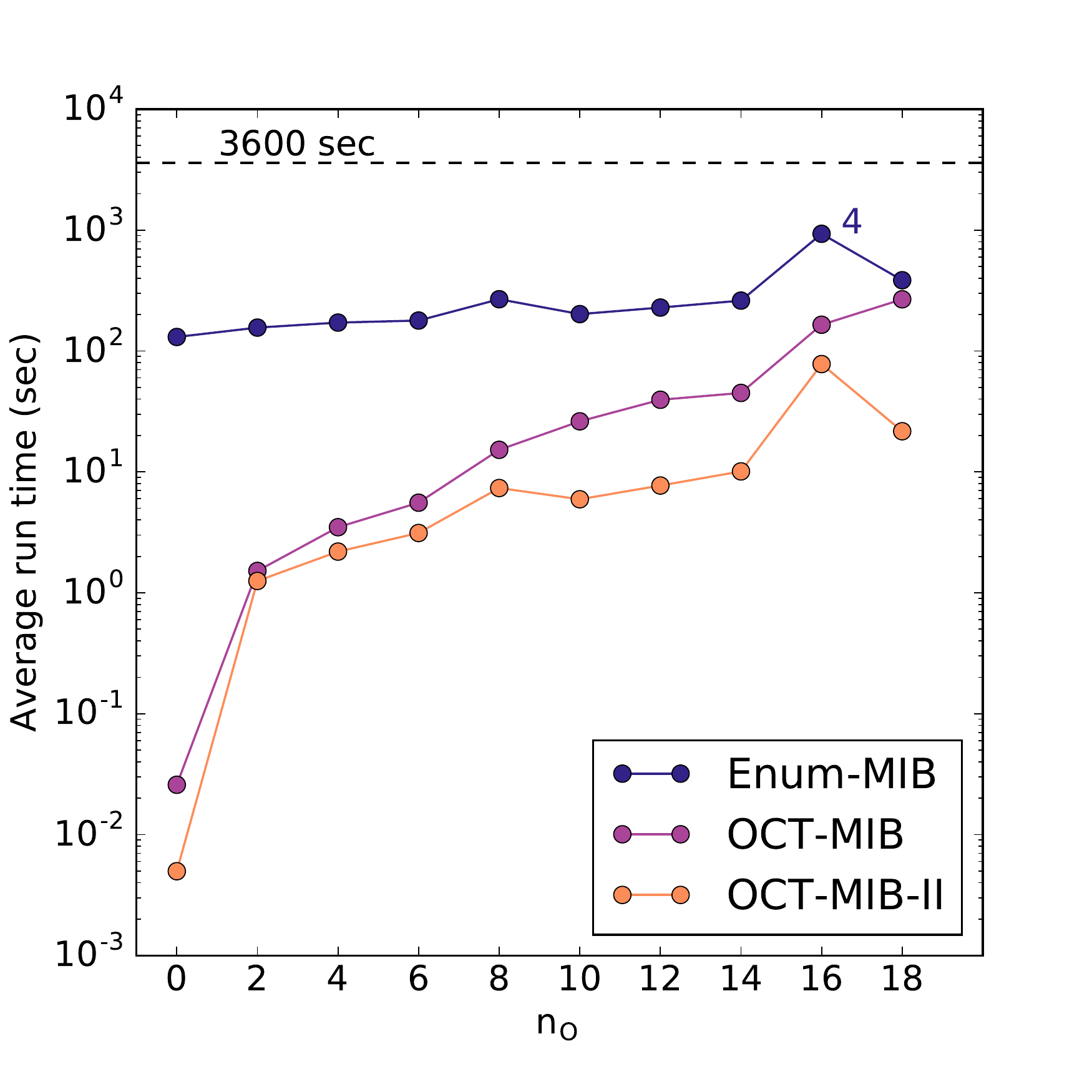}%
\hspace{0.01cm}
    \includegraphics [width=0.5\textwidth]{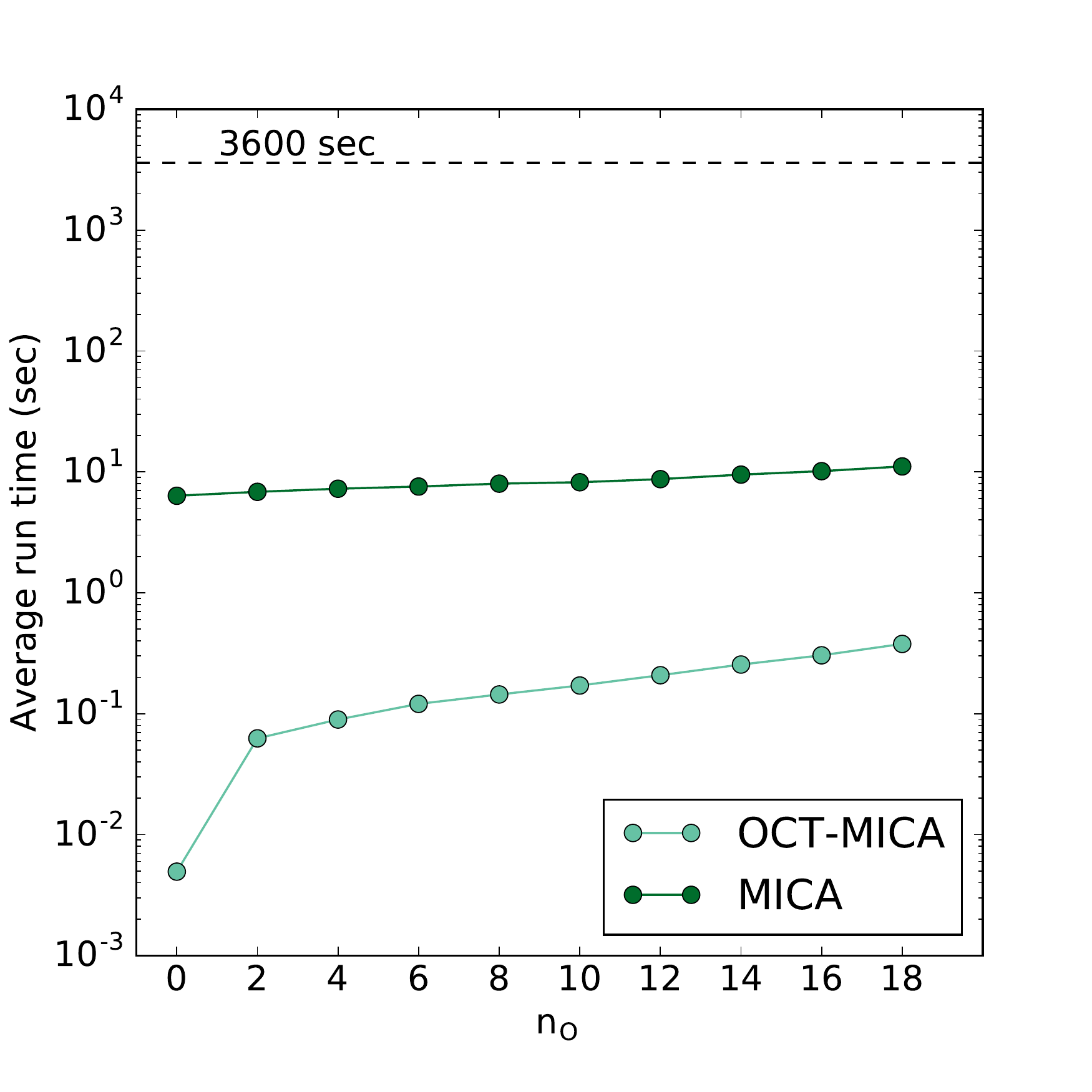}%
    \caption{\label{fig:app}
	Runtimes of the MIB-enumerating (left) and MB-enumerating (right) algorithms on graphs where $n_B = 1000$ and $n_O$ was varied.
    }
\end{figure*}

\begin{figure*}[t!]
    \includegraphics [width=0.5\textwidth]{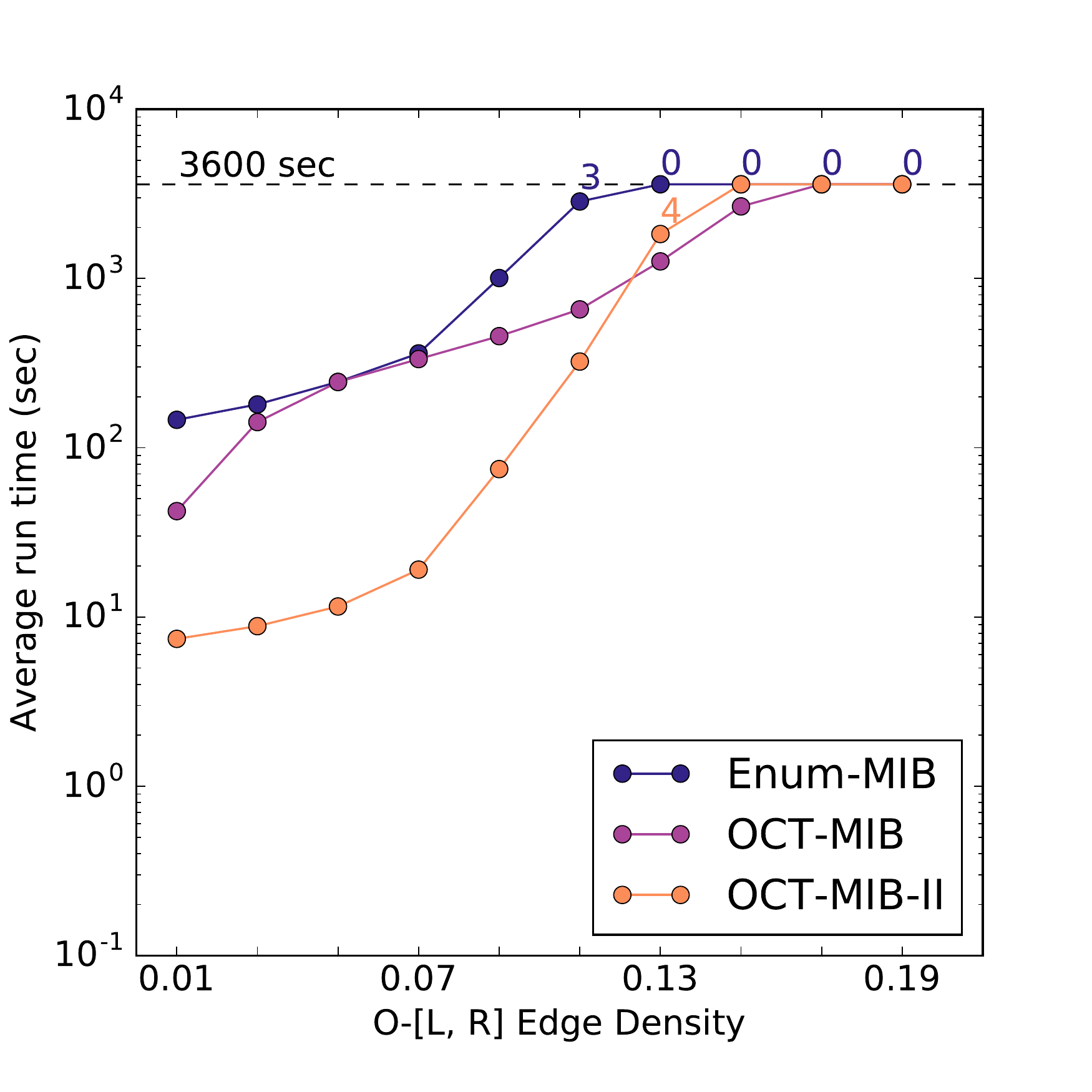}%
\hspace{0.01cm}
    \includegraphics [width=0.5\textwidth]{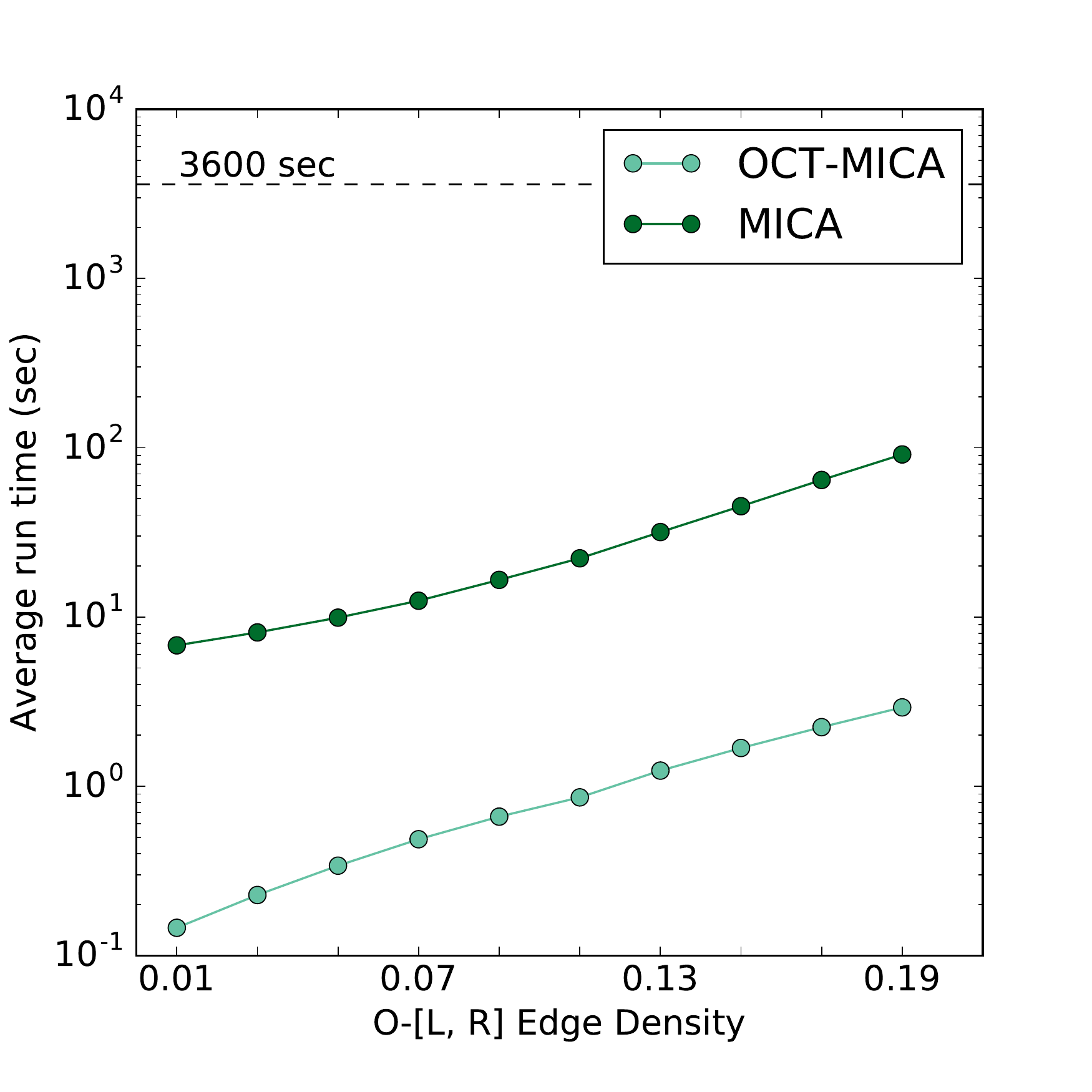}%
    \caption{\label{fig:app9}
	Runtimes of the MIB-enumerating (left) and MB-enumerating (right) algorithms on graphs where $n_B = 1000$ and $n_O = 19 \approx 3\log_3(n_B)$. The expected edge density between $O$ and $\{L,R\}$ was varied.
    }
\end{figure*}

\begin{figure*}[t!]
    \includegraphics [width=0.5\textwidth]{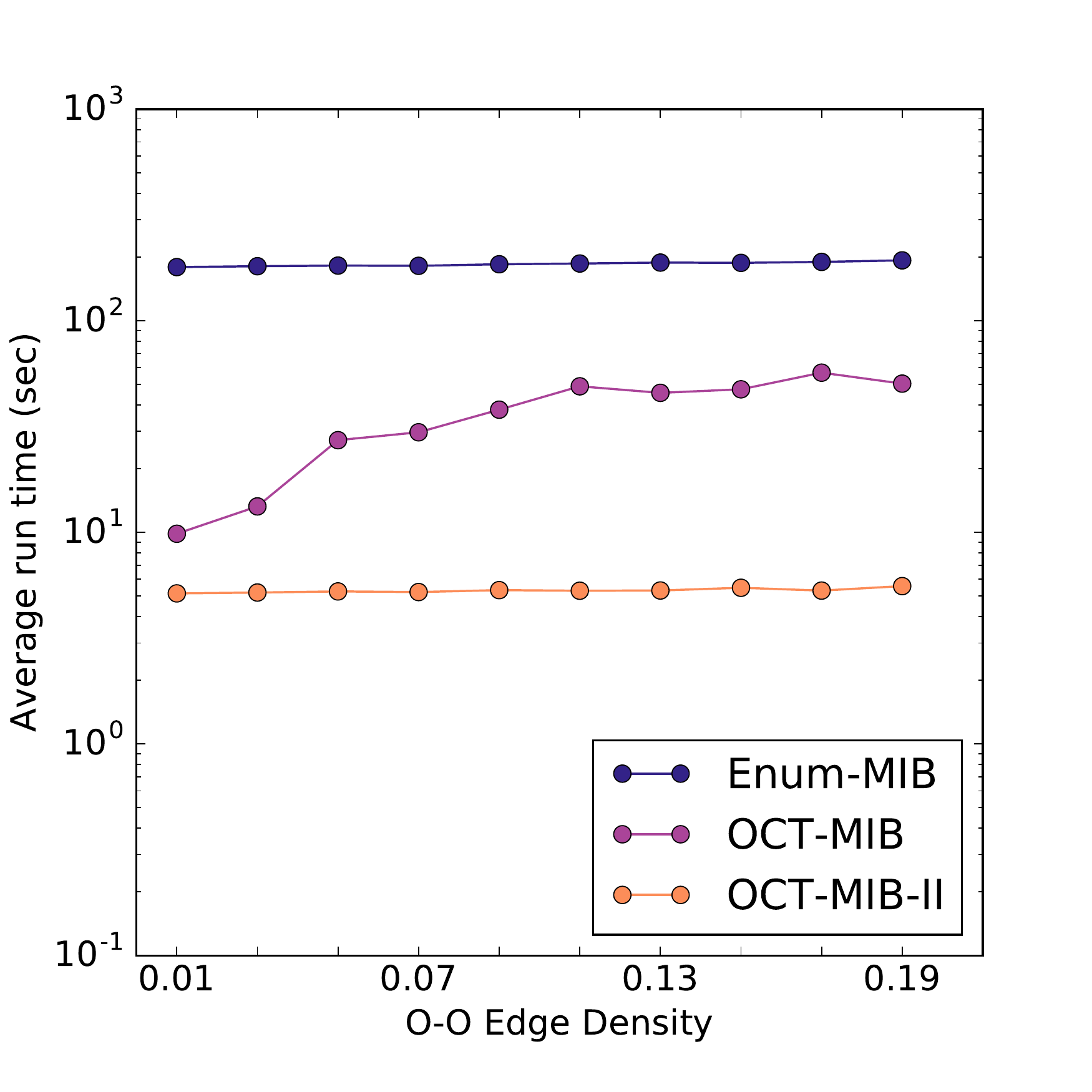}%
\hspace{0.01cm}
    \includegraphics [width=0.5\textwidth]{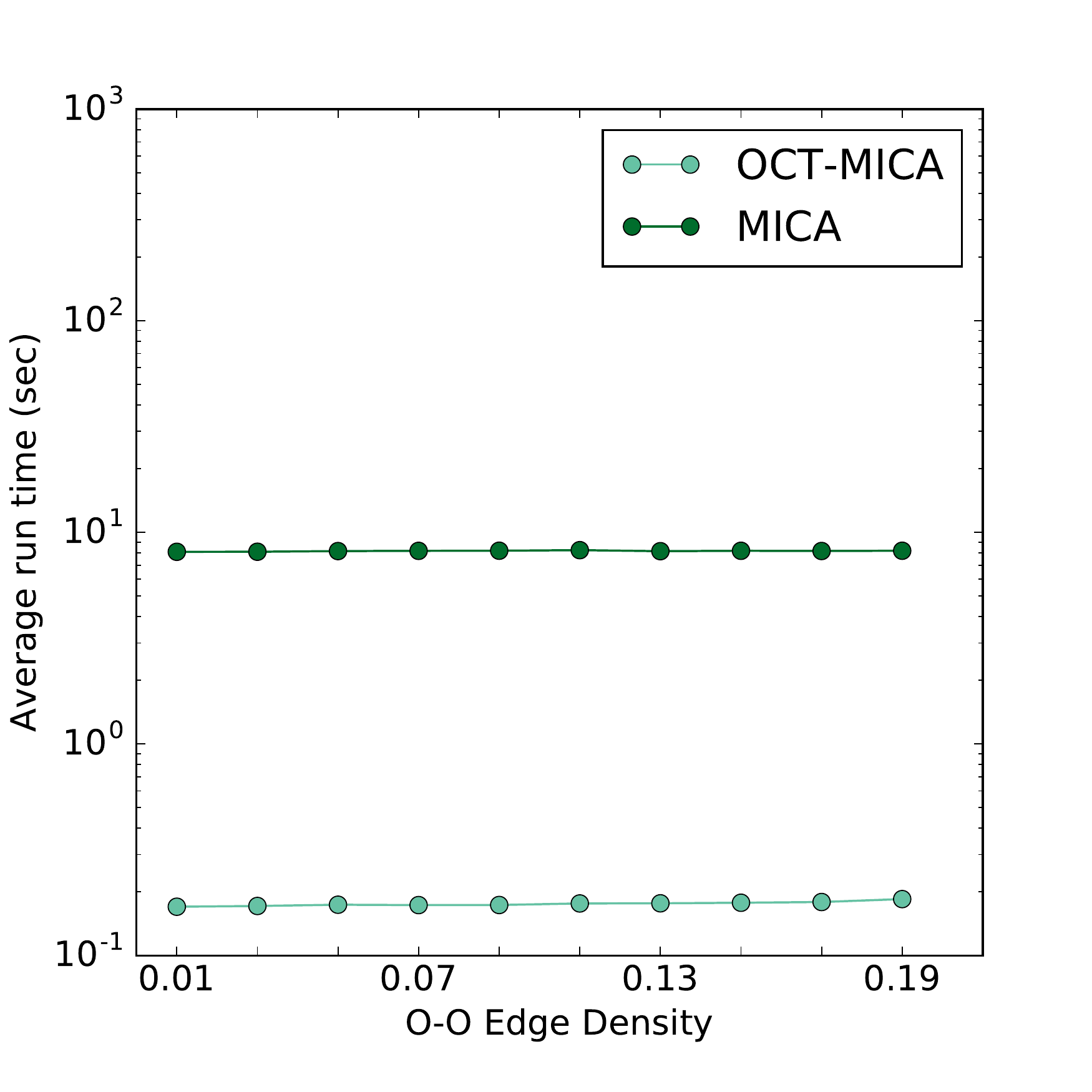}%
    \caption{\label{fig:app10}
	Runtimes of the MIB-enumerating (left) and MB-enumerating (right) algorithms on graphs where $n_B = 1000$ and $n_O = 10$. The expected edge density within $O$ was varied.
    }
\end{figure*}

\begin{figure*}[t!]
    \includegraphics [width=0.5\textwidth]{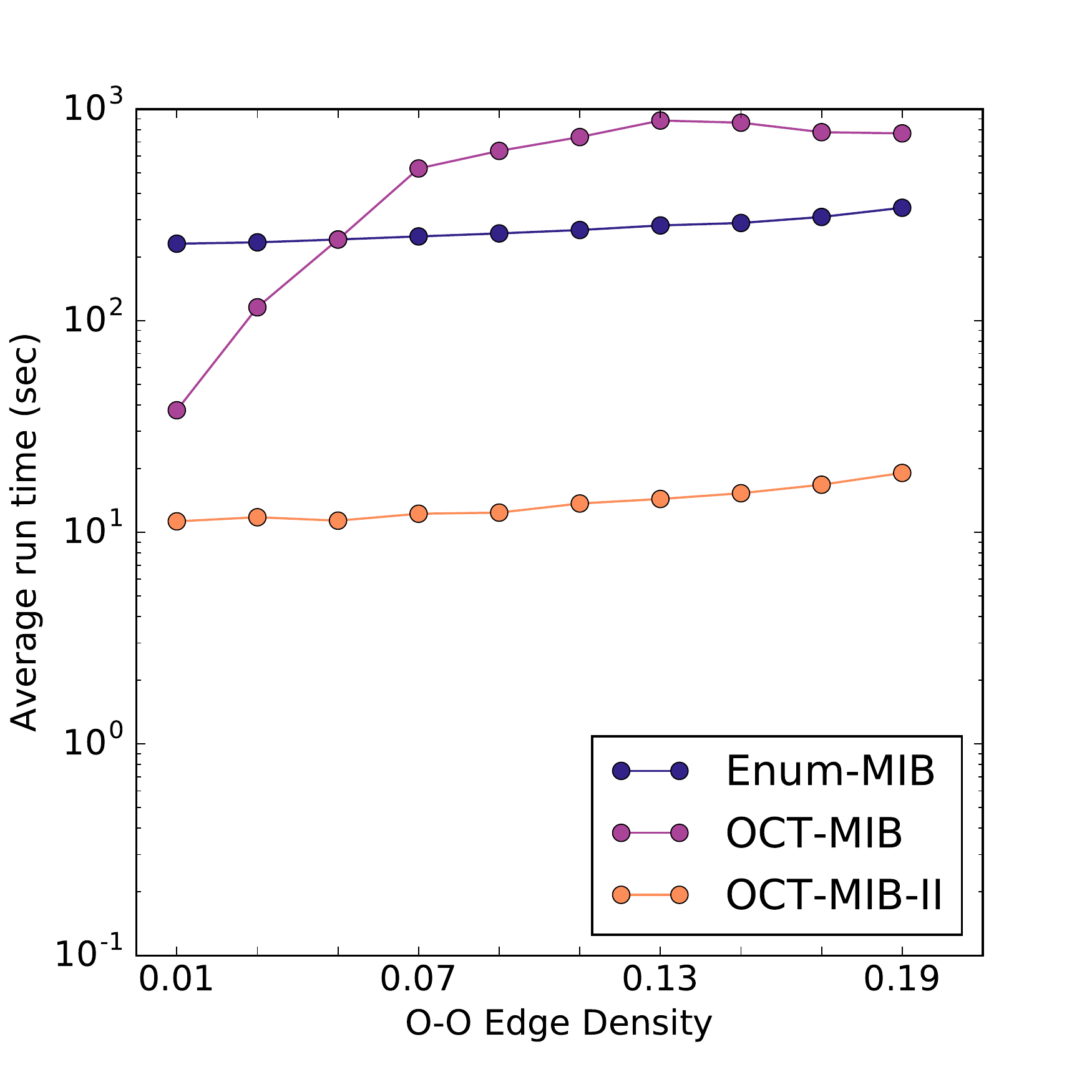}%
\hspace{0.01cm}
    \includegraphics [width=0.5\textwidth]{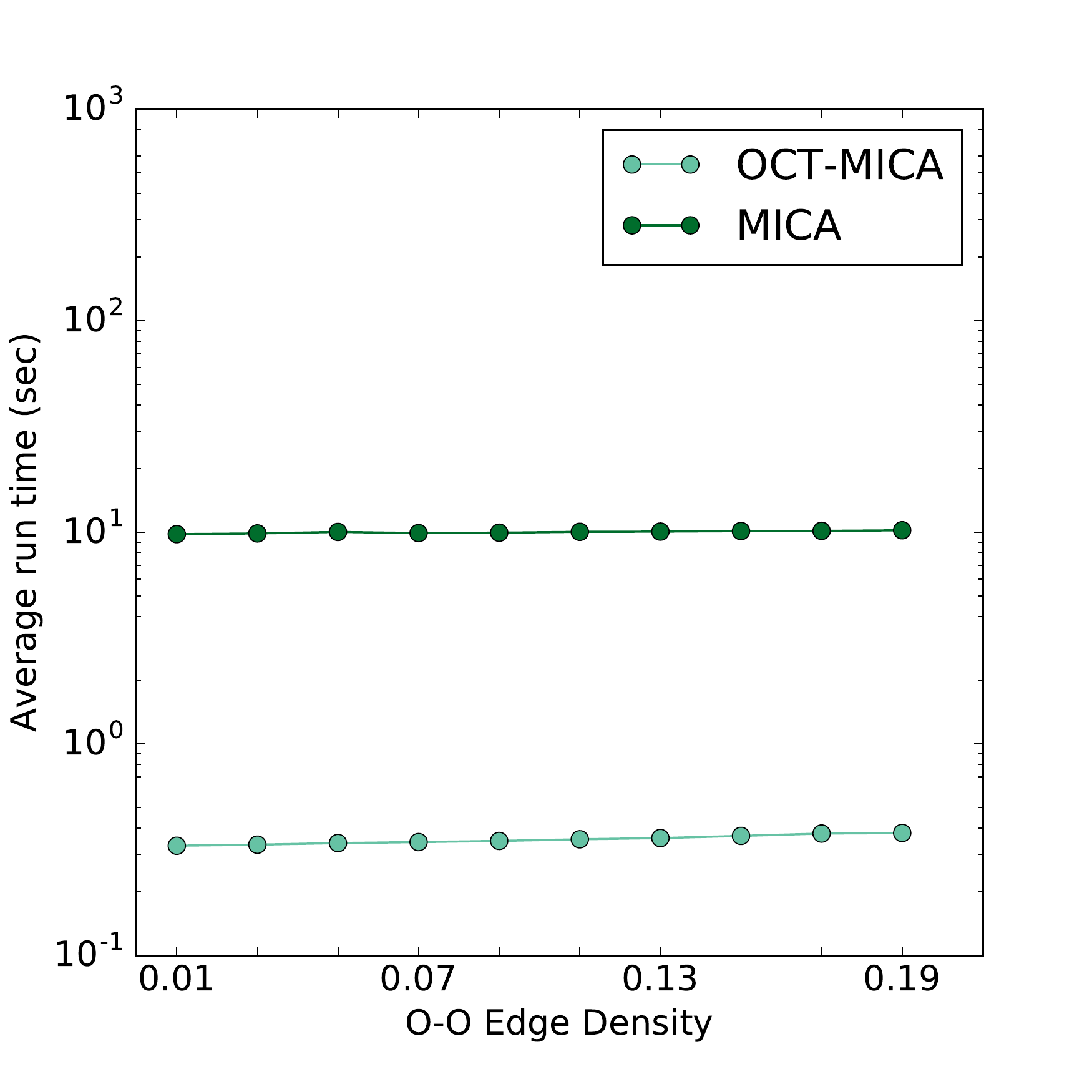}%
    \caption{\label{fig:app11}
	Runtimes of the MIB-enumerating (left) and MB-enumerating (right) algorithms on graphs where $n_B = 1000$ and $n_O = 19 \approx 3\log_3(n_B)$. The expected edge density within $O$ was varied.
    }
\end{figure*}

\begin{figure*}[t!]
    \includegraphics [width=0.5\textwidth]{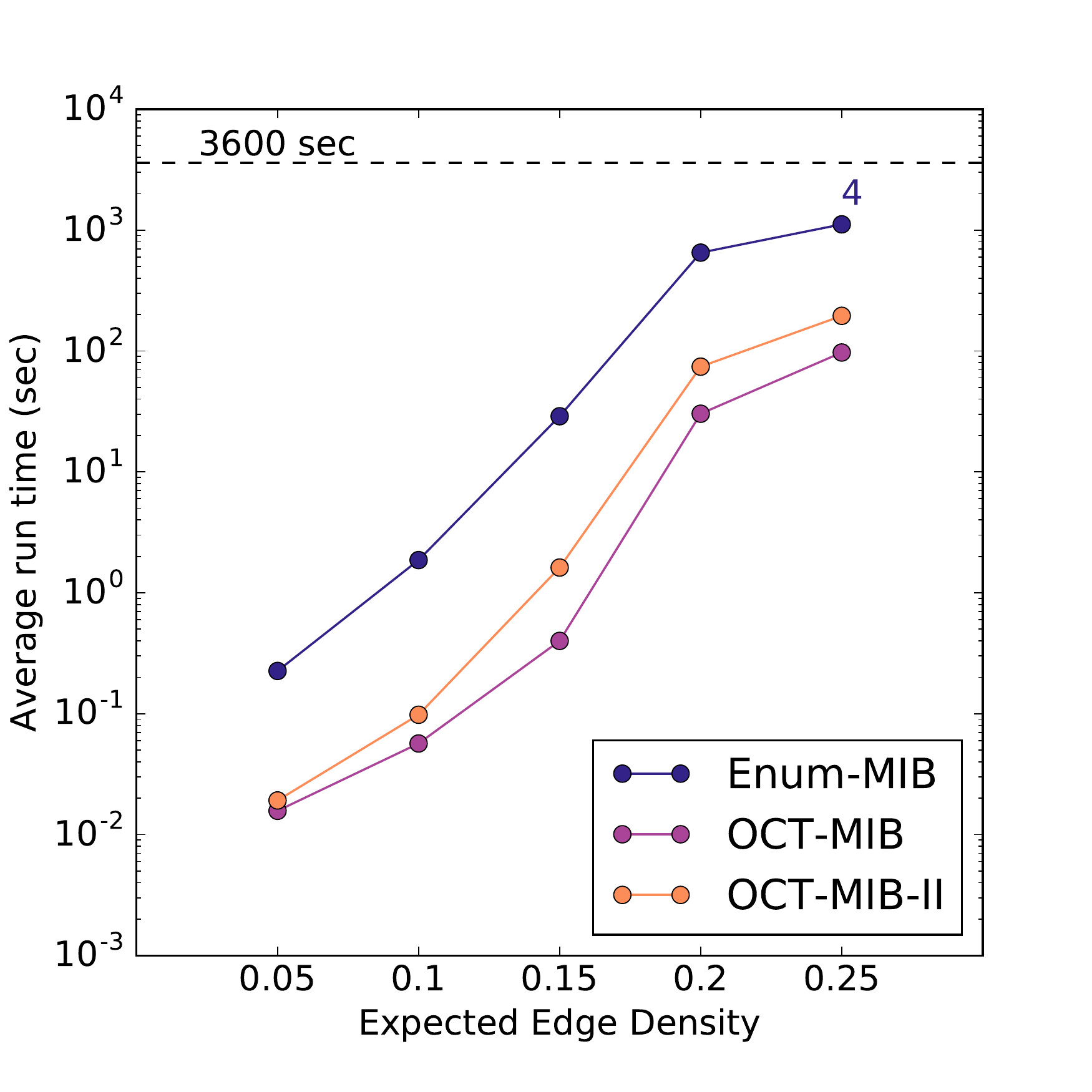}%
\hspace{0.01cm}
    \includegraphics [width=0.5\textwidth]{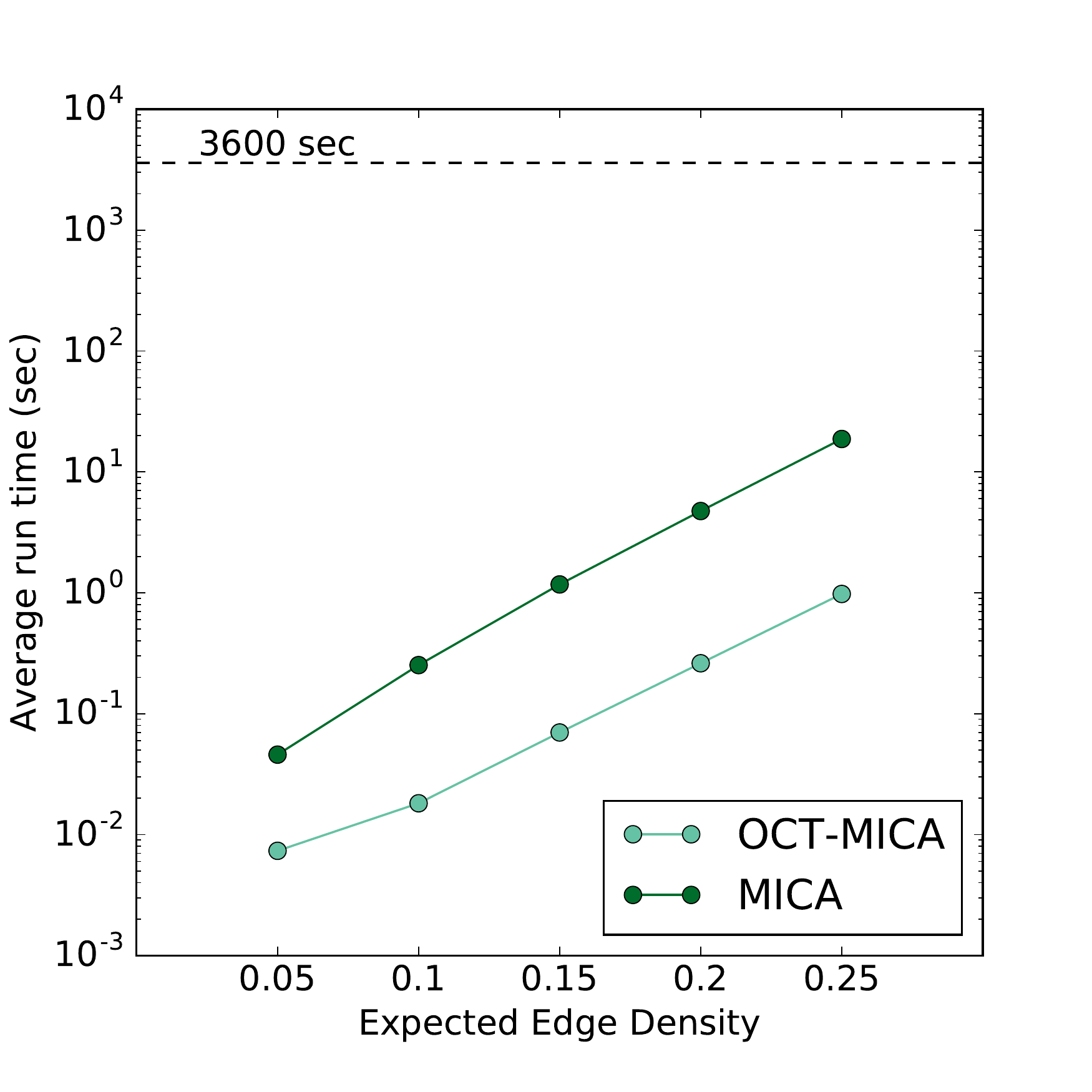}%
    \caption{\label{fig:app12}
	Runtimes of the MIB-enumerating (left) and MB-enumerating (right) algorithms on graphs where $n_B = 150$, $n_L = n_R$ and $n_O = 5$. The expected edge density in the graph was varied except for the expected edge density within $O$ which was fixed to $0.05$.
    }
\end{figure*}

\begin{figure*}[t!]
    \includegraphics [width=0.5\textwidth]{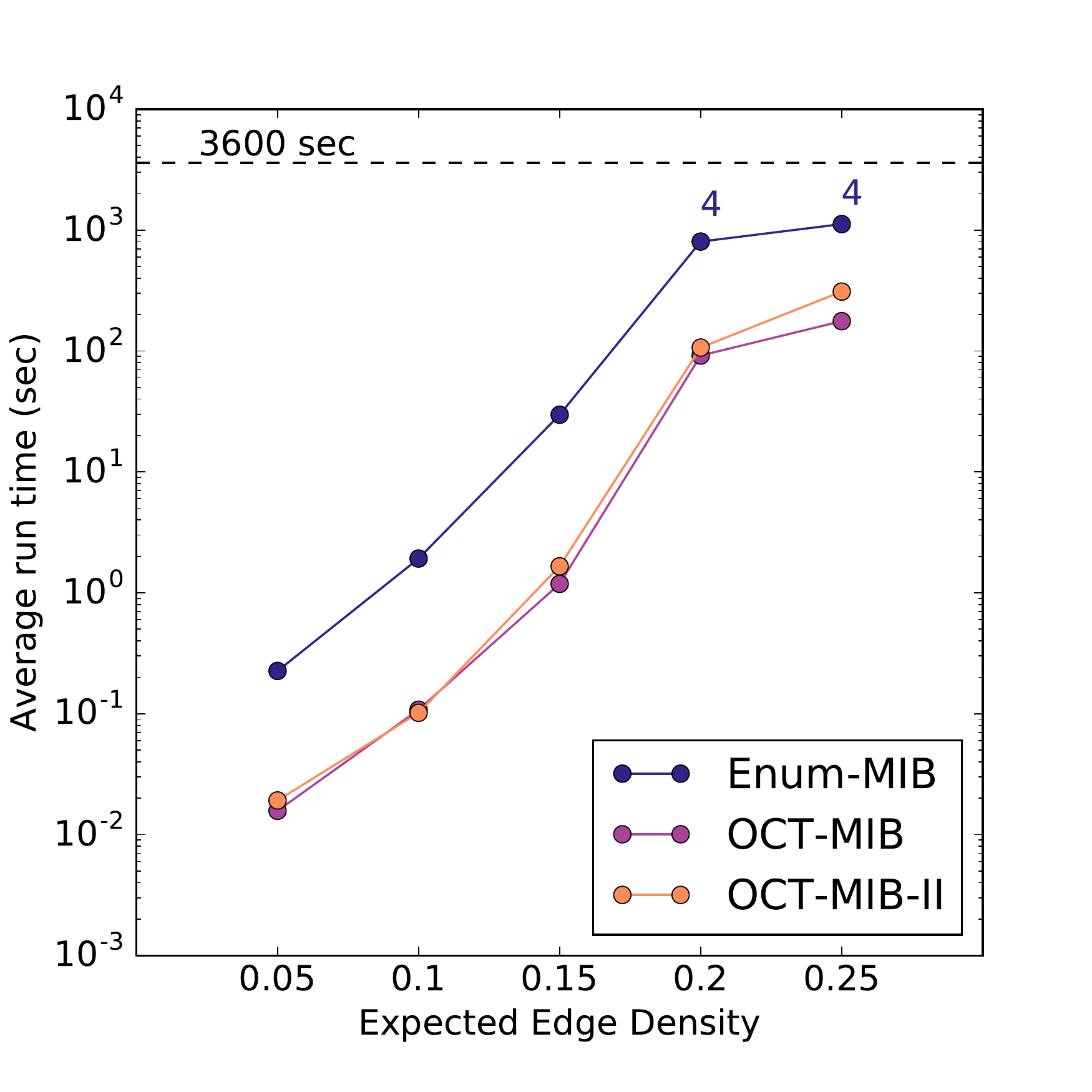}%
\hspace{0.01cm}
    \includegraphics [width=0.5\textwidth]{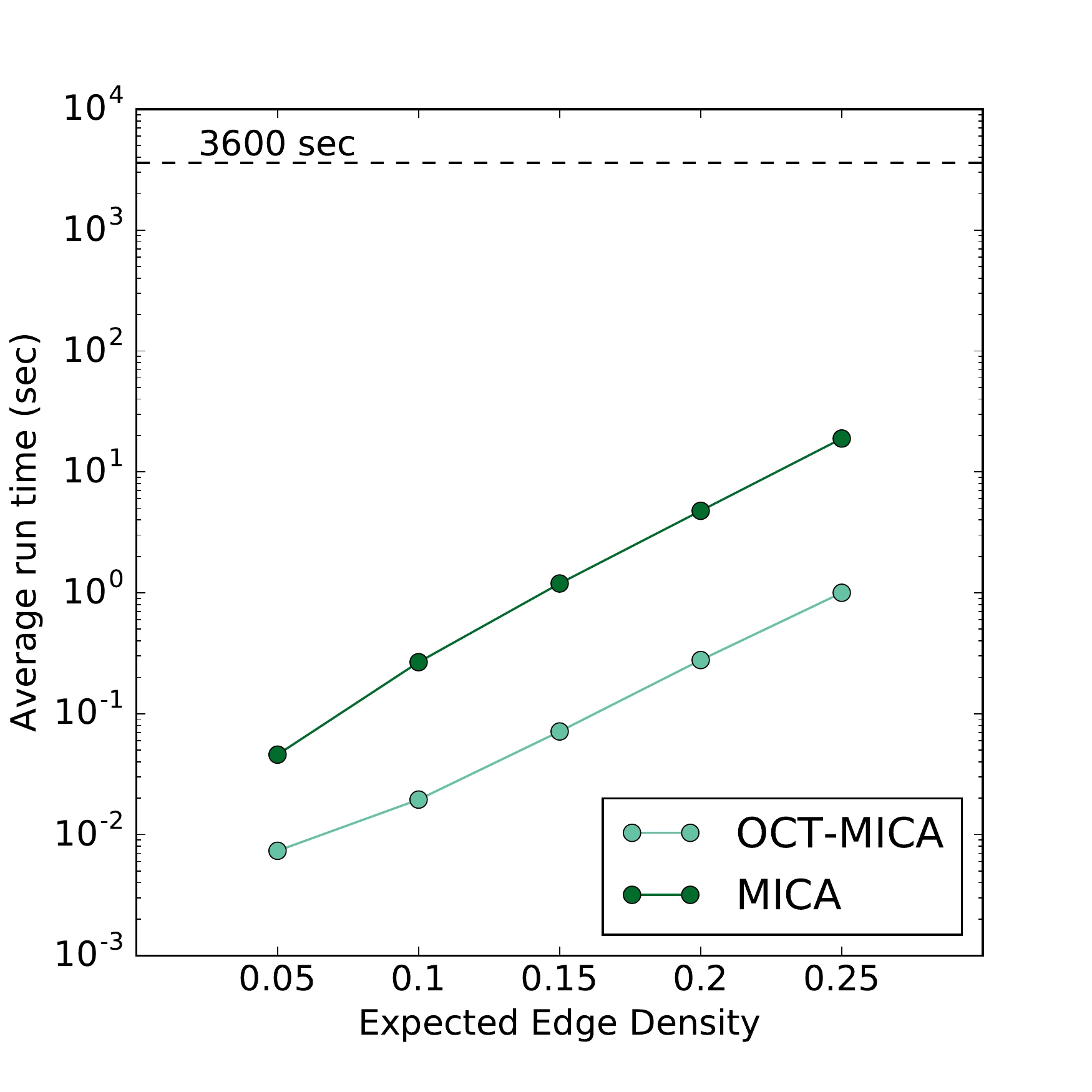}%
    \caption{\label{fig:app13}
	Runtimes of the MIB-enumerating (left) and MB-enumerating (right) algorithms on graphs where $n_B = 150$, $n_L = n_R$ and $n_O = 5$. The expected edge density in the graph was varied, including the expected edge density within $O$.
    }
\end{figure*}

\begin{figure*}[t!]
    \includegraphics [width=0.5\textwidth]{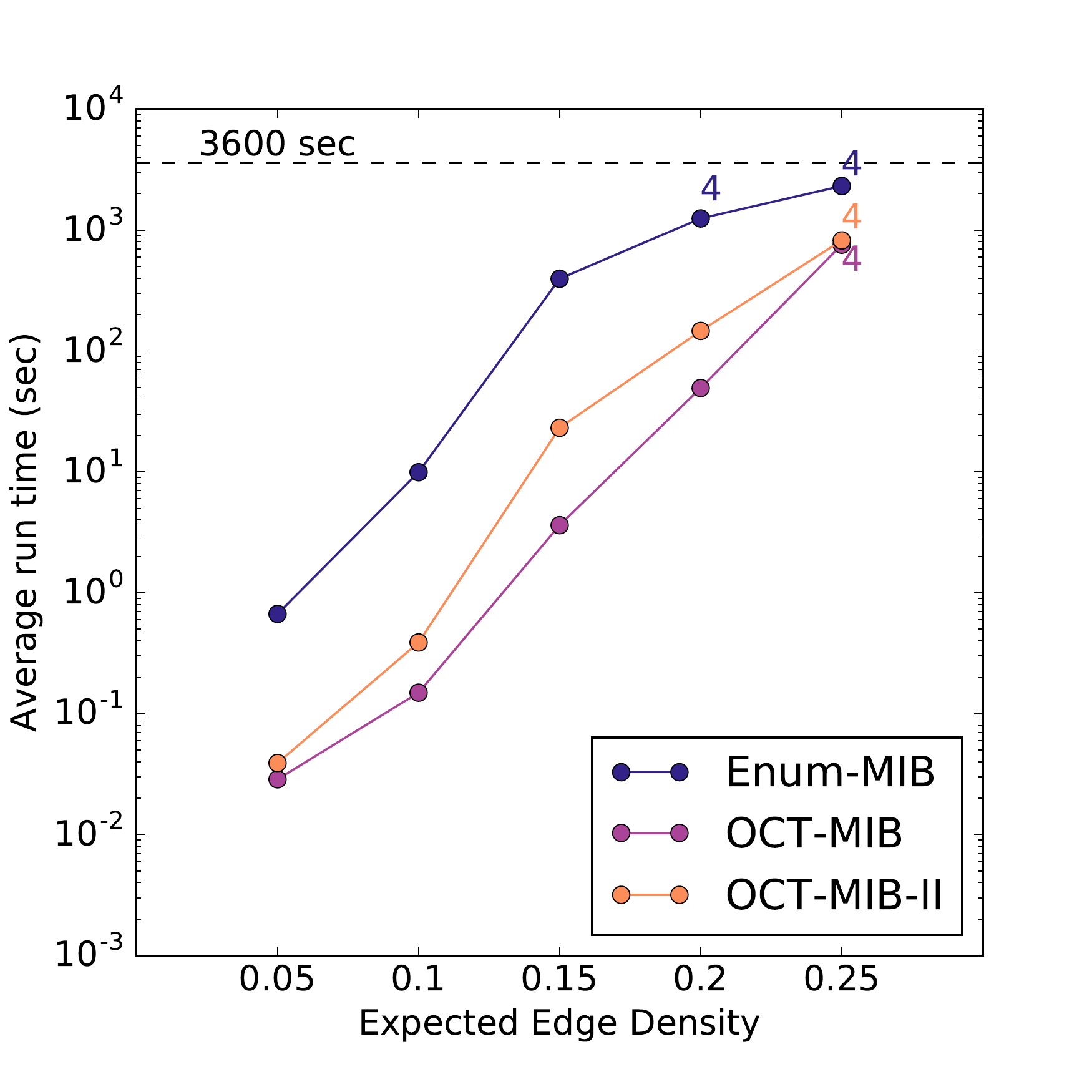}%
\hspace{0.01cm}
    \includegraphics [width=0.5\textwidth]{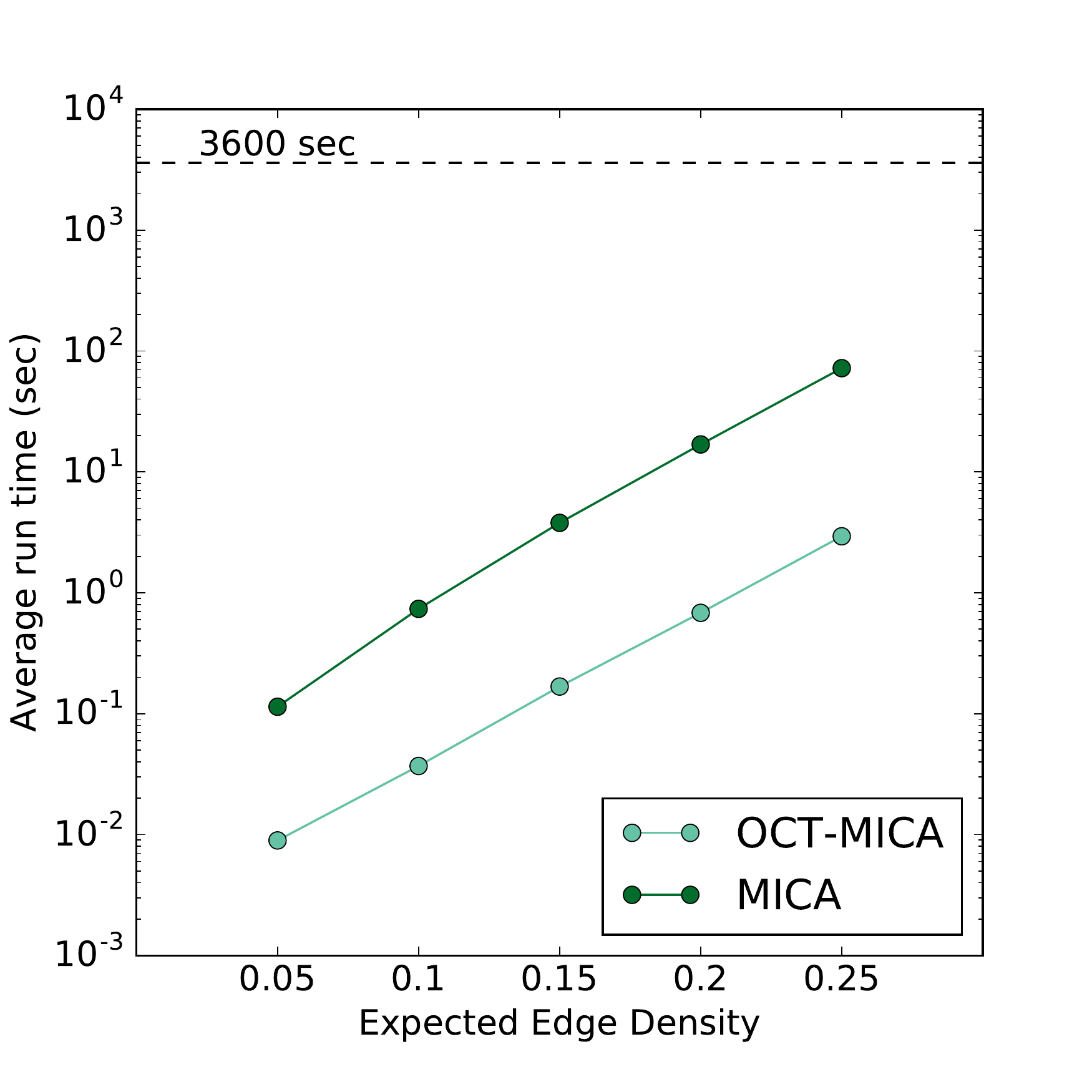}%
    \caption{\label{fig:app14}
	Runtimes of the MIB-enumerating (left) and MB-enumerating (right) algorithms on graphs where $n_B = 200$, $n_L = n_R$ and $n_O = 5$. The expected edge density in the graph was varied except for the expected edge density within $O$ which was fixed to $0.05$.
    }
\end{figure*}

\begin{figure*}[t!]
    \includegraphics [width=0.5\textwidth]{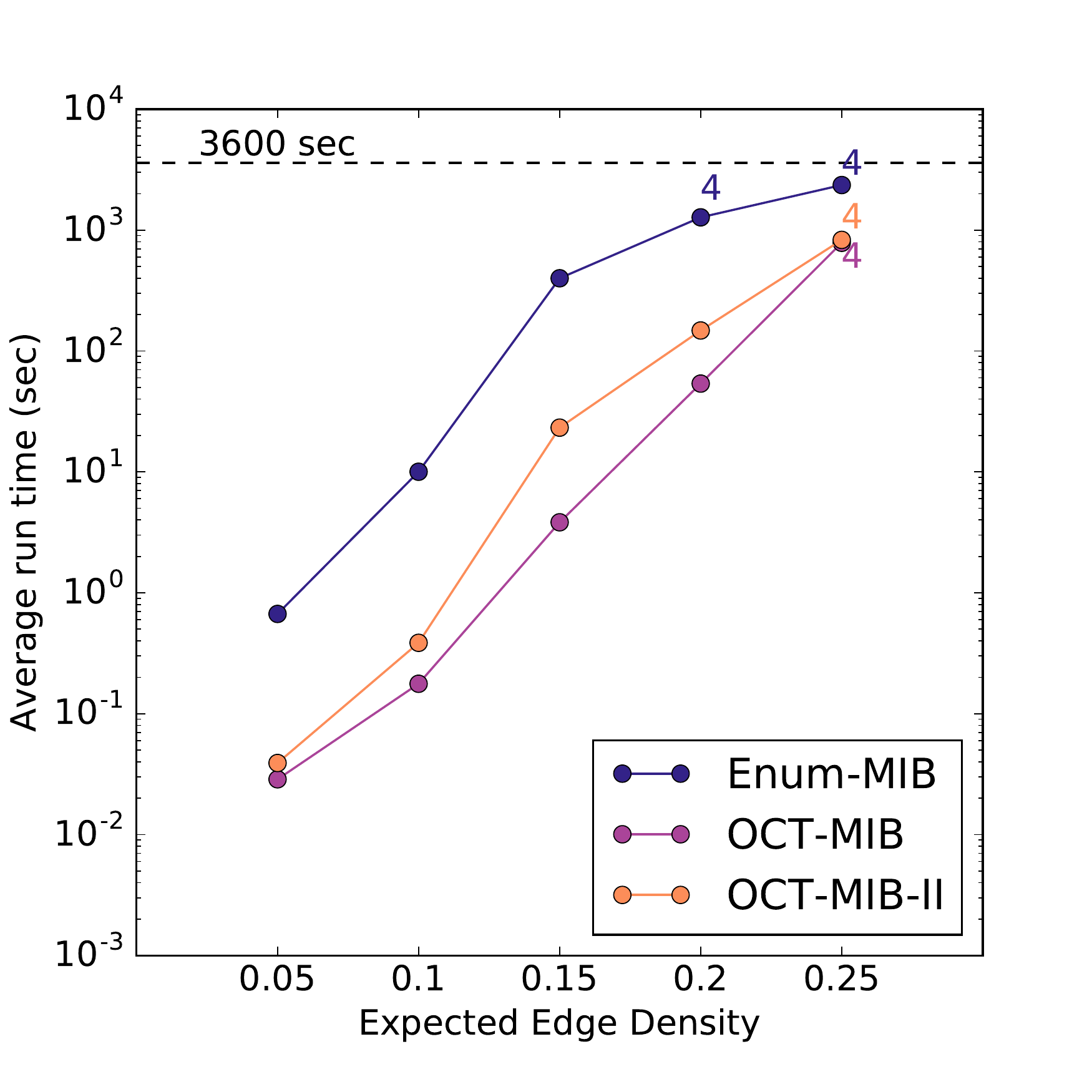}%
\hspace{0.01cm}
    \includegraphics [width=0.5\textwidth]{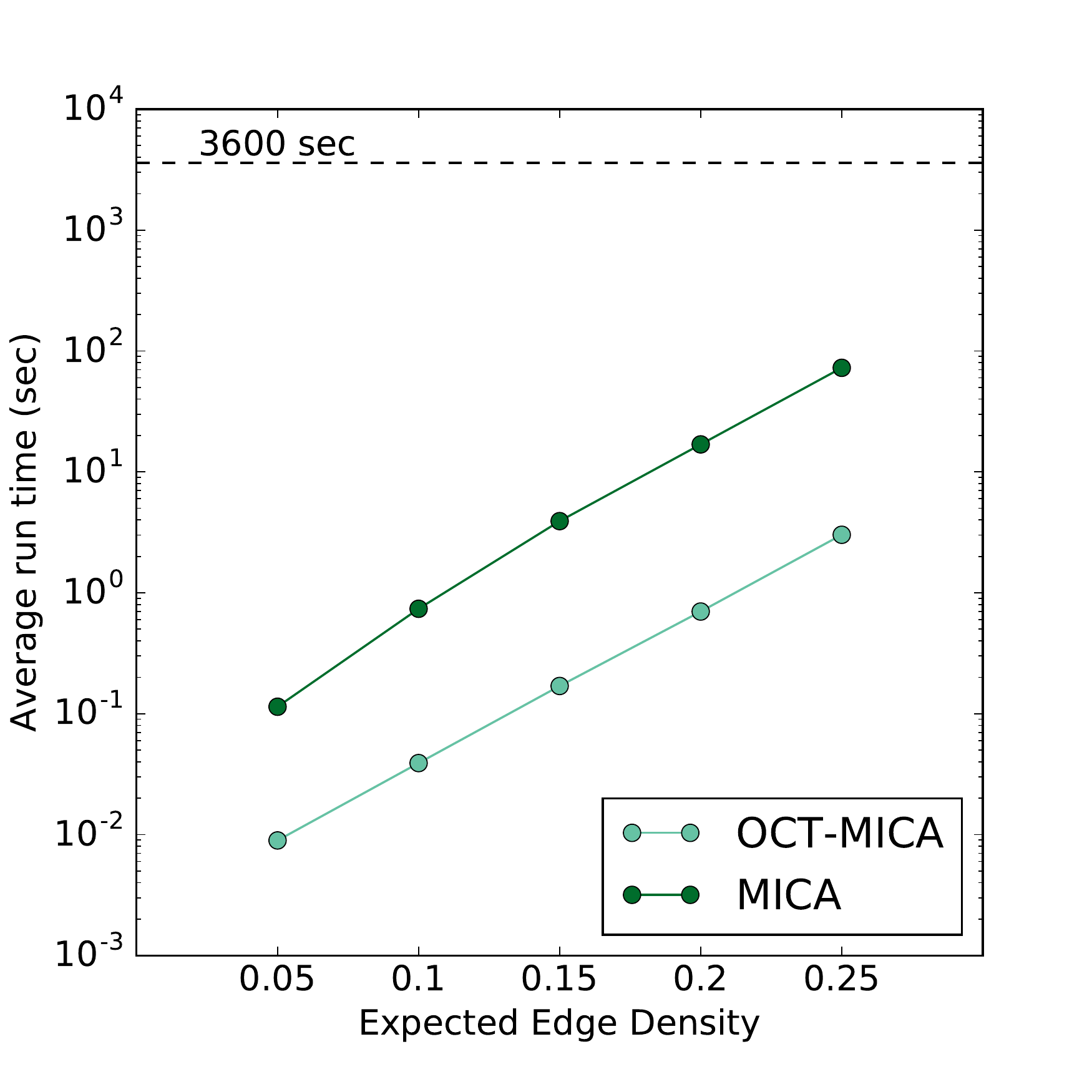}%
    \caption{\label{fig:app15}
	Runtimes of the MIB-enumerating (left) and MB-enumerating (right) algorithms on graphs where $n_B = 200$, $n_L = n_R$ and $n_O = 5$. The expected edge density in the graph was varied, including the expected edge density within $O$.
    }
\end{figure*}

\clearpage

\begin{figure*}[t!]
    \includegraphics [width=0.5\textwidth]{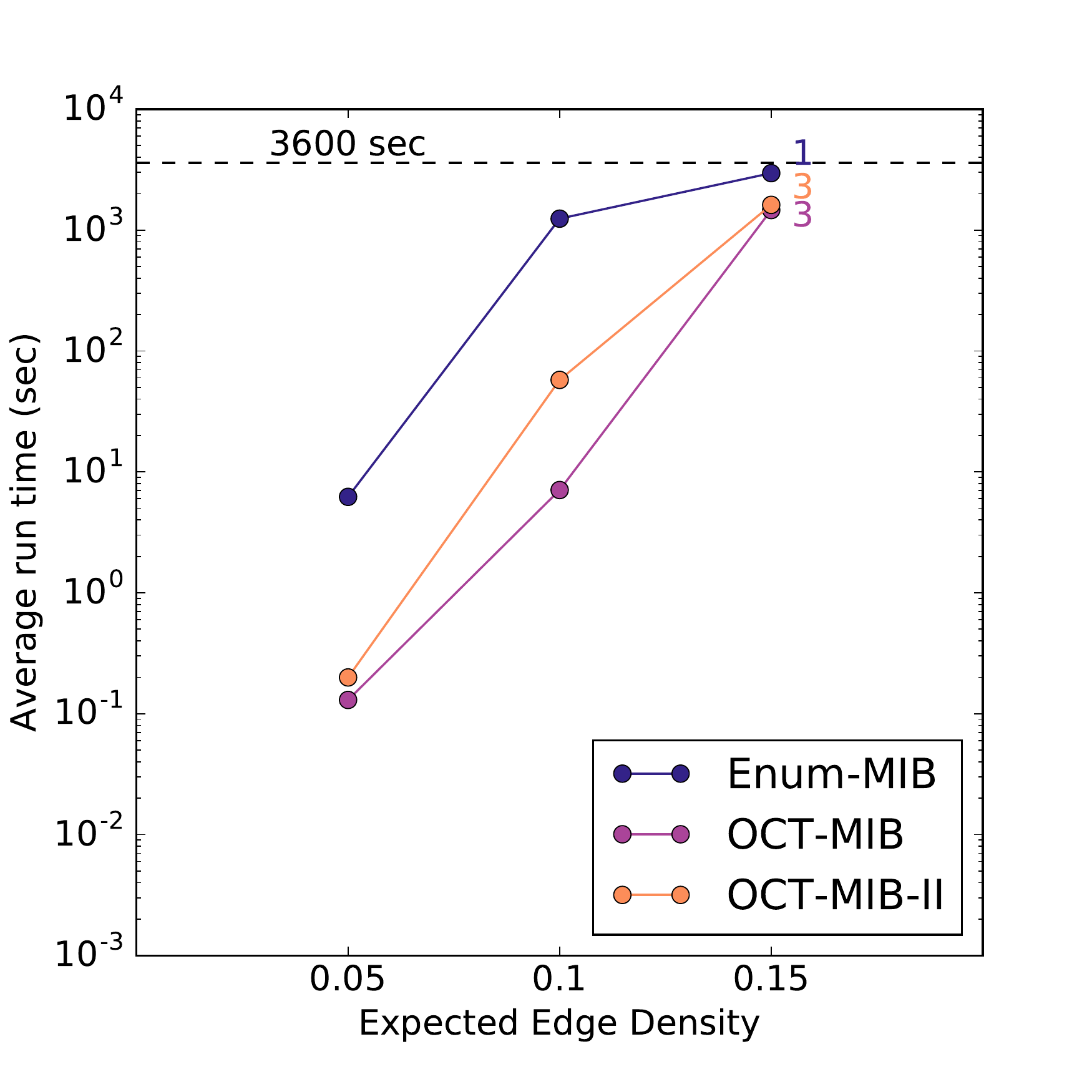}%
\hspace{0.01cm}
    \includegraphics [width=0.5\textwidth]{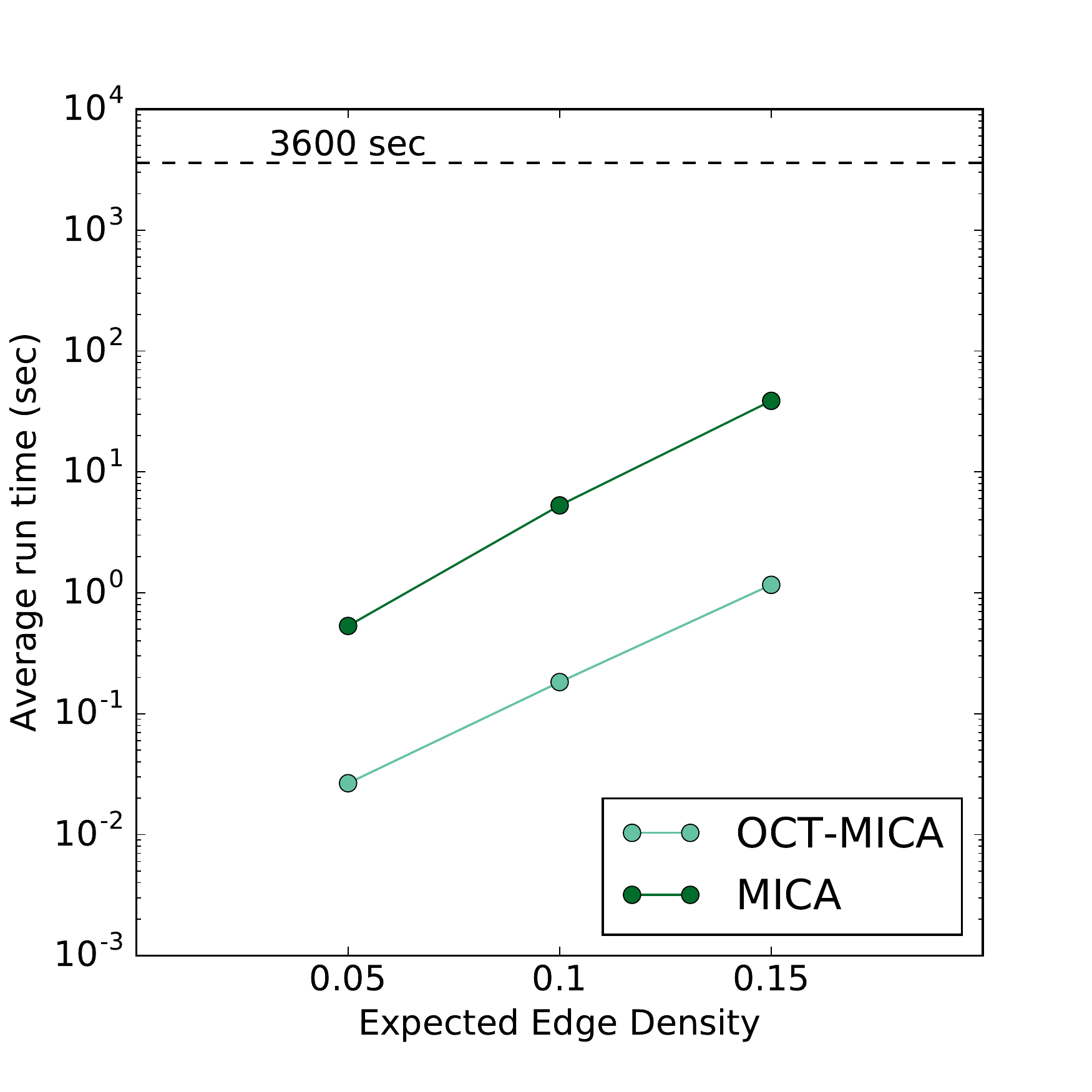}%
    \caption{\label{fig:app16}
	Runtimes of the MIB-enumerating (left) and MB-enumerating (right) algorithms on graphs where $n_B = 300$, $n_L = n_R$ and $n_O = 5$. The expected edge density in the graph was varied except for the expected edge density within $O$ which was fixed to $0.05$.
    }
\end{figure*}

\begin{figure*}[t!]
    \includegraphics [width=0.5\textwidth]{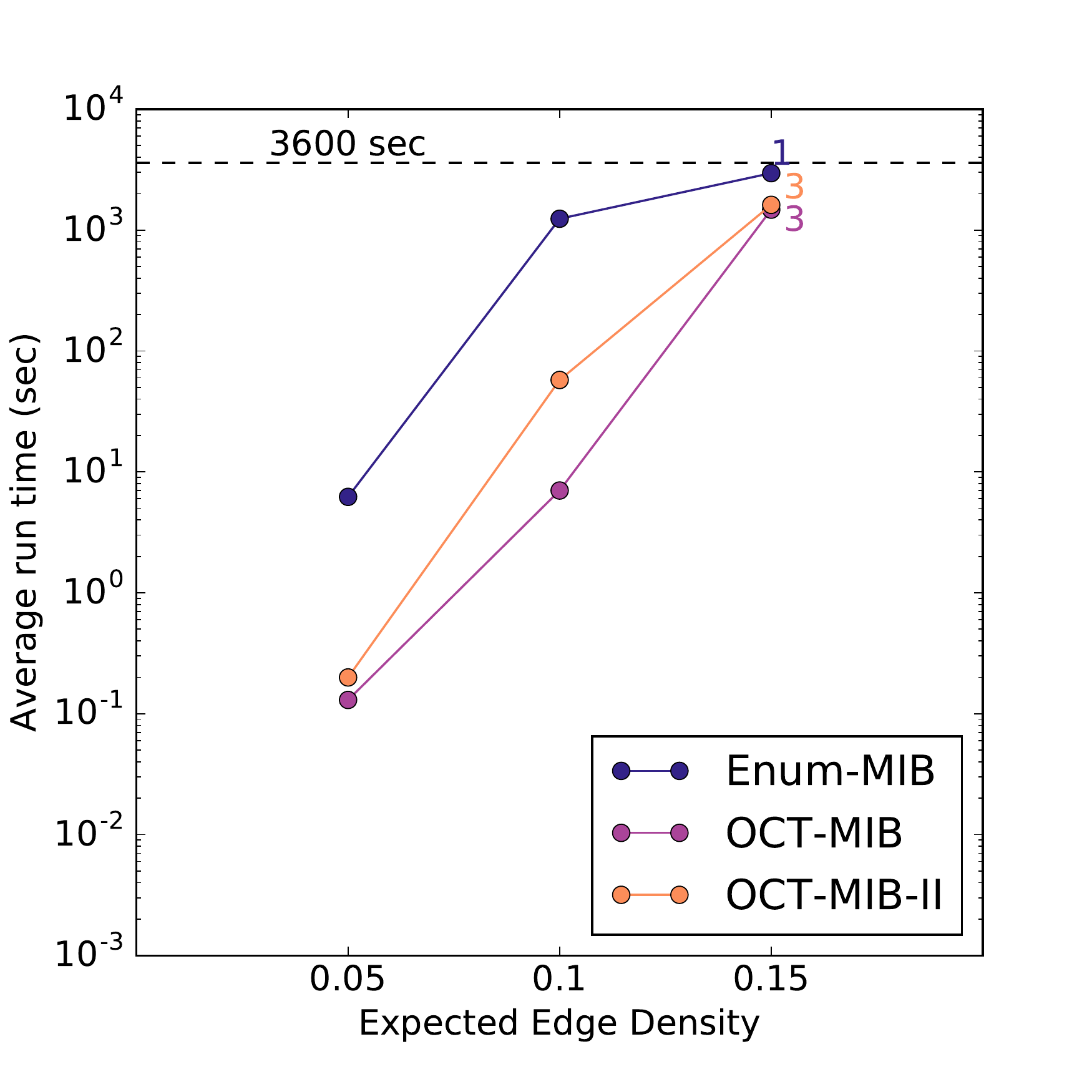}%
\hspace{0.01cm}
    \includegraphics [width=0.5\textwidth]{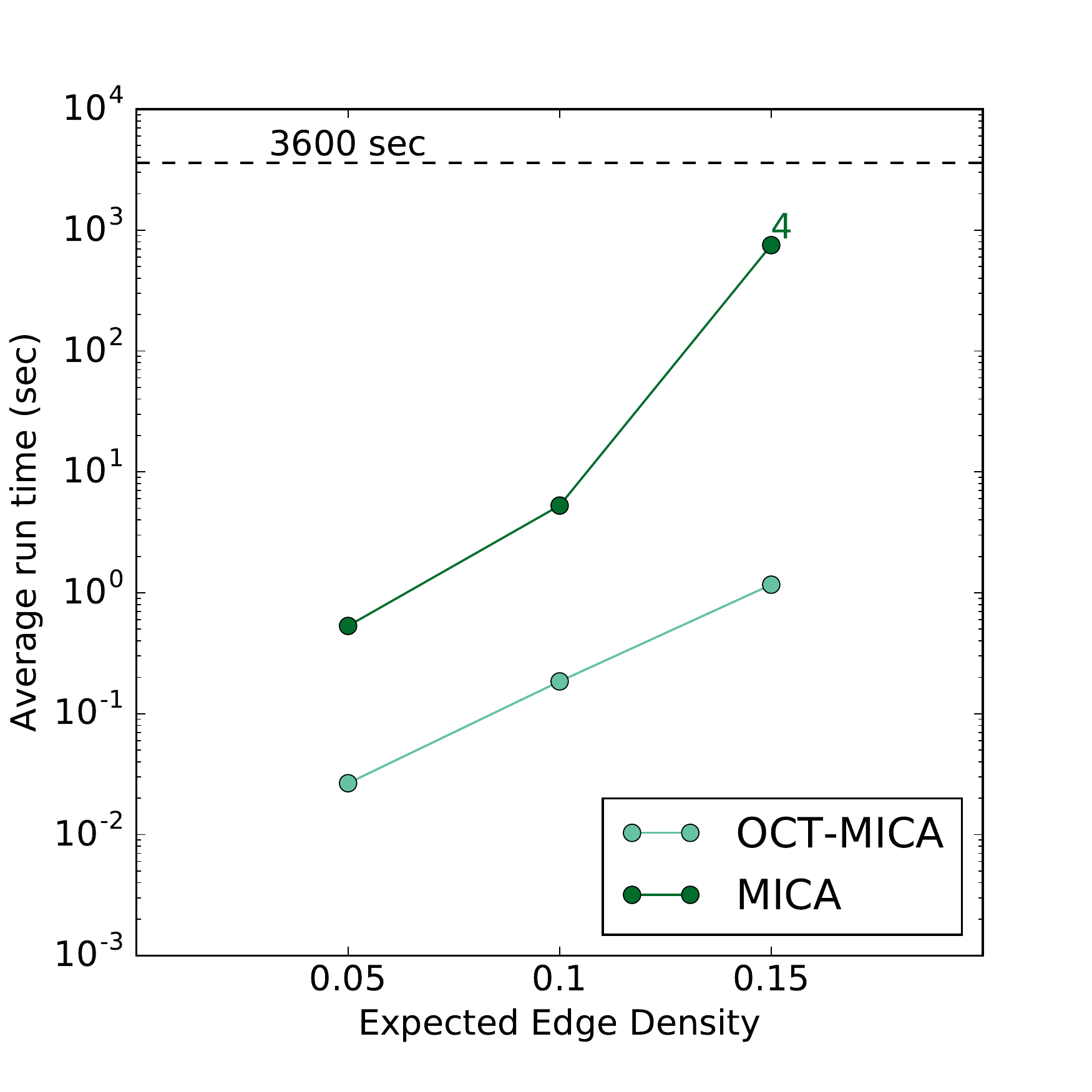}%
    \caption{\label{fig:app17}
	Runtimes of the MIB-enumerating (left) and MB-enumerating (right) algorithms on graphs where $n_B = 300$, $n_L = n_R$ and $n_O = 5$. The expected edge density in the graph was varied, including the expected edge density within $O$.
    }
\end{figure*}

\begin{table}

{\small
    \begin{tabularx}{\textwidth}{@{\extracolsep{\fill}} cccc|cccc|ccc}
        \toprule
        $G$ & $n_B$ & $m$ & $n_O$ & $|M_I|$ & \iOCTMIB & \OCTMIB & \nonlex & $|M_B|$ & \OCTMICA & \MICA \\
        \midrule
\textsf{aa-10} & 69 & 191 & 6 & 178 & 0.008 & 0.023 & 0.057 & 98 & 0.007 & 0.031 \\
\textsf{aa-11} & 102 & 307 & 11 & 424 & 0.055 & 0.115 & 0.259 & 206 & 0.018 & 0.120 \\
\textsf{aa-13} & 129 & 383 & 12 & 523 & 0.083 & 0.239 & 0.470 & 269 & 0.028 & 0.166 \\
\textsf{aa-14} & 125 & 525 & 19 & 1460 & 0.366 & 0.902 & 1.254 & 605 & 0.090 & 0.485 \\
\textsf{aa-15} & 66 & 179 & 7 & 206 & 0.010 & 0.019 & 0.053 & 113 & 0.011 & 0.030 \\
\textsf{aa-16} & 13 & 15 & 0 & 15 & 0.000 & 0.000 & 0.000 & 8 & 0.000 & 0.000 \\
\textsf{aa-17} & 151 & 633 & 25 & 2252 & 1.023 & 2.132 & 3.457 & 1250 & 0.242 & 1.137 \\
\textsf{aa-18} & 87 & 381 & 14 & 660 & 0.100 & 0.173 & 0.389 & 823 & 0.090 & 0.351 \\
\textsf{aa-19} & 191 & 645 & 19 & 1262 & 0.449 & 1.569 & 2.385 & 519 & 0.069 & 0.450 \\
\textsf{aa-20} & 224 & 766 & 19 & 1607 & 0.705 & 2.431 & 3.809 & 949 & 0.154 & 1.061 \\
\textsf{aa-21} & 28 & 90 & 9 & 116 & 0.006 & 0.013 & 0.008 & 213 & 0.019 & 0.030 \\
\textsf{aa-22} & 167 & 641 & 16 & 1520 & 0.423 & 1.387 & 2.629 & 560 & 0.074 & 0.638 \\
\textsf{aa-23} & 139 & 508 & 18 & 1766 & 0.435 & 0.788 & 1.651 & 1530 & 0.210 & 1.000 \\
\textsf{aa-24} & 258 & 1108 & 21 & 3890 & 2.108 & 9.140 & 14.167 & 1334 & 0.237 & 2.477 \\
\textsf{aa-25} & 14 & 15 & 1 & 10 & 0.000 & 0.001 & 0.001 & 10 & 0.000 & 0.000 \\
\textsf{aa-26} & 92 & 284 & 13 & 583 & 0.084 & 0.186 & 0.309 & 370 & 0.030 & 0.128 \\
\textsf{aa-27} & 118 & 331 & 11 & 458 & 0.054 & 0.270 & 0.343 & 229 & 0.015 & 0.114 \\
\textsf{aa-28} & 167 & 854 & 27 & 2606 & 1.464 & 2.201 & 4.162 & 2814 & 0.755 & 3.250 \\
\textsf{aa-29} & 276 & 1058 & 21 & 3122 & 1.909 & 8.418 & 10.707 & 1924 & 0.382 & 3.344 \\
\textsf{aa-30} & 39 & 71 & 4 & 56 & 0.002 & 0.007 & 0.006 & 36 & 0.002 & 0.007 \\
\textsf{aa-31} & 30 & 51 & 2 & 37 & 0.002 & 0.002 & 0.002 & 22 & 0.001 & 0.002 \\
\textsf{aa-32} & 143 & 750 & 30 & 4167 & 2.286 & 7.694 & 5.290 & 3154 & 0.684 & 2.635 \\
\textsf{aa-33} & 193 & 493 & 4 & 578 & 0.046 & 0.204 & 0.993 & 218 & 0.012 & 0.218 \\
\textsf{aa-34} & 133 & 451 & 13 & 705 & 0.132 & 0.316 & 0.756 & 275 & 0.031 & 0.226 \\
\textsf{aa-35} & 82 & 269 & 10 & 459 & 0.037 & 0.108 & 0.178 & 215 & 0.019 & 0.081 \\
\textsf{aa-36} & 111 & 316 & 7 & 248 & 0.015 & 0.076 & 0.155 & 143 & 0.011 & 0.078 \\
\textsf{aa-37} & 72 & 170 & 5 & 135 & 0.005 & 0.018 & 0.054 & 82 & 0.005 & 0.022 \\
\textsf{aa-38} & 171 & 862 & 26 & 4270 & 2.428 & 5.223 & 7.586 & 4964 & 1.136 & 5.179 \\
\textsf{aa-39} & 144 & 692 & 23 & 2153 & 0.872 & 1.574 & 3.034 & 1177 & 0.237 & 1.009 \\
\textsf{aa-40} & 136 & 620 & 22 & 2727 & 1.022 & 2.086 & 2.973 & 1911 & 0.301 & 1.324 \\
\textsf{aa-41} & 296 & 1620 & 40 & 11705 & 16.519 & 82.439 & 50.205 & 20375 & 9.059 & 47.789 \\
\textsf{aa-42} & 236 & 1110 & 30 & 6967 & 5.646 & 45.560 & 21.244 & 8952 & 2.428 & 13.479 \\
\textsf{aa-43} & 63 & 308 & 18 & 905 & 0.137 & 0.294 & 0.311 & 875 & 0.116 & 0.302 \\
\textsf{aa-44} & 59 & 163 & 10 & 211 & 0.014 & 0.024 & 0.051 & 158 & 0.008 & 0.037 \\
\textsf{aa-45} & 80 & 386 & 20 & 1768 & 0.336 & 0.775 & 0.859 & 1716 & 0.244 & 0.796 \\
\textsf{aa-46} & 161 & 529 & 13 & 719 & 0.157 & 0.438 & 0.922 & 374 & 0.036 & 0.257 \\
\textsf{aa-47} & 62 & 229 & 14 & 572 & 0.057 & 0.082 & 0.138 & 451 & 0.051 & 0.127 \\
\textsf{aa-48} & 89 & 343 & 17 & 896 & 0.144 & 0.338 & 0.497 & 519 & 0.060 & 0.230 \\
\textsf{aa-49} & 26 & 62 & 5 & 50 & 0.004 & 0.002 & 0.003 & 74 & 0.006 & 0.013 \\
\textsf{aa-50} & 113 & 468 & 18 & 1272 & 0.322 & 0.778 & 1.098 & 1074 & 0.132 & 0.612 \\
\textsf{aa-51} & 78 & 274 & 11 & 429 & 0.035 & 0.082 & 0.174 & 250 & 0.020 & 0.078 \\
\textsf{aa-52} & 65 & 231 & 14 & 690 & 0.073 & 0.135 & 0.200 & 431 & 0.040 & 0.122 \\
\textsf{aa-53} & 88 & 232 & 12 & 340 & 0.036 & 0.186 & 0.162 & 199 & 0.011 & 0.052 \\
\textsf{aa-54} & 89 & 233 & 12 & 286 & 0.027 & 0.063 & 0.113 & 177 & 0.015 & 0.039 \\
        \bottomrule
    \end{tabularx}
}
    \caption{\label{tab:WH-all-aa}
        The runtimes (rounded to nearest thousandth-of-a-second) of the biclique-enumeration algorithms on the Afro-American subset of the Wernicke-H{\"u}ffner computational biology data~\cite{WERNICKE}.
    }
\end{table}

\begin{table}

{\small
    \begin{tabularx}{\textwidth}{@{\extracolsep{\fill}} cccc|cccc|ccc}
        \toprule
        $G$ & $n_B$ & $m$ & $n_O$ & $|M_I|$ & \iOCTMIB & \OCTMIB & \nonlex & $|M_B|$ & \OCTMICA & \MICA \\
        \midrule
\textsf{j-10} & 55 & 117 & 3 & 52 & 0.002 & 0.009 & 0.010 & 39 & 0.001 & 0.010 \\
\textsf{j-11} & 51 & 212 & 5 & 63 & 0.003 & 0.014 & 0.011 & 36 & 0.003 & 0.012 \\
\textsf{j-13} & 78 & 210 & 6 & 224 & 0.015 & 0.028 & 0.074 & 90 & 0.009 & 0.032 \\
\textsf{j-14} & 60 & 107 & 4 & 44 & 0.004 & 0.007 & 0.003 & 38 & 0.003 & 0.003 \\
\textsf{j-15} & 44 & 55 & 1 & 13 & 0.001 & 0.000 & 0.004 & 10 & 0.001 & 0.000 \\
\textsf{j-16} & 9 & 10 & 0 & 10 & 0.000 & 0.000 & 0.000 & 3 & 0.000 & 0.000 \\
\textsf{j-17} & 79 & 322 & 10 & 317 & 0.025 & 0.051 & 0.127 & 126 & 0.014 & 0.056 \\
\textsf{j-18} & 71 & 296 & 9 & 154 & 0.011 & 0.038 & 0.053 & 91 & 0.012 & 0.028 \\
\textsf{j-19} & 84 & 172 & 3 & 105 & 0.002 & 0.010 & 0.019 & 46 & 0.002 & 0.013 \\
\textsf{j-20} & 241 & 640 & 1 & 274 & 0.013 & 0.065 & 0.484 & 228 & 0.009 & 0.188 \\
\textsf{j-21} & 33 & 102 & 9 & 107 & 0.006 & 0.012 & 0.008 & 197 & 0.017 & 0.024 \\
\textsf{j-22} & 75 & 391 & 9 & 221 & 0.020 & 0.051 & 0.080 & 113 & 0.009 & 0.048 \\
\textsf{j-23} & 76 & 369 & 19 & 682 & 0.095 & 0.404 & 0.217 & 459 & 0.057 & 0.132 \\
\textsf{j-24} & 142 & 387 & 4 & 150 & 0.013 & 0.027 & 0.089 & 104 & 0.007 & 0.025 \\
\textsf{j-25} & 14 & 14 & 0 & 14 & 0.000 & 0.000 & 0.000 & 3 & 0.000 & 0.000 \\
\textsf{j-26} & 63 & 156 & 6 & 156 & 0.007 & 0.019 & 0.035 & 67 & 0.003 & 0.013 \\
\textsf{j-28} & 90 & 567 & 13 & 492 & 0.073 & 0.130 & 0.244 & 416 & 0.044 & 0.193 \\
        \bottomrule
    \end{tabularx}
}
    \caption{\label{tab:WH-all-j}
        The runtimes (rounded to nearest thousandth-of-a-second) of the biclique-enumeration algorithms on the Japanese subset of the Wernicke-H{\"u}ffner computational biology data~\cite{WERNICKE}.
    }
\end{table}